\documentclass[12pt,letterpaper]{article}
\input{glyphtounicode}
\usepackage[T1]{fontenc}
\usepackage{lmodern}
\usepackage[utf8]{inputenc}
\usepackage{geometry}
\usepackage{booktabs}
\usepackage{amsmath}
\usepackage{amsthm} 
\usepackage{amsfonts}
\usepackage{dsfont}
\usepackage{cancel}
\usepackage{mathtools}
\usepackage{datetime}
\usepackage{bm}
\usepackage{enumitem}
\setlist{noitemsep, topsep=2pt}
\usepackage{changepage}
\usepackage{amsthm}
\usepackage{verbatim}
\usepackage{amsmath,amssymb}
\usepackage{graphicx}
\usepackage{emptypage}
\usepackage{mathabx}
\usepackage{newlfont}
\usepackage{tikz-cd}
\usepackage{tikz}
\usepackage[all,cmtip]{xy}
\usepackage[bottom]{footmisc}
\usepackage[titletoc,title]{appendix}
\usepackage{chngcntr}
\usepackage{apptools}
\usepackage{listings}
\usepackage{color}
\usepackage{sgame, tikz} 

\usepackage{kpfonts}  
\usepackage{natbib}
\AtAppendix{\counterwithin{lem}{section}}
\usepackage{caption}
\usepackage{subcaption}
\usepackage{float}
\usepackage{thmtools,thm-restate}
\usepackage{epigraph}
\usepackage{xr-hyper}
\usepackage{hyperref}
\usepackage{wasysym}
\hypersetup{colorlinks,linkcolor=blue,allcolors=blue}
\usepackage{sgamevar}
\hypersetup{urlcolor=black}

\theoremstyle{plain} 
\newtheorem{cor}{Corollary} 
\newtheorem{prop}{Proposition}
 
\newtheorem{theorem}{Theorem}
\newtheorem{lemma}{Lemma}
\theoremstyle{definition} 
\newtheorem{ex}{Example}

\makeatletter
\def\th@remark{
  \thm@headfont{\bfseries}
  \normalfont 
  \thm@preskip\topsep \divide\thm@preskip\tw@
  \thm@postskip\thm@preskip
}
\makeatother
\allowdisplaybreaks
\theoremstyle{remark}

\let\emptyset\varnothing

\DeclareMathOperator{\E}{\mathds{E}}
\renewcommand{\P}{\mathds{P}}

\newcommand{\R}{\mathds{R}}

\newcommand{\N}{\mathds{N}}

\newcommand{\A}{\mathcal{A}}

\newcommand{\Y}{\mathcal{Y}}

\renewcommand{\1}{\mathds{1}}

\newcommand\eqid{\stackrel{d}{=}}

\renewcommand{\epsilon}{\varepsilon}

\renewcommand*\d{\mathop{}\!\mathrm{d}}
\newcommand{\argmax}{\operatornamewithlimits{argmax}}

\newcommand{\X}{\mathcal{X}}

\usepackage{setspace}

\theoremstyle{definition}

\reversemarginpar
\makeatletter
\g@addto@macro\normalsize{%
  \setlength\abovedisplayskip{6pt plus 2pt minus 4pt}%
  \setlength\belowdisplayskip{6pt plus 2pt minus 4pt}%
  \setlength\abovedisplayshortskip{0pt plus 2pt}%
  \setlength\belowdisplayshortskip{4pt plus 2pt minus 2pt}%
}
\makeatother

\newdate{draftdate}{02}{10}{2023}
\usdate

\usepackage[nameinlink]{cleveref}
\crefname{manualasm}{assumption}{assumptions}
\crefname{cor}{corollary}{corollaries}
\crefalias{prop}{proposition}
\crefname{claim}{claim}{claims}
\crefname{ex}{example}{examples}
\crefname{defn}{definition}{definitions}
\crefname{rmk}{remark}{remarks}
\crefname{alg}{algorithm}{algorithms}

\begin{document}
\title{\textbf{Screening Frontiers}\thanks{I thank Ben Brooks, Piotr Dworczak, Jeff Ely, Nima Haghpanah, Jacob Leshno, Andrew McClellan, Meg Meyer, Lea Nagel, Mariann Ollár, Daniel Rappoport, Doron Ravid, Roberto Saitto, Alex Teytelboym, Udayan Vaidya, as well as seminar and conference audiences at University of Chicago, University of Michigan, Rice Theory Mini-Conference, Brown University, University of Mannheim, EC'26, and ESSET 2026 for helpful comments and discussions. Some results in this paper were also reported in my job market paper (\citealt{yang2023workingpaper}), which has since been split into separate papers.}
}
\author{Frank Yang\thanks{Department of Economics, Stanford University. Email: \href{mailto:shuny@stanford.edu}{\color{blue}shuny@stanford.edu}.}}
\date{\today\\
}
\maketitle
\begin{abstract}
A principal screens an agent with an arbitrary set of allocations $\mathcal{X}$. The agent's values for the allocations are comonotonic. A subset $\mathcal{X}^\star \subseteq \mathcal{X}$ is a \textit{\textbf{surplus-elasticity frontier}} if \textit{(i)} any other allocation has a demand curve that is pointwise lower and less elastic than that of some allocation in $\mathcal{X}^\star$ and \textit{(ii)} the allocations in $\mathcal{X}^\star$ can be ordered in terms of their demand curves such that a higher demand curve is more inelastic. We show that any surplus-elasticity frontier is an optimal menu. Moreover, if the incremental demand curves along the frontier are also ordered by their elasticities, then the frontier remains optimal even if randomization is allowed. The frontier is agnostic to type distributions and redistributive welfare weights---the same menu remains optimal for a broad class of objectives, yielding a decomposition between menu design and optimal pricing. Under the same incremental-elasticity condition, a generalized notion of the frontier is not only sufficient but also necessary for this robust optimality. As applications, we derive new results on optimal bundling, taxation, sequential screening, selling information, and regulating a data-rich monopolist. \\


\end{abstract}
\setcounter{page}{1}
\newpage

\addtocontents{toc}{\protect\setcounter{tocdepth}{2}} 
\newpage

\section{Introduction}

Products are increasingly feature-rich. How should a multiproduct firm design its product mix? In this paper, we answer this question by considering a screening model with an arbitrary allocation space $\mathcal{X}$. The only substantive assumption we make is that consumers can be ordered so that higher types value every allocation more.\footnote{When $\mathcal{X}$ is finite, this is equivalent to assuming that the distribution of value types in $\R^{\mathcal{X}}$ is a \textit{\textbf{comonotonic}} joint distribution. It restricts consumer preferences \textit{across} types but imposes no restrictions on preferences over allocations for any \textit{fixed} type. \Cref{app:multidim} shows that exact comonotonicity is not necessary: our results continue to hold approximately when consumer preferences are sufficiently positively correlated.}

In this environment, we argue that there are two key dimensions that the firm should optimize for: \textit{(i)} the level of surplus and \textit{(ii)} the elasticity of demand.\footnote{For expositional simplicity, we assume zero marginal costs; when there are positive costs, we can incorporate them by redefining the demand curves as net of production costs, with price elasticities defined with respect to the cost-adjusted demand curve.} Specifically, for each allocation, consider the \textit{\textbf{demand curve}} generated when the product is sold alone.  We call a menu of products $\mathcal{X}^\star \subseteq \mathcal{X}$ a \textit{\textbf{surplus-elasticity frontier}} if:
\begin{itemize}
    \item[\textit{(i)}] any other allocation has a demand curve that is pointwise lower and less elastic than that of some allocation in $\mathcal{X}^\star$; 
    \item[\textit{(ii)}] the allocations in $\mathcal{X}^\star$ can be ordered in terms of their demand curves such that a higher demand curve is also more inelastic. 
\end{itemize}
 \Cref{fig:frontier} illustrates this definition in the case where any two demand curves can be pointwise compared both in levels and in elasticities. In such a case, a surplus-elasticity frontier is just like a Pareto frontier with the two optimizing dimensions being the level of demand curves and the elasticity of demand curves. More generally, because each of the two comparisons is only a partial order, a surplus-elasticity frontier need not always exist. 

However, whenever such a frontier exists, we show that it is an optimal menu (\Cref{thm:main}). Moreover, if the \textit{\textbf{incremental demand curves}} along the frontier---the differences of consecutive demand curves---are also ordered by their price elasticities, then the frontier is optimal even among \textit{\textbf{stochastic}} mechanisms (\Cref{thm:main2}). As we show, for profit maximization, the optimal prices can then be determined by the demand profile method that simply prices the upgrades according to the incremental demand curves, and, under mild assumptions, every product along the frontier must be used (\Cref{cor:profile}). Importantly, a generalization of our frontier concept is not only sufficient but also \textit{\textbf{necessary}} for optimality, in a sense made precise below (\Cref{thm:converse}).

Conceptually, our results are very intuitive. The level of demand determines the \textit{\textbf{surplus}} available in a product in the absence of any private information. The elasticity of demand, as we show, is equivalent to a precise notion of \textit{\textbf{dispersion}} in the values for a given product on the log scale, and hence dictates the information rents to the agent per unit of surplus. In this sense, the surplus-elasticity frontier balances creating surplus and controlling information rents. 

Technically, our results are ``dimension-reduction'' results---as we show, the agents' preferences over the allocations on the frontier always satisfy the Spence--Mirrlees condition and hence the problem of pricing the frontier allocations can always be solved via standard one-dimensional methods. Our results reduce the original arbitrary allocation space $\mathcal{X}$ to a subset of allocations $\mathcal{X}^\star$ for which the problem becomes tractable. As we discuss, these abstract results accommodate objectives beyond profit maximization, such as weighted welfare maximization, and immediately lead to a variety of new results in important screening problems such as optimal bundling, optimal taxation, sequential screening, selling information, and monopoly regulation.

\begin{figure}[t]
    \centering
    \includegraphics[width=0.44\linewidth]{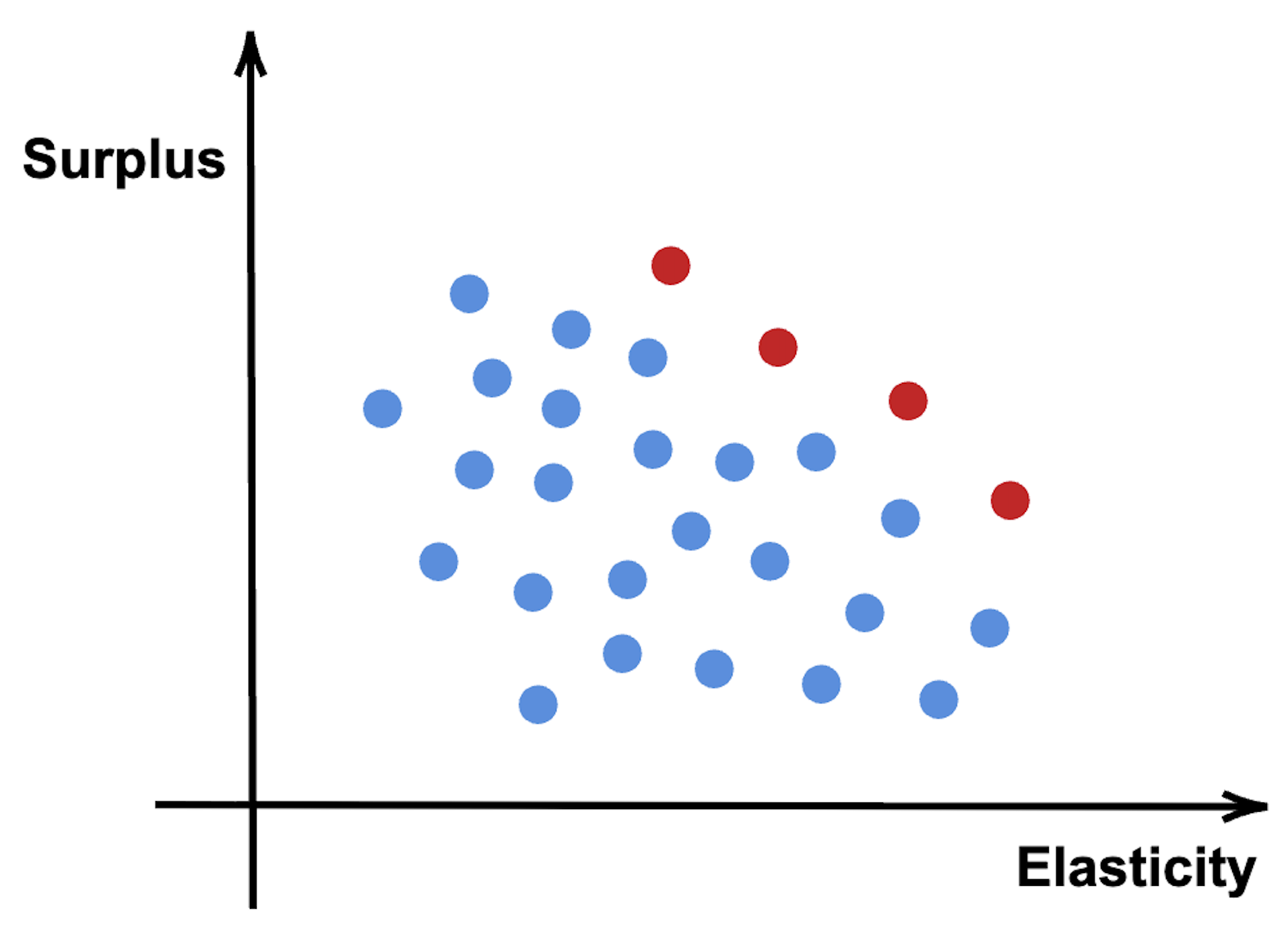}
    \caption{Illustration of Surplus-Elasticity Frontier. Each point is a product, identified by its sold-alone demand curve. The figure illustrates a case where any two demand curves can be pointwise compared in terms of their levels and their elasticities, in which case a surplus-elasticity frontier always exists (highlighted in red) and is optimal by \Cref{thm:main}. }
    \label{fig:frontier}
\end{figure}

To illustrate the results, we start with two warm-up examples. 

\begin{ex}\label{ex:one-good}
Consider the classic problem of selling a single good while allowing for lotteries. The allocation space is $\mathcal{X} = [0,1]$, where $x$ is the probability of allocating the good. The demand curve induced by any lottery $x \in \mathcal{X}$ is simply a \textit{\textbf{linear scaling}} of the full-allocation demand curve, which lowers its level but does not affect its elasticity. Thus, the surplus-elasticity frontier consists only of the degenerate lottery that allocates the good for sure, i.e., $\X^\star = \{1\}$. By \Cref{thm:main}, a posted price is then optimal, recovering the well-known result of \citet{Myerson1981} and \citet{riley1983optimal}.
\qed
\end{ex}

\begin{figure}[t]
    \centering
    \includegraphics[width=0.44\linewidth]{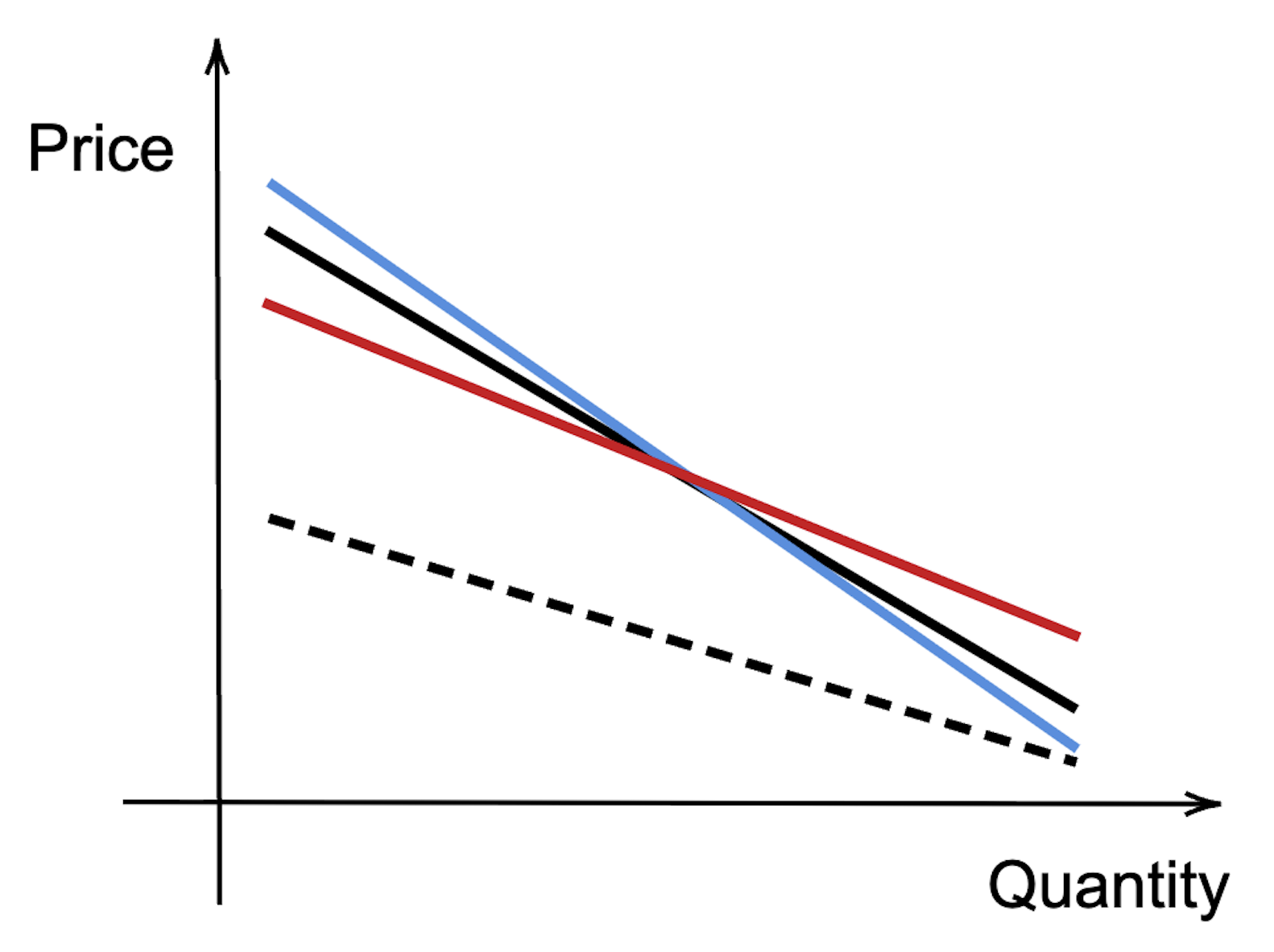}
    \caption{Illustration of Changes in Demand Curves. The dashed black line is the original demand curve. The solid lines are three demand curves that increase surplus but have different elasticity effects: Solid black line keeps the same elasticities (a \textit{scaling}); solid red line increases elasticities (a \textit{flattening}); solid blue line decreases elasticities (a \textit{steepening}).}
    \label{fig:demandcurves}
\end{figure}

\begin{ex}\label{ex:two-dimension}
Now, as another illustration, consider a monopolist who offers a product mix. Each product $x$ has two different dimensions of features, say \textit{\textbf{functionality}} and \textit{\textbf{design}}, $(\alpha, \beta) \in \mathcal{X} = [0, 1]^2$. For any consumer of type $t\in[1, 2]$, their value is given by $v\big((\alpha, \beta), t\big) = (\alpha + t)^{\beta}$. The monopolist wants to select a menu of these products to offer. Note that an increase in either feature increases the level of demand, but they have opposite effects on the elasticity of demand---an increase in the functionality $\alpha$ \textit{\textbf{flattens}} the demand curve (i.e., makes it more elastic) whereas an increase in the design $\beta$ \textit{\textbf{steepens}} the demand curve (i.e., makes it more inelastic).\footnote{Indeed, the demand curve for each product is simply $P(x, q) = (\alpha + F^{-1}(1-q))^{\beta}$, where $F$ is the distribution of $t$, and hence increasing $\beta$ rotates the demand curve \textit{clockwise}, whereas increasing $\alpha$ rotates the demand curve \textit{counterclockwise}. 
} \Cref{fig:demandcurves} illustrates. Thus, the surplus-elasticity frontier is $\X^\star = \{(1, \beta)\}_{\beta\in[0,1]}$. By \Cref{thm:main}, the monopolist should offer a menu of products that share full functionality but differ in the details of the design. Moreover, the incremental demand curves along this frontier also become more inelastic as the design $\beta$ increases. Therefore, by \Cref{thm:main2}, the menu $\X^\star = \{(1, \beta)\}_{\beta\in[0,1]}$ is optimal even if we allow for stochastic mechanisms, i.e., with the allocation space being $\Delta([0, 1]^2)$.  \qed
\end{ex}

In both \Cref{ex:one-good} and \Cref{ex:two-dimension}, the optimal prices depend on the distribution of consumer types, but the same menu $\X^\star$ (priced differently) is optimal for any type distribution. This turns out not to be a coincidence. As we discuss, the definition of a surplus-elasticity frontier is actually \textit{\textbf{agnostic}} to the type distribution. As a consequence, just like in the one-good case of \citet{Myerson1981}, our results provide a clean decomposition between menu design and optimal pricing: \textit{(i)} once the designer can identify the frontier from the demand system, the designer can already assemble an optimal menu of products just like in \Cref{ex:one-good} and \Cref{ex:two-dimension}, and \textit{(ii)} as the designer obtains more information about the type distribution (e.g., via price experimentation), the designer can further fine-tune the prices of the options in the menu. 

Perhaps also surprisingly, the surplus-elasticity frontier remains optimal even if the principal does not just care about revenue. All of our results turn out to hold identically even if the designer has redistributive preferences such as in optimal taxation or redistributive allocation problems, as long as the welfare weights are \textit{\textbf{redistributive}} in the sense that the designer weakly prefers transferring one dollar from a high type to a low type. Whenever a surplus-elasticity frontier exists, the \textit{\textbf{same}} frontier is optimal simultaneously for a broad set of objectives. For instance, in \Cref{ex:two-dimension}, the same menu $\X^\star = \{(1, \beta)\}_{\beta\in[0,1]}$ remains optimal even if the designer has redistributive preferences (regardless of the specific welfare weights) as well as access to stochastic mechanisms.

\begin{figure}[t]
    \centering
    \includegraphics[width=0.44\linewidth]{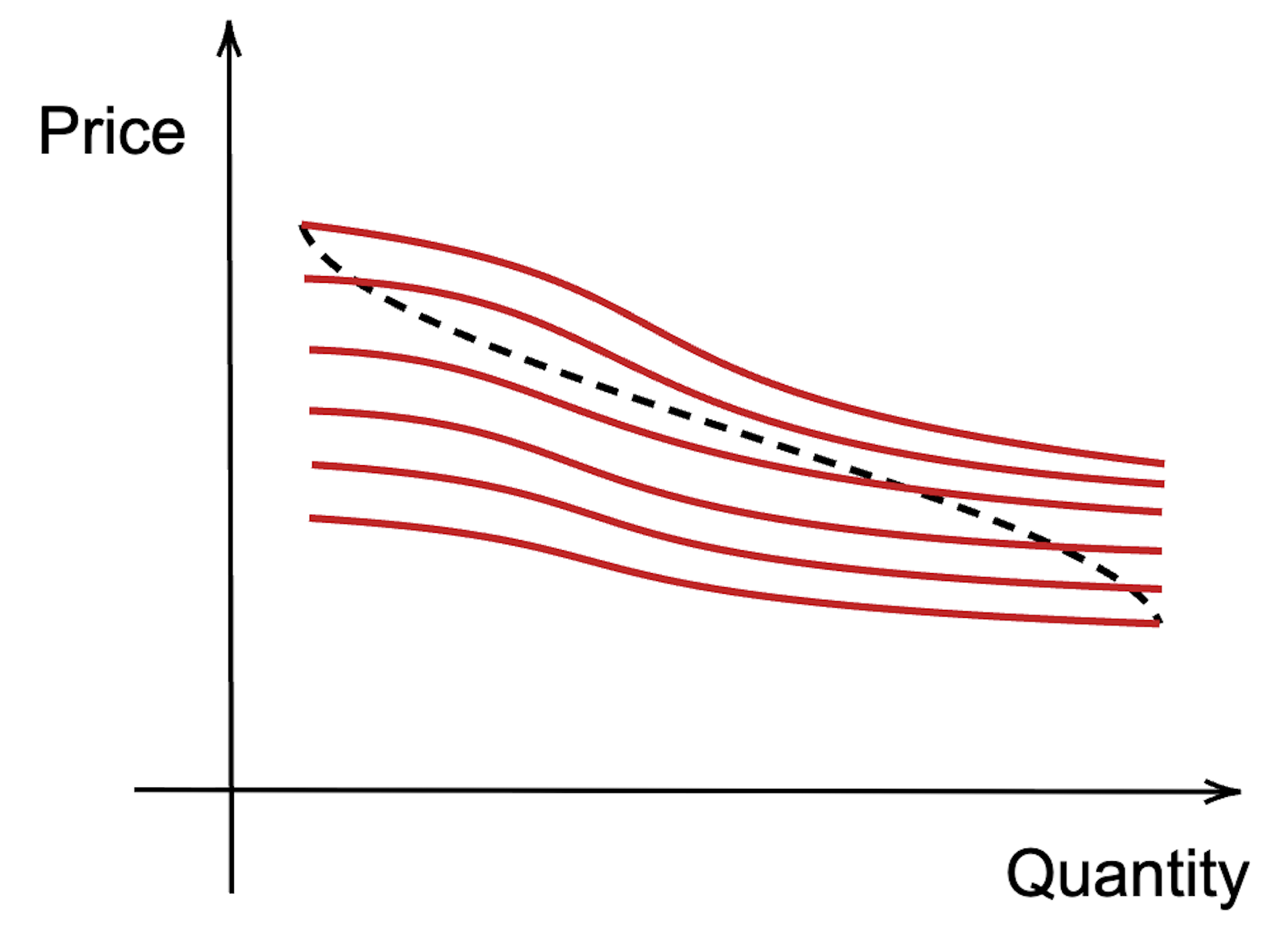}
    \caption{Illustration of Covering of a Demand Curve. The dashed black curve is the original demand curve. The collection of solid red demand curves \textit{covers} the original one by single-crossing it from below individually and covering it collectively. }
    \label{fig:covering}
\end{figure}

Both requirements of a surplus-elasticity frontier are needed for our results---in particular, as we show, without the second requirement (the ``anti-chain'' property), using only the ``undominated'' allocations can be strictly suboptimal (\Cref{ex:counterexample}). The two requirements, however, can be weakened to a generalized notion, which we show is not only sufficient but, under suitable conditions, necessary for such \textit{\textbf{robust optimality}}. We generalize by weakening the pointwise elasticity comparisons to single-crossing comparisons and allowing a collection of possibly randomized demand curves to ``dominate'' a given one. Say that a collection of randomized demand curves \textit{\textbf{covers}} a given demand curve if they single-cross the given one from below individually and cover its graph collectively. \Cref{fig:covering} illustrates. A menu of allocations $\X^\star$ is a \textit{\textbf{generalized frontier}} if its elements can be ordered, with an accompanying stochastically ordered set of lotteries $\A \subseteq \Delta(\X^\star)$, such that \textit{(i)} the demand curve of every allocation $x \in \X$ can be covered by a collection of demand curves from $\A$ and \textit{(ii)} the demand curves of the elements in $\X^\star$ satisfy increasing differences. Our proof of \Cref{thm:main} actually shows that any surplus-elasticity frontier is a generalized frontier and that any generalized frontier is optimal---thus, it suffices to find a generalized frontier (\Cref{thm:main3}). The conclusion of \Cref{thm:main2} also applies: if the incremental demand curves are ordered by elasticities along the generalized frontier, then it is optimal even among stochastic mechanisms. Importantly, under the same incremental-elasticity condition, a generalized frontier is not only sufficient but also necessary for profit maximization under every type distribution, and therefore \textit{\textbf{characterizes}} robust optimality (\Cref{thm:converse}).

Beyond the conceptual insights, our main results deliver robust conditions for ``when simple mechanisms are optimal'' across various applications. We now describe the applications in turn.

\paragraph{Optimal Bundling.}\hspace{-2mm}Consider a designer who allocates bundles of a finite set of goods, where the consumers are vertically ordered and can have non-additive values. An immediate application of \Cref{thm:main} and \Cref{thm:main2} identifies a new condition (\textit{\textbf{robust nesting condition}}) for when \textit{\textbf{nested bundling}}---where more expensive bundles include all the goods in the less expensive ones---is optimal (\Cref{prop:nesting}). As our result reveals, the key logic behind nested bundling is that the seller successively chooses smaller but more elastic bundles to span the surplus-elasticity frontier. Unlike the demand conditions in \citet{yang2023nested}, which depend on the type distribution and hold only for profit maximization, the robust nesting condition yields a clean decomposition between menu design and optimal pricing: once the seller can identify the frontier from the demand system, she can already assemble the optimal menu, regardless of the specific type distribution. Moreover, since the results hold beyond profit maximization, they open the door to applying results from the bundling literature to redistributive allocation problems. As a special case, we show that the \textit{\textbf{ratio-monotonicity condition}} of \citet{haghpanah2021pure} actually implies that pure bundling is optimal even if the designer is a social planner who has \textit{\textbf{redistributive}} preferences that favor agents who have lower values for the grand bundle; moreover, since this holds for all redistributive welfare weights, it enables us to perform comparative statics on the optimal mechanism (\Cref{cor:pure}). Perhaps surprisingly, the designer never introduces any smaller bundle for the poorer agents to consume, even though she cares about them more---under ratio-monotonicity, the grand bundle itself is the surplus-elasticity frontier and hence most effective at controlling richer agents' information rents per unit of surplus. The only instrument is thus the grand-bundle price, which the optimal mechanism \textit{\textbf{lowers}} if the welfare weights are pointwise increased but \textit{\textbf{raises}} if the welfare weights shift toward poorer agents in the majorization order.

\paragraph{Optimal Taxation and Ordeals.}\hspace{-2mm}Consider the optimal nonlinear taxation problem with quasilinear preferences (\citealt{diamond1998optimal}). Suppose the government has the ability to combine the \textit{\textbf{tax instruments}} with \textit{\textbf{costly screening}} to induce self-targeting in transfers. Would it be beneficial to do so? \citet{Nichols1982b} show that it is indeed possible for the government to use socially wasteful ordeals to improve the net welfare because the ordeals can help with screening. Applying our main results, we show that there are two key forces that determine whether an ordeal should be used in the presence of tax instruments (\Cref{prop:ordeal}): \textit{(i)} As anticipated by \citet{Nichols1982b}, ordeals cannot be helpful if the higher ability types have lower costs from performing the ordeals; \textit{(ii)} even if the direction of sorting is correct, ordeals still cannot be helpful if an appropriately defined \textit{\textbf{take-up elasticity}} with respect to a small tax cut becomes lower when an ordeal is imposed. Intuitively, the reduction in surplus by an ordeal must be compensated via an increase in the screening effectiveness, captured by the elasticity.

\paragraph{Sequential Screening.}\hspace{-2mm}Consider the sequential screening problem introduced in \citet{Courty2000}. A seller contracts with a buyer on the sale of a good after the buyer has some initial estimate $t$ of their value $\theta$ but before fully learning $\theta$. This is a canonical example of \textit{\textbf{dynamic mechanism design}}. In general, the optimal mechanism can be very sophisticated, involving a continuum of options with different up-front fees for different refund options (\citealt{Courty2000}). In practice, however, even in many markets with buyer uncertainty, the selling mechanisms tend to be simple, sometimes simply a single \textit{\textbf{nonrefundable posted price}}. When is such a simple mechanism optimal in sequential screening? Applying our main results, we show that a nonrefundable posted price is optimal if higher ex ante types not only have higher expected willingness to pay (in the sense of FOSD) but also have greater uncertainty about their ex post values in the sense of the \textit{\textbf{dispersive order on the log scale}} (\Cref{prop:sequential}). To see why, note that a refund option with an up-front fee actually \textit{\textbf{reduces}} the ex post surplus by inducing agents with lower ex post values to forgo consumption; by our main results, it can be justified only by increasing the screening effectiveness. Yet under the log-dispersion order---which captures the elasticity of how ex post values vary with ex ante types---the screening effectiveness cannot be increased; hence, a nonrefundable posted price is optimal. Even though sequential screening is a well-studied problem, we are not aware of any primitive condition under which a nonrefundable posted price is optimal. One reason is methodological: much of the dynamic mechanism design literature relies on the Myersonian approach that relaxes the global integral monotonicity constraint (\citealt*{Pavan2014}), and thus mechanisms with significant bunching, such as a nonrefundable posted price, can be difficult to obtain without explicitly tracking global incentive constraints. As we discuss, our proof method is non-Myersonian and deals with global incentive constraints.

\paragraph{Selling Information.}\hspace{-2mm}Consider a designer selling \textit{\textbf{Blackwell experiments}} to an agent who has a decision problem to solve, just like in \citet*{bergemann2018design}. However, unlike \citet*{bergemann2018design}, suppose that the agent's prior about the state is common knowledge but the decision problem is their private information. Recall that every decision problem can be represented as a \textit{\textbf{convex indirect utility function}} of the posterior beliefs; each agent of type $t$ is thus identified with a convex indirect utility function $u(\mu, t)$, with the types vertically ordered in that higher types value every nontrivial experiment more relative to no information. A class of signals commonly used for tractability is called \textit{\textbf{truth-or-noise signals}} (\citealt{lewis1994supplying})---such a signal reports the true state with some probability and otherwise is an independent draw from the prior. When are they optimal? Applying our main results, we show that in any symmetric environment with a uniform prior, if the agent's indirect utility function $u(\,\cdot\,, t)$ is \textit{\textbf{increasingly convex}} as the type increases (in the Arrow--Pratt sense)---so that a higher type tends to have decision problems that depend \textit{\textbf{more sensitively}} on the state---then it is optimal, among arbitrary Blackwell experiments, to offer a menu of truth-or-noise signals (\Cref{prop:info}). These signals span a generalized frontier where a higher signal-to-noise ratio creates more surplus but a less elastic demand curve. In the opposite case where $u(\,\cdot\,, t)$ is \textit{\textbf{increasingly concave}}, simply selling \textit{\textbf{full information}} is optimal since its demand curve is most elastic.

\paragraph{Monopoly Regulation.}\hspace{-2mm}Our last application studies the classic problem of monopoly regulation, just like in \citet{Baron1982RegulatingCosts}. However, unlike \citet{Baron1982RegulatingCosts}, consider a monopolist who has \textit{\textbf{rich data}} about the consumers and will engage in personalized pricing in the absence of regulation. Following \citet{strack2025non}, we model the data-rich monopolist as observing the consumers' values $\theta$ and designing any \textit{\textbf{pricing rule}} that specifies a randomized price for each consumer $\theta$. The monopolist privately observes its type $t$ which determines its cost $c(\theta, t)$ for serving any consumer $\theta$; as in \citet{strack2025non}, this allows for interdependent costs, which are very common in settings involving consumer data (e.g., credit markets and insurance markets). Consider a regulator who contracts with the monopolist on the pricing rules: as in \citet{Baron1982RegulatingCosts}, the regulator maximizes a weighted sum of consumer and producer surplus but does not observe the cost type $t$, and therefore posts a menu of pricing rules with associated subsidies or taxes. We say that a mechanism is a \textit{\textbf{uniform pricing regulation}} if the menu consists only of deterministic pricing rules that do not depend on the consumers' types $\theta$---under such a regulation, the determination of taxes/subsidies exactly collapses to the model of \citet{Baron1982RegulatingCosts}. Of course, absent the tax instrument, uniform pricing protects the consumers; with it, are uniform pricing regulations optimal? We show that if the total surplus $\theta - c(\theta, t)$ is \textit{\textbf{log-submodular}}, then uniform pricing regulations are optimal (\Cref{prop:regulation})---consumer data has no regulatory value, collapsing optimal regulation to the classic Baron--Myerson analysis. However, if the total surplus is strictly \textit{\textbf{log-supermodular}}, then uniform pricing regulations generally fail to be optimal; instead, \textit{\textbf{discriminatory pricing regulations}} are optimal---the regulator permits personalized pricing and designs a suitable tax schedule to transfer the surplus to the consumers. As we discuss, our conditions capture the \textit{\textbf{degree of adverse selection}}---in markets with advantageous selection or mild adverse selection, uniform pricing regulations are optimal; with severe adverse selection, optimal regulations involve discriminatory pricing. The log-modularity of the surplus is exactly the elasticity comparison with respect to the monopolist type that delivers a screening frontier in our sense.

\paragraph{Related Literature.}\hspace{-2mm}We view our main contributions as twofold. First, our conceptual contribution is to explicitly identify two intuitive channels, surplus and dispersion, measured by the levels and the elasticities of the sold-alone demand curves that together pin down the optimal menu in a general screening environment. Second, our methodological contribution is a dimension-reduction technique that shows how various screening problems with rich allocation spaces can be characterized by identifying a suitable screening frontier that balances creating surplus and controlling information rents. We build on a large literature on multidimensional screening, most of which focuses on the optimal bundling problem. The model is traditionally set up with additive values and multidimensional types (\citealt*{McAfee1989MultiproductValues,Rochet1998, mcafee1988multidimensional, manelli2006bundling, daskalakis2017strong}) and turns out to be analytically and computationally intractable (\citealt{Rochet2003}; \citealt{daskalakis2014complexity}; \citealt{lahr2024extreme}). A recent line of work relaxes the additivity requirement and explores the consequences of \textit{\textbf{comonotonic type spaces}} where the consumer types are totally ordered (\citealt{haghpanah2021pure}; \citealt{ghili2021characterization}; \citealt{yang2023nested}). Our model builds heavily on this line of literature by fully embracing the one-dimensional type space and fully relaxing the allocation space requirement.\footnote{As we show in \Cref{app:multidim}, such a model also approximates general multidimensional screening models with multidimensional types that have sufficiently high positive correlation (see \Cref{thm:positive-correlation}).} In contemporaneous work, \citet{haghpanah2025screening} also study a model with an arbitrary allocation space but restrict attention to binary types, which allows them to study environments where types are not ordered.  Relatedly, \cite{pernoud2025bundling} characterize optimal mechanisms with non-ordered linear type spaces assuming additive values and provide a microfoundation for comonotonic types using equilibrium learning; \cite*{frick2026multidimensional} study the monopolist's problem with precise consumer data and provide structural results with linear type spaces.  Our proof method builds on and further develops the non-Myersonian approach in \citet{yang2022costly}.\footnote{See \citet{yang2023nested} for a different, Myersonian approach that maximizes a suitably defined virtual value function pointwise, with intuition based on marginal revenue curves. As we explain in \Cref{subsec:bundling}, a key advantage of the present approach is that it yields conditions robust to both type distributions and specific objectives, thereby separating menu design from optimal pricing.} In particular, we generalize the \textit{\textbf{downward sufficiency theorem}} in \citet{yang2022costly} to allow for arbitrary redistributive welfare weights, which, as our applications show, opens the door to applying results from multidimensional screening to recent redistributive allocation problems as well as classic optimal taxation and monopoly regulation problems.\footnote{See also \citet{doligalski2025optimal} for applying multidimensional screening techniques to taxation. There is a substantial literature in public finance studying taxation with multiple instruments (e.g., \citealt*{ferey2024sufficient}) or multidimensional heterogeneity (e.g., \citealt*{golosov2025optimal}).}

The remainder of the paper proceeds as follows. \Cref{sec:model} presents the model. \Cref{sec:main} presents the main results. \Cref{sec:proof} sketches the main proofs. \Cref{sec:application} presents the applications. \Cref{sec:conclude} concludes. \Cref{app:proof} contains the omitted proofs.

\section{Model}\label{sec:model}

A principal screens an agent. The allocation space is a compact metric space $\mathcal{X} \ni x$. The agent has a one-dimensional type $t \in \mathcal{T}:=[\underline{t}, \overline{t}]$, distributed according to a continuously differentiable $F$ with a positive density $f$. The agent has quasilinear preferences, given by $v(x, t) - p$ where $p$ is the payment. The agent has access to an outside option $\emptyset$, with the payoff normalized to $0$. We assume that the agent's value function $v(x, t)$ is continuous on $\X \times \mathcal{T}$, strictly increasing and differentiable in $t$ with $v_t(x, t)$ continuous on $\mathcal{X} \times \mathcal{T}$.  To simplify notation, we augment $\X$ with the outside option $\emptyset$.

The principal wants to maximize a weighted sum of welfare and revenue. By the revelation principle, it is without loss of generality to consider direct-revelation mechanisms. A \textit{\textbf{(direct-revelation) mechanism}} is a measurable map $(x, p): \mathcal{T} \rightarrow \mathcal{X} \times \R$ that satisfies the incentive compatibility (IC) and individual rationality (IR) constraints: 
\begin{align*}
    v(x(t), t) - p(t) &\geq     v(x(\hat{t}), t) - p(\hat{t})  &&\text{for all $t, \hat{t} \in \mathcal{T}$\,;} \\
   v(x(t), t) - p(t) &\geq    0  &&\text{for all $t \in \mathcal{T}$}\,.  
\end{align*}
Formally, the principal's objective is given by a weighted sum of welfare and revenue: 
\[\E\big[\lambda (t) \big(v(x(t), t) - p(t)\big)\big] + \E\big[p(t)\big]\,,\]
where $\lambda(\,\cdot\,) \geq 0$ is a continuous \textit{\textbf{welfare weight function}}. If providing allocation $x$ incurs a production cost $C(x)$, the principal's objective becomes $\E\big[\lambda(t) \big(v(x(t), t) - p(t)\big)\big] + \E\big[p(t) - C(x(t))\big]$. For any welfare weights $\lambda$, production costs can be normalized to zero by working with cost-adjusted values and payments.\footnote{In particular, all of our demand curves are understood to be net of production costs.} For the problem to be bounded, we assume $\E[\lambda(t)] \leq 1$.\footnote{If $\E[\lambda(t)] > 1$, the principal can approach infinite value by scaling up a lump-sum~transfer~to~the~agent.} We assume that the welfare weights are \textit{\textbf{redistributive}} in the sense that $\lambda(t)$ is nonincreasing in $t$.\footnote{Since the agent's preferences are pinned down by $t$, the formulation is equivalent to having $(\lambda, t)$ jointly distributed with $\E[\lambda \mid t]$ being nonincreasing in $t$.} With constant $\lambda$, this objective captures any weighted average of surplus and revenue, with $\lambda \equiv 0$ corresponding to profit maximization.  With varying $\lambda$, it accommodates optimal taxation with quasilinear preferences\footnote{This corresponds to $\E[\lambda(t)]=1$ since revenue would be transferred as a lump sum to the agent.} and redistributive allocation problems in which lower types receive weakly higher welfare weights.

A \textit{\textbf{menu}} $\mathcal{Y}$ is any subset of $\mathcal{X}$ (which we assume also includes the outside option $\varnothing$).\footnote{To simplify notation, we omit the inclusion of $\emptyset$ from a menu whenever it is clear.} 
A menu $\mathcal{Y}$ is \textit{\textbf{optimal}} if there exists an optimal mechanism $(x, p)$ such that $x(t) \in \mathcal{Y}$ for all types $t$. Since the space of allocations $\mathcal{X}$ is abstract, we can allow for stochastic mechanisms by encoding them in $\mathcal{X}$. However, some of our results require only comparisons among the elements of $\mathcal{X}$ to imply optimality even if the allocation space is extended to $\Delta(\mathcal{X})$ (under expected utility). Thus, we explicitly say that a menu $\mathcal{Y}$ is \textit{\textbf{optimal among stochastic mechanisms}} if there exists a mechanism $(x, p)$ such that $x(t) \in \mathcal{Y}$ for all types $t$ and $(x, p)$ is optimal among all IC and IR mechanisms mapping from $\mathcal{T}$ to $\Delta(\mathcal{X}) \times \R$.

Let $P(x, q)$ be the \textit{\textbf{demand curve}} for allocation $x$ when $x$ is sold alone: 
\[P(x, q) := v(x, F^{-1}(1 - q))\,.\]
Let $\eta(x, q)$ be the corresponding \textit{\textbf{elasticity curve}} for allocation $x$:
\[\eta(x, q) := \Big|\frac{\d \log P(x, q)}{\d \log q}\Big|^{-1} \,.\]
Note that a higher $\eta$ represents a more elastic demand curve. For any $x, x' \in \mathcal{X}$, we write 
\[x \preceq_{\text{surplus}} x' \,\,\,\text{ if $P(x, q) \leq P(x', q)$ for all $q$}\,;\]
and 
\[x \preceq_{\text{elasticity}} x' \,\,\,\text{ if $\eta(x, q) \leq \eta(x', q)$ for all $q$}\,.\]
We write $\prec_{\text{surplus}}$ or $\prec_{\text{elasticity}}$ if the above inequalities hold strictly for all $q \in (0, 1)$. Note that for the above elasticity comparison to be well defined, we implicitly assume that $P(x, q) > 0$, i.e., $v(x, t) > 0$ for any $x \neq \emptyset$.\footnote{As a convention, for $x = \emptyset$, we take the elasticity to be pointwise infinity.} We will relax this assumption for our notion of generalized frontiers, which replaces elasticity comparisons with single-crossing comparisons of demand curves.

\section{Main Results}\label{sec:main}

\subsection{Optimality of Surplus-Elasticity Frontier}
A compact subset $\mathcal{X}^\star \subseteq \mathcal{X}$ is a \textit{\textbf{surplus-elasticity frontier}} if:
\begin{itemize}
    \item[\textit{(i)}] for any $x \not\in \mathcal{X}^\star$, there exists $x' \in \mathcal{X}^\star$ such that 
    \[\text{ $x \preceq_{\text{surplus}} x'$ and $x \preceq_{\text{elasticity}} x'$ } \,\,;\]
    \item[\textit{(ii)}] there exists an index set $S\subset \R$ such that $\X^\star = \{x_s\}_{s\in S}$ and for all $s < s'$ 
    \[x_{s} \prec_{\text{surplus}} x_{s'} \text{ and } x_{s} \succ_{\text{elasticity}} x_{s'}\,\,.\]
\end{itemize}
The first requirement is a ``domination'' condition; the second requirement is an ``anti-chain'' condition. Intuitively, when both the surplus order and the elasticity order are total orders, this definition is just like a Pareto frontier as in \Cref{fig:frontier}. 

Our first main result says that any surplus-elasticity frontier is an optimal menu. 
\begin{theorem}\label{thm:main}
Any surplus-elasticity frontier is an optimal menu. 
\end{theorem}
The proof of \Cref{thm:main} is in the appendix. We sketch the proof in \Cref{sec:proof}.

We make a few remarks about this result here. 

First, note that 
\[x \preceq_{\text{surplus}} x'  \iff v(x, t) \leq v(x', t) \text{ for all $t$} \tag{1}\]
and 
\[\qquad \qquad \quad \,\, x \preceq_{\text{elasticity}} x'  \iff \frac{\d}{\d t} \log v(x, t) \geq \frac{\d}{\d t} \log v(x', t) \text{ for all $t$}\,.\footnote{Indeed, to see this, note that $x \preceq_{\text{elasticity}} x'$ if and only if $\frac{v(x, t)}{v_t(x, t)} \cdot \frac{f(t)}{1 - F(t)} \leq \frac{v(x', t)}{v_t(x', t)} \cdot \frac{f(t)}{1 - F(t)}$ for all $t$, and the common factor $\frac{f(t)}{1 - F(t)}$ cancels.}  \tag{2} \label{eq:log-d}\]
Thus, the definition of a surplus-elasticity frontier does not depend on the type distribution $F$. Moreover,  the definition of the frontier is also agnostic to the welfare weights $\lambda(t)$. Perhaps surprisingly,  \Cref{thm:main} asserts that the \textit{same} frontier is optimal across all type distributions $F$ and all weakly decreasing welfare weights $\lambda(t)$---we refer to this property as \textit{\textbf{robust optimality}}. Thus, just like in \citet{Myerson1981}, \Cref{thm:main} provides a clean decomposition between menu design and optimal pricing: \textit{(i)} once the designer can identify the frontier from the demand system, she can already assemble an optimal menu of products, and \textit{(ii)} as the designer obtains more information about the details of the environment (e.g., via price experimentation), she can fine-tune the prices.

Second, to further illustrate the intuition of the surplus-elasticity frontier, we connect the two orders to two stochastic orders. Let $\preceq_{\text{fosd}}$ denote the \textit{\textbf{first-order stochastic dominance order}}. Let $\preceq_{\text{disp}}$ denote the \textit{\textbf{dispersive order}} (\citealt{muller2002comparison}), i.e., two random variables $X \preceq_{\text{disp}} Y$ if $F^{-1}_X(q') - F^{-1}_X(q) \leq F^{-1}_Y(q') - F^{-1}_Y(q)$ for all $q < q' \in [0, 1]$. Let $V^x$ denote the random variable induced by $v(x, t)$. Note that 
\[x \preceq_{\text{surplus}} x'  \iff V^x \preceq_{\text{fosd}} V^{x'}  \tag{3}\]
and 
\[\qquad \, \, \,\, x \preceq_{\text{elasticity}} x'  \iff \log V^{x} \succeq_{\text{disp}} \log V^{x'}\,.  \tag{4}\]
In particular, the elasticity order coincides with the reverse dispersive order applied to log values. Intuitively, a more elastic demand curve corresponds to a less dispersed log-value distribution and hence results in lower information rents per unit of surplus.\footnote{Indeed, see also \citet*{yang2023comparison} who show that log-level dispersion is important in understanding the rent an agent can secure in one-dimensional linear screening problems.}

Third, along the surplus-elasticity frontier, note that the agent's preferences satisfy \textit{\textbf{increasing differences}}: For any $s < s'$, 
\[x_{s} \prec_{\text{surplus}} x_{s'} \text{ and } x_{s} \succ_{\text{elasticity}} x_{s'} \implies \text{$v(x_{s'}, t) - v(x_s, t)$ is strictly increasing in $t$}\,.\footnote{Indeed, for all $s < s'$, we can write $v_t(x_s, t) = \frac{v_t(x_s, t)}{v(x_s, t)} \cdot v(x_s, t) < \frac{v_t(x_{s'}, t)}{v(x_{s'}, t)} \cdot v(x_{s'}, t) = v_t(x_{s'}, t)$,
where the strict inequality uses that the menu $\X^\star$ is a surplus-elasticity frontier and that $0 < v(x_{s}, t) \leq v(x_{s'}, t)$. } \tag{5}\]
Thus, the optimal pricing problem along the frontier can always be solved via well-known one-dimensional screening methods. 

Fourth, a particular special case is when the frontier consists of only one element, say $\overline{x}$. That is, element $\overline{x}$ generates more surplus than any other element in $\mathcal{X}$ and has a more elastic demand curve than that of any other element in $\mathcal{X}$. \Cref{thm:main} says that it is optimal to simply offer this maximal-surplus element, i.e., no screening distortion is needed. This connects to the optimality of pure bundling as studied in \citet{haghpanah2021pure}. Indeed, note that by \eqref{eq:log-d}, we have 
\[x \preceq_{\text{elasticity}} x'  \iff v(x, t)/v(x', t) \text{ is nondecreasing in $t$}\,.  \tag{6}\]
The main result of \citet{haghpanah2021pure} establishes that if $v(b, t) / v(\overline{b}, t)$ is nondecreasing in $t$ for all bundles $b$, then selling only the grand bundle $\overline{b}$ is profit maximizing even if randomization is allowed. This can be seen from \Cref{thm:main} by taking $\mathcal{X}$ to be the space of lotteries over bundles, and observing that mixtures of deterministic bundles yield a randomized bundle with demand curve pointwise below that of the grand bundle, and pointwise less elastic than that of the grand bundle given the ratio-monotonicity condition---the surplus-elasticity frontier consists of only the grand bundle. Perhaps surprisingly, \Cref{thm:main} shows that the identical ratio-monotonicity condition of \citet{haghpanah2021pure} actually implies the optimality of pure bundling (or equivalently, no screening distortion besides exclusion) for arbitrary redistributive welfare weights. We discuss more applications to optimal bundling in \Cref{subsec:bundling}.

Fifth, in the above application to optimal bundling, we take the allocation space to be the space of lotteries over bundles, but it turns out that we may be able to compare only elements in $\mathcal{X}$ and conclude the optimality of the menu in a much larger allocation space $\Delta(\mathcal{X})$---\Cref{subsec:stochastic} develops a strengthening of the surplus-elasticity frontier that precisely gives such stochastic optimality. In general, a minimal optimal menu---one from which no option can be removed---would be a subset of the surplus-elasticity frontier, but as we show in \Cref{subsec:stochastic}, there are conditions under which the surplus-elasticity frontier is itself a minimal optimal menu for profit maximization.

Sixth, given the broad intuition behind the surplus order and the elasticity order, a curious reader may be tempted to conjecture that it is always optimal to use  the ``undominated'' elements according to the surplus order and the elasticity order. This conjecture, however, turns out to be false. To illustrate this subtlety, our next example shows that without the ``anti-chain'' property in our definition, the ``undominated'' elements, even if they satisfy the  increasing differences property, do \textit{not} suffice for optimality---the optimal menu may require two allocations where one allocation's demand curve is pointwise lower and pointwise less elastic than the other allocation's. 

\begin{ex}[``Anti-chain'' property cannot be dropped in \Cref{thm:main}]\label{ex:counterexample}
There are three types $\{t_L, t_M, t_H\}$ with probabilities given by $(0.45, 0.44, 0.11)$. (The three-type case can be approximated via a continuous type distribution.) There are three non-outside-option allocations $\{x_1, x_2, x_3\}$. The values of each type for each allocation are given by: 
\[
\begin{array}{c|ccc}
& t_L & t_M & t_H\\
\hline
x_3 & 1    &  30 & 200 \\
x_2 & 0.8 & 24 & 160 \\
x_1 & 12  & 13 & 140 
\end{array}
\]
Note that this example satisfies increasing differences. Moreover, $x_2$ is equivalent to a lottery over $\{x_3, \varnothing\}$ with probability $0.8$ on $x_3$ and probability $0.2$ on $\varnothing$. 

Consider the menu $\{x_1, x_3\}$. By construction, we have $x_2 \preceq_{\text{surplus}} x_3$ and $x_2 \preceq_{\text{elasticity}} x_3$. Thus, condition \textit{(i)} (the ``domination'' condition) holds for the menu $\{x_1, x_3\}$. However, we claim that the menu $\{x_1, x_3\}$ is not optimal. Since increasing differences hold on this menu, it is straightforward to verify that the optimal pricing of this menu would be to offer $x_3$ alone at a price of $200$, generating a revenue of $22$. Now, consider the menu 
$\{x_1, x_2, x_3\}$ with prices given by $p(x_1) = 12$, $p(x_2) = 23$, and $p(x_3)= 63$. Under this menu, type $t_L$ finds it optimal to choose $x_1$, type $t_M$ finds it optimal to choose $x_2$, and type $t_H$ finds it optimal to choose $x_3$.  Thus, the seller collects a revenue 
\[12 \cdot 0.45 + 23 \cdot 0.44 + 63 \cdot 0.11 = 22.45 > 22\,,\] strictly higher than the maximal revenue under the menu $\{x_1, x_3\}$. Moreover, it can be verified that the menu $\{x_1, x_2, x_3\}$ with the given prices is profit maximizing. \qed
\end{ex}

The two requirements of a surplus-elasticity frontier can, however, be suitably weakened; we develop this generalization in \Cref{subsec:generalized}. \Cref{ex:counterexample} is also instructive for illustrating the subtleties in the proof; we will return to this example in \Cref{sec:proof} when sketching the proof of \Cref{thm:main}.

\paragraph{Ordered Surplus and Elasticity.}\hspace{-2mm}In most of our applications, the existence of a surplus-elasticity frontier is shown by construction, which then implies optimality by \Cref{thm:main}. However, when both the surplus and elasticity comparisons are total orders, the existence of a surplus-elasticity frontier is guaranteed, just like a standard Pareto frontier:

\begin{cor}\label{cor:ordered}
Consider any finite $\X = \{(a_i, b_i)\}_{i=1}^n\subset \R^2$.  Suppose $(a, b) \preceq_{\emph{surplus}} (a', b')$ if $a \leq a'$ and strictly so if $a < a'$; $(a, b) \preceq_{\emph{elasticity}} (a', b')$ if $b \leq b'$ and strictly so if $b < b'$. Then, the surplus-elasticity frontier exists and is optimal. 
\end{cor}

\paragraph{Which Product Dimension to Distort?}\hspace{-2mm}As another illustration of \Cref{thm:main}, we generalize \Cref{ex:two-dimension} about the distortion of product features: 
\begin{prop}\label{prop:distortion}
Suppose $\mathcal{X} = [0, 1]^2$, $v(x_1, x_2, t)$ is strictly increasing in $(x_1, x_2)$, log-submodular in $(x_1, t)$, and strictly log-supermodular in $(x_2, t)$. Then, $\{(1, x_2)\}_{x_2 \in [0, 1]}$ is an optimal menu.   
\end{prop}
\begin{proof}
Since $v(x_1, x_2, t)$ is nondecreasing in $x_1$ and log-submodular in $(x_1, t)$, note that $(x_1, x_2) \preceq_{\text{surplus}} (1, x_2)$ and $(x_1, x_2) \preceq_{\text{elasticity}} (1, x_2)$. Thus, $\X^\star := \{(1, x_2)\}_{x_2 \in [0, 1]}$ satisfies the ``domination'' condition. Since $v(x_1, x_2, t)$ is strictly increasing in $x_2$ and strictly log-supermodular in $(x_2, t)$, note that for any $x_2 < x'_2$, we have $(1, x_2) \prec_{\text{surplus}} (1, x'_2)$ and $(1, x_2) \succ_{\text{elasticity}} (1, x'_2)$. Thus, $\X^\star$ also satisfies the ``anti-chain'' condition. So $\X^\star$ is a surplus-elasticity frontier, and thus optimal by \Cref{thm:main}. 
\end{proof}

Recall that in \Cref{ex:two-dimension}, $v(x_1, x_2, t)$ takes the form $(x_1 + t)^{x_2}$, which is log-submodular in $(x_1, t)$ and log-supermodular in $(x_2, t)$---these notions exactly generalize the idea that an increase in feature $x_1$ (e.g., basic functionality) \textit{\textbf{flattens}} the demand curve, while an increase in feature $x_2$ (e.g., product design) \textit{\textbf{steepens}} the demand curve, leading to no distortion of dimension $x_1$ in the surplus-elasticity frontier.\footnote{When $v(x_1, x_2, t)$ does not depend on $x_2$, \Cref{prop:distortion} itself nests one-dimensional screening results on the profitability of price discrimination (\citealt{johnson2003multiproduct,anderson2009price}).}

\subsection{Stochastic Optimality and Minimal Optimality}\label{subsec:stochastic}

To state our second result, for any $x \prec_{\text{surplus}} x'$, write the \textit{\textbf{incremental demand curve}} as 
\[P(x', q\mid x) := P(x', q) - P(x, q)\,,\]
and the \textit{\textbf{incremental elasticity curve}} as 
\[\eta(x', q \mid x) := \Big|\frac{\d \log P(x', q\mid x)}{\d \log q}\Big|^{-1}\,.\]
A menu $\Y = \{x_s\}_{s \in S}$, with its allocations ordered according to the surplus order, has the \textit{\textbf{ordered incremental elasticity property}} if the incremental demand curves are downward sloping and they are increasingly inelastic in that for all $s_1 < s_2 < s_3$, 
\[\eta(x_{s_2}, q \mid x_{s_1}) > \eta(x_{s_3}, q \mid x_{s_2})\, \text{ for all $q \in (0, 1)$}\,. \tag{7} \label{eq:incremental-ela}\]
 In the case of finite $S$, it is enough to check the above for \textit{\textbf{adjacent increments}}. Moreover, because the elasticity comparison is pointwise, by the same reasoning as in \eqref{eq:log-d}, this property is also \textit{agnostic} to the type distribution since it is equivalent to the pointwise comparison of the log-derivative of the incremental utility with respect to the type.\footnote{That is, $\frac{\d}{\d t} \log(v(x_i, t) - v(x_{i-1}, t)) < \frac{\d}{\d t} \log(v(x_{i+1}, t) - v(x_{i}, t))$ for all $t$ and all $i$.}

We say that a surplus-elasticity frontier is a \textit{\textbf{strong frontier}} if it has the ordered incremental elasticity property. This property strengthens the frontier's ``anti-chain'' condition: not only the original demand curves but also the incremental demand curves become increasingly inelastic along the surplus order. This simple strengthening turns out to be enough for optimality among all stochastic mechanisms:

\begin{theorem}\label{thm:main2}
Any strong frontier is optimal among stochastic mechanisms. 
\end{theorem}

The proof of \Cref{thm:main2} is in the appendix. We sketch the proof in \Cref{sec:proof}. 

\paragraph{Stochastic versus Deterministic Mechanisms.}\hspace{-2mm}It is well known that even with one-dimensional screening problems, stochastic mechanisms can outperform deterministic mechanisms (\citealt{Strausz2006}). An immediate consequence of \Cref{thm:main2} is a new condition for when deterministic mechanisms are optimal in one-dimensional screening problems: 

\begin{prop}\label{prop:stochastic}
Suppose $\mathcal{X} = [0, 1]$ and $v(x, t)$ is continuously differentiable and strictly increasing in $x$ with $v(0, t) = 0$. Then, if $v_x(x, t)$ is \emph{strictly log-supermodular} or \emph{weakly log-submodular}, then a deterministic mechanism is optimal. 
\end{prop}
\begin{proof}
Suppose $\log(v_x(x,t))$ is strictly supermodular. Then, 
$\frac{v_x(x_2, t)}{v_x(x_1, t)}$
is strictly increasing in $t$ for all $x_2 > x_1$, which implies that the demand curves generated by the marginal values $v_x(x, t)$ have ordered elasticities, with higher $x$ generating a more inelastic demand curve. This also implies that $\frac{v(x_2, t)}{v(x_1, t)}= \frac{\int_{0}^{x_2}v_x(s, t) \d s}{\int_{0}^{x_1}v_x(s, t) \d s}$ is strictly increasing in $t$ for all $x_2 > x_1$. Therefore, note that $\X$ itself is a strong frontier and hence optimal in $\Delta(\X)$ by \Cref{thm:main2}. 

Now, suppose $\log(v_x(x,t))$ is weakly submodular. Then, $\frac{v_x(x_2, t)}{v_x(x_1, t)}$ is weakly decreasing in $t$ for all $x_2 > x_1$, which implies that $\frac{v(x_2, t)}{v(x_1, t)}= \frac{\int_{0}^{x_2}v_x(s, t) \d s}{\int_{0}^{x_1}v_x(s, t) \d s}$ is weakly decreasing in $t$ for all $x_2 > x_1$. Therefore, note that $\X^\star = \{1\}$ is a surplus-elasticity frontier, and automatically a strong frontier since the ordered incremental elasticity property has no bite. Thus, offering the single option $x = 1$ with a suitable price is optimal among all stochastic mechanisms by \Cref{thm:main2}.
\end{proof}

In contrast to the known result on the optimality of deterministic mechanisms by \citet{Strausz2006}, this result does not depend on the type distribution or the redistributive welfare weights, which can certainly violate the Myersonian regularity conditions and lead to bunching.

\paragraph{Demand Profile and Minimal Optimality.}\hspace{-2mm}The pricing problem on a surplus-elasticity frontier can always be solved by one-dimensional methods. However, it turns out that if the frontier is a strong frontier and the objective is profit maximization, then the pricing problem can always be solved via a very simple method, the \textit{\textbf{demand profile method}}, which simply prices each upgrade against the incremental demand curve. To see this connection, suppose that the strong frontier is \textit{\textbf{differentiable}} in the sense that $P(x_s, q)$ is continuously differentiable in $s$. Without loss of generality, normalize the index $s \in [0, 1]$. Let $\Delta P(s, q) := \partial_s P(x_s, q)$ denote the marginal demand curve for the upgrade at $s$. Since the agent's preferences have strict increasing differences on the frontier, by the demand profile method (\citealt{wilson1993nonlinear}), an upper bound on the revenue of any mechanism with allocations restricted to $\X^\star$ can be obtained by consecutive \textit{\textbf{upgrade pricing}}, $\int \max_{q \in [0, 1]}(\Delta P(s, q) \cdot q) \d s$, where for each $s$, the inner maximization problem solves a separate upgrade pricing problem on the marginal upgrade. The above pointwise objective after taking log is $\log \Delta P(s, q) + \log q$, which is submodular in $(s, q)$ given that $\X^\star$ has the ordered incremental elasticity property. Thus, there exists a selection of pointwise optimizers $s \mapsto q^\star(s)$ that is nonincreasing and hence can be implemented via the profile of upgrade prices, i.e., the demand profile method is valid on the strong frontier. 

Now, say that the frontier is \textit{\textbf{interior}} if any monopoly quantity of the upgrade pricing problem is interior, $\argmax \{\Delta P(s, q) \cdot q\} \subset (0, 1)$, and say that it is \textit{\textbf{minimal optimal}} if for any closed subset $\X' \subset \X^\star$ that is also optimal, $\X' = \{x_{s}\}_{s \in S'}$ where $S'$ has Lebesgue measure $1$. Combining the above argument with \Cref{thm:main2}, we immediately obtain: 

\begin{cor}[Demand Profile and Minimal Optimality]\label{cor:profile}
Suppose the objective is profit maximization (i.e., $\lambda \equiv 0$). For any differentiable strong frontier $\X^\star$, the demand profile method applied to $\X^\star$ yields a mechanism that is optimal among all stochastic mechanisms, i.e., with allocation space being $\Delta(\X)$. Moreover, if such a frontier is interior, then it is minimal optimal.
\end{cor}

\subsection{Generalized Frontier and Robust Optimality}\label{subsec:generalized}

Our third main result introduces a generalization of the frontier concept. This notion of a generalized frontier is not only sufficient but, under suitable conditions, also necessary for robust optimality. Intuitively, the generalization allows \textit{(i)} a collection of randomized demand curves to ``dominate'' a single demand curve and \textit{(ii)} the ``flatness'' of the demand curves to be measured via single crossings instead of pointwise elasticities. 

To state this generalization, we use the notation $P(a, q)$ to write the pointwise mixture of the demand curves induced by a lottery $a \in \Delta(\X)$. For a partially ordered set $(\Y, \preceq)$, recall that its induced \textit{\textbf{stochastic dominance order}} $\preceq_{\text{st}}$ is defined on $\Delta(\Y)$ by $a \preceq_{\text{st}} a'$ if $\E_{a}[h(x)]\leq \E_{a'}[h(x)]$ for all monotone functions $h$ on $\Y$. 

We say that $x$ is \textit{\textbf{covered}} by $\mathcal{A} \subset \Delta(\mathcal{X})$ if for any $q\in [0, 1]$, there exists $a \in \mathcal{A}$ such that $P(a, \,\cdot\,)$ weakly single-crosses $P(x, \,\cdot\,)$ from below at $q$ (recall \Cref{fig:covering}).\footnote{A function $g$ \textit{\textbf{weakly single-crosses}} $h$ \textit{\textbf{from below at}} $k$ if $g(s) \leq h(s)$ for all $s \leq k$ and $g(s) \geq h(s)$ for all $s \geq k$.} 
We say a compact subset $\mathcal{X}^\star \subset \mathcal{X}$ is a \textit{\textbf{generalized frontier}} if $\X^\star$ can be totally ordered by some $\preceq$ and there exists some compact $\mathcal{A} \subset \Delta(\X^\star)$ totally ordered by $\preceq_{\text{st}}$ such that 
\begin{itemize}
    \item[\textit{(i)}]  Every $x \in \X$ is covered by $\mathcal{A}$\,;
    \item[\textit{(ii)}] $P(x', q) - P(x, q)$ is strictly decreasing in $q$ for all $x \prec x' \in \X^\star$. 
\end{itemize}
A special case of a generalized frontier is when every element $x \in \X$ is covered by $\X^\star$ for which the agent's preferences satisfy increasing differences along some ordering (in which case we can simply take $\mathcal{A} \subset \Delta(\X^\star)$ to be the set of Dirac measures).

Note that, as with the surplus-elasticity frontier, the definition of a generalized frontier is agnostic to the type distribution and redistributive welfare weights. As we discuss in \Cref{sec:proof}, the proof of \Cref{thm:main} actually shows that \textit{(i)} a surplus-elasticity frontier is a generalized frontier and \textit{(ii)} a generalized frontier is sufficient for optimality. Thus, we immediately obtain the following result from the proofs of \Cref{thm:main} and \Cref{thm:main2}:

\begin{theorem}\label{thm:main3}
Any generalized frontier is an optimal menu. Moreover, if the generalized frontier has surplus-ordered demand curves satisfying the ordered incremental elasticity property, then it is optimal among stochastic mechanisms. 
\end{theorem}

We next show that, under the same incremental-elasticity condition, the notion of a generalized frontier in fact \textit{\textbf{characterizes}} robust optimality, i.e., optimality under any type distribution and any redistributive welfare weights, and then describe two immediate applications of \Cref{thm:main3}.  

\paragraph{Characterization of Robust Optimality.}\hspace{-2mm}Intuitively, a generalized frontier uses single crossings instead of pointwise elasticities to measure the ``flatness'' of a demand curve. This generalized-frontier property turns out to be also necessary if we want a finite menu that satisfies the ordered incremental elasticity property (which, recall, is agnostic to the type distribution) to be profit-maximizing under every type distribution.  

\begin{theorem}[Characterization]\label{thm:converse}
Suppose $v(x, t) \geq 0$ for all $x$ and $t$, and $\mathcal{Y}$ is a finite surplus-ordered menu satisfying the ordered incremental elasticity property. Then, $\mathcal{Y}$ is optimal for profit maximization under every Borel type distribution if and only if $\mathcal{Y}$ is a generalized frontier. 
\end{theorem}

The proof of \Cref{thm:converse} is in the appendix. In \Cref{thm:converse}, we extend the admissible type distributions to all Borel type distributions to allow for atoms, which does not change the sufficiency direction of \Cref{thm:main3}; we then construct type distributions with finite supports such that optimality under the profit maximization objective with respect to these type distributions implies that the original menu must be a generalized frontier, certifying the necessity direction. As a consequence of \Cref{thm:main3} and \Cref{thm:converse}, for any finite menu satisfying the ordered incremental elasticity property,\footnote{For a surplus-ordered menu, it can be shown that the ordered incremental elasticity property is itself necessary if the menu is \textit{minimal optimal} for profit maximization under a ``rich'' set of type distributions.} the following are equivalent: \textit{(i)} it is a generalized frontier, \textit{(ii)} it is an optimal menu for profit maximization under all type distributions, \textit{(iii)} it is an optimal menu for all redistributive welfare weights and all type distributions, and \textit{(iv)} it is optimal among stochastic mechanisms for all redistributive welfare weights and all type distributions. 

\paragraph{Screening with Contracts.}\hspace{-2mm}Many contracts can be viewed as one-dimensional functions (as we discuss in \Cref{sec:application}). Here, we describe an abstract setting where the allocation space is the space of bounded measurable functions, and the agent's utility is a linear functional. Consider the space of bounded measurable functions mapping from $[0, 1]$ to $[0, 1]$, endowed with the $L_1$ norm. Suppose that the valuation of each type $t$ for each function $x$ is linear: $v(x, t) := \int x(s) u(s, t) \d G(s)$, where $u(s, t)$ is a positive bounded measurable function that is continuously differentiable in $t$ with $u_t$ uniformly bounded, and $G$ is absolutely continuous with full support.
\begin{prop}\label{prop:function}
Suppose $u(s, t)$ is strictly increasing in $t$. If $u$ is log-supermodular, then $\{\1_{[0, k]}\}_k$ is an optimal menu. If $u$ is log-submodular, then $\{\1_{[k,1]}\}_k$ is an optimal menu. 
\end{prop}
\begin{proof}
Suppose $u$ is log-supermodular. Fix any function $x(s)$ and any type $t$. For each fixed $t$, consider $x^\star(s; t) = \1_{s \leq k(t)}$ where $k(t)$ is chosen such that $\int x^\star(s; t) u(s, t) \d G(s)= \int x(s) u(s, t) \d G(s)$. Such a $k(t)$ exists by the intermediate value theorem. Now observe that, by \citet{karlin1956theory}, since $x^\star(\,\cdot\,; t)$, by construction, weakly single-crosses $x(\,\cdot\,)$ from above at $k(t)$, and $u$ is log-supermodular, we have that \[\Delta(z):= \int \big(x^\star(s; t) - x(s)\big) u(s, z) \d G(s)\]
weakly single-crosses $0$ from above at $z = t$. Therefore, the demand curve resulting from the allocation $x^\star(\,\cdot\,; t)$ must weakly single-cross that resulting from $x(\,\cdot\,)$ from below at the quantile of type $t$, i.e., at $q = 1 - F(t)$. Since this holds for all $t$, the demand curve resulting from $x(\,\cdot\,)$ is covered by those of $\X^\star := \{\1_{[0, k]}\}_{k}$ according to our definition. Moreover, along the ordering given by $k$, we have that $\int^{k_2}_0 u(s, t) \d G(s) - \int^{k_1}_0 u(s, t) \d G(s)  = \int_{k_1}^{k_2} u(s, t) \d G(s)$, which is strictly increasing in $t$. Then, taking $\mathcal{A}$ as the collection of Dirac measures supported on each of $\{\1_{[0, k]}\}_{k}$ shows that $\X^\star$ is a generalized frontier. Thus, by \Cref{thm:main3}, $\{\1_{[0, k]}\}_{k}$ must be optimal.\footnote{Formally, our model assumes $\X$ is compact, but as the proof of \Cref{thm:main3} shows, this is not needed as long as $\X^\star$ is compact.} The case of $u$ being log-submodular is exactly symmetric, because we can construct $x^\star(s; t) = \1_{s \geq k(t)}$ where $k(t)$ is defined similarly. By the same single-crossing argument,  $\{\1_{[k, 1]}\}_k$ is a generalized frontier and hence optimal by \Cref{thm:main3}. 
\end{proof}

\paragraph{Ordeal Mechanisms.}\hspace{-2mm}Consider a utility function $u(c, y, t)$ where there exists some $y_0$ such that $u(c, y_0, t) \geq u(c, y, t) \geq 0$ for all $(c, y) \in \mathcal{X}$. When is it optimal to set $y = y_0$? This question can be interpreted as when we should expect \textit{\textbf{ordeal mechanisms}} to be used in either profit-maximizing contexts or, more generally, in contexts with redistributive concerns. We will discuss in detail an application to optimal taxation in \Cref{subsec:ordeal}, but let us first illustrate how to think about the problem using a generalized frontier.
\begin{prop}\label{prop:costly}
Suppose $u(c, y_0, t)$ is strictly supermodular and log-supermodular in $(c, t) \in [0, 1]^2$, and satisfies $u(0, y_0, t) = 0$. If for any $(c, y) \in \X$, $\eta(c, y, q) \leq \eta(c, y_0, q) \text{ for all $q$}$, then $\big\{(c, y_0)\big\}_{c}$ is an optimal menu (i.e., the optimal menu involves no costly screening).
\end{prop}
\begin{proof}
Fix any $(c, y)$ and $t$ with $u(c, y, t) > 0$. By the intermediate value theorem, there exists $c^\star_t \in [0, c]$ such that $u(c, y, t) = u(c^\star_t, y_0, t)$.
Now, note that $\frac{u(c, y, s)}{u(c^\star_t, y_0, s)} = \frac{u(c, y_0, s)}{u(c^\star_t, y_0, s)} \cdot \frac{u(c, y, s)}{u(c, y_0, s)}$ is nondecreasing in $s$, since the first term is nondecreasing in $s$ by log-supermodularity of $u(c, y_0, s)$ in $(c, s)$ and the second term is nondecreasing in $s$ by the elasticity comparison given in the statement. It follows that the demand curve $P(c^\star_t, y_0, q)$ single-crosses $P(c, y, q)$ from below with the crossing point corresponding to the quantile of type $t$. Since this holds for all $t$, it follows immediately that any $(c, y)$ is covered by $\{(c^\star_t, y_0)\}_{t}$ and hence by $\{(c, y_0)\}_{c}$. Given that $u(c, y_0, t)$ has strict increasing differences in $(c, t)$, it then follows immediately that $\{(c, y_0)\}_{c}$ is a generalized frontier and hence optimal by \Cref{thm:main3}. 
\end{proof}

In monopolistic pricing contexts, this means that costly screening can only be profitable if, with an ordeal requirement, the demand for some good becomes more \textit{\textbf{elastic}}. In the optimal taxation context, as we discuss in \Cref{subsec:ordeal}, the demand elasticity comparison reduces to a \textit{\textbf{take-up elasticity}} comparison.

\section{Proof Sketches for the Main Results} \label{sec:proof}

In this section, we sketch the proof of \Cref{thm:main} assuming a finite allocation space $\mathcal{X}$, and then discuss the proofs of \Cref{thm:main2} and \Cref{thm:main3} at the end. The appendix provides details and generalizes the proofs to arbitrary $\mathcal{X}$ via approximation. 

\paragraph{A False Proof Strategy.}\hspace{-2mm}Before presenting the actual proof, it is instructive to sketch a more intuitive but false proof strategy.\footnote{I thank Nima Haghpanah for suggesting this discussion.} Let $\mathcal{X}^\star$ be the surplus-elasticity frontier. Fix any mechanism $(x, p)$. For every type $t$, consider \textit{(i)} replacing the allocation $x(t)$ with some $\tilde{x}(t) \in \mathcal{X}^\star$ where $x(t) \preceq_{\text{surplus}} \tilde{x}(t)$ and $x(t) \preceq_{\text{elasticity}} \tilde{x}(t)$, and  \textit{(ii)} replacing $p(t)$ with $\tilde{p}(t)$ such that 
\[v(x(t), t) - p(t) = v(\tilde{x}(t), t) - \tilde{p}(t)\,.\]
Since $\tilde{x}(t)$ generates more surplus than $x(t)$, this construction would yield $\tilde{p}(t) \geq p(t)$, leading to a higher revenue for the principal under truthful reporting. The hope is that the elasticity comparison would then imply that the modified mechanism would remain incentive compatible or at least would not introduce deviations that lower the total payment to the principal. However, this proof strategy would in general fail. In particular, note that it does not use the ``anti-chain'' property in our definition. 

To see the failure directly, let us return to \Cref{ex:counterexample}. Consider the menu 
$\{x_1, x_2, x_3\}$ with prices given by $p(x_1) = 12$, $p(x_2) = 23$, and $p(x_3)= 63$. The corresponding direct mechanism is given by $x(t_L) = x_1, p(t_L) = 12$, $x(t_M) = x_2, p(t_M) = 23$, and $x(t_H) = x_3, p(t_H) = 63$. Recall that $x_2$ is equivalent to a lottery over $\{x_3, \varnothing\}$ and hence dominated by $x_3$ in both the surplus order and the elasticity order. Now, consider replacing the allocation of type $t_M$ from $x_2$ to $x_3$ while keeping type $t_M$'s payoff the same. This would require $24 - p(t_M) = 30 - \tilde{p}(t_M)$ and hence $\tilde{p}(t_M) = 29$. Clearly, this would not result in an incentive-compatible mechanism---in particular, type $t_H$ would mimic $t_M$ to get the allocation $x_3$ at the lower price $29$ instead of $63$. The net change in the revenue would be 
\[0.44 \cdot (29 - 23) + 0.11 \cdot (29 - 63) = -1.1 < 0\,,\]
and hence this would not be a revenue-improving modification. 

\paragraph{The Actual Proof Strategy.}\hspace{-2mm}The actual proof strategy for \Cref{thm:main} is more involved, building on \citet{yang2022costly}. It consists of three steps.

\paragraph{\textbf{Step 1.}}\hspace{-2mm}Let $\X^\star$ be ordered by $\preceq_\text{surplus}$.  Fix any IC and IR mechanism $(x, p)$. We show that there exists a stochastic mechanism (which we refer to as a \textit{\textbf{reconstruction}}) $(a, p)$ where $a:\mathcal{T} \rightarrow \Delta(\X^\star)$ is a reconstructed allocation rule using only lotteries over elements in the frontier $\X^\star$, and $p$ is the same payment rule such that: 
\begin{itemize}
    \item $(a, p)$ is \textit{\textbf{downward incentive compatible}}, i.e., it satisfies all downward incentive constraints with $t > \hat{t} \in \mathcal{T}$ (and satisfies all the IR constraints)\,;
    \item $(a, p)$ generates \textit{\textbf{identical payoffs}} to each type as in $(x, p)$ (and hence also generates the same payoff to the principal) assuming truthful reporting by the agent\,;
    \item $\text{Ran}(a) \subseteq \Delta(\mathcal{X}^\star)$ consists of \textit{\textbf{stochastically ordered}} lotteries in the sense that $\text{Ran}(a)$ can be totally ordered by $\preceq_{\text{st}}$. 
\end{itemize}

\paragraph{\textbf{Step 2.}}\hspace{-2mm}We show that there exists another stochastic mechanism $(\tilde{a}, \tilde{p})$ where $\tilde{a}:\mathcal{T} \rightarrow \Delta(\X^\star)$ is a stochastic allocation rule further modifying the allocation rule $a$, and $\tilde{p}$ is a modified payment rule such that 
\begin{itemize}
    \item $(\tilde{a}, \tilde{p})$ is \textit{\textbf{fully incentive compatible}}\,;
    \item $(\tilde{a}, \tilde{p})$ generates a \textit{\textbf{weakly higher payoff}} for the principal compared to $(a, p)$\,;
    \item $\text{Ran}(\tilde{a}) \subseteq \Delta(\mathcal{X}^\star)$ continues to consist of \textit{\textbf{stochastically ordered}} lotteries. 
\end{itemize}
In particular, this step shows that the downward incentive constraints are \textit{\textbf{sufficient}} for the optimal solution. This step is of independent interest; as we discuss, it generalizes the \textit{\textbf{downward sufficiency theorem}} in \citet{yang2022costly} that holds for profit maximization in any one-dimensional screening problem to allow for redistributive welfare weights.

\paragraph{\textbf{Step 3.}}\hspace{-2mm}We show that there exists a \textit{\textbf{deterministic}} mechanism $(x^\dagger, p^\dagger)$ where $x^\dagger:\mathcal{T} \rightarrow \X^\star$ is a deterministic allocation rule that uses only elements in the frontier $\mathcal{X}^\star$, and $p^\dagger$ is a modified payment rule such that 
\begin{itemize}
    \item $(x^\dagger, p^\dagger)$ is \textit{\textbf{fully incentive compatible}}\,;
    \item $(x^\dagger, p^\dagger)$ generates a \textit{\textbf{weakly higher payoff}} for the principal compared to $(\tilde{a}, \tilde{p})$\,.
\end{itemize}
In particular, this step shows that in one-dimensional screening problems, the possibility of using lotteries that are ordered by stochastic dominance cannot improve the principal's objective. We refer to this result as the \textit{\textbf{purification lemma}}; its proof is by a monotone coupling argument.

\subsection{Reconstruction Lemma}
Fix any IC and IR mechanism $(x, p)$. For all $t$, let 
\[x^{+}(t):= \inf\Big\{x' \in \X^\star: v(x', t) \geq v(x(t), t)\Big\}\,,\quad  x^{-}(t):= \sup\Big\{x' \in \X^\star: v(x', t) < v(x(t), t)\Big\}\,,\]
where the $\inf$ and $\sup$ are defined with respect to the total order $\preceq_{\text{surplus}}$ on $\X^\star$, with the convention that when $x(t) = \varnothing$, we set $x^{+}(t) = x^{-}(t) = \varnothing$.

By construction, if $x(t)\in \X^\star$, then $x^{+}(t) = x(t)$. Moreover, for all $t$ where $x(t) \neq \varnothing$, 
\[v(x^{+}(t), t) \geq v(x(t), t) > v(x^{-}(t), t)\,.\]
Now, for all $t$ where $x(t) \neq \varnothing$, let 
\[\alpha(t) := \frac{v(x(t), t) - v(x^{-}(t), t)}{ v(x^{+}(t), t) -  v(x^{-}(t), t)} \in (0, 1]\,,\]
and for all $t$ where $x(t) = \varnothing$, let $\alpha(t) = 1$. The $(\X^\star, x, p)$-\textit{\textbf{reconstruction}} is a stochastic mechanism defined by: 
\begin{itemize}
    \item assigning each reported type $t$ a lottery over $\{x^{+}(t), x^{-}(t)\}$ with probability $\alpha(t)$ to allocation $x^{+}(t)$ and probability  $1- \alpha(t)$ to allocation $x^{-}(t)$\,;
    \item keeping the payment $p(t)$ for each reported type $t$ unchanged. 
\end{itemize}

By construction, the $(\X^\star, x, p)$-reconstruction yields the same payoff to each agent type $t$ as $(x, p)$ does (and hence the same payoff to the principal), assuming truthful reporting. The next lemma shows that the reconstruction preserves all downward IC constraints. 

\begin{lemma}[Reconstruction]\label{lem:reconstruct}
Let $\X^\star=\{\varnothing, x_1, \dots, x_m\}$ be a surplus-elasticity frontier. Let $(x, p)$ be any deterministic mechanism that satisfies all downward IC constraints. Let $(a, p)$ be the $(\X^\star, x, p)$-reconstruction. Then, $(a, p)$ satisfies all downward IC constraints: for all $\hat{t} < t \in \mathcal{T}$,
$v(a(t), t) - p(t) \geq v(a(\hat{t}), t) - p(\hat{t})$, where $v(a, t):= \E_{x\sim a}[v(x, t)]$. 
\end{lemma}

The proof is in the appendix. The intuition can be understood as follows. Fix any type $t$ with $x(t) \not\in \X^\star$. Recall that, by construction, type $t$ is \textit{\textbf{indifferent}} between the binary lottery $a(t)$ supported on $\{x^{+}(t), x^{-}(t)\} \subseteq \mathcal{X}^\star$ and the original allocation $x(t)$. By the ``domination'' property of $\X^\star$, there exists $\hat{x} \in \X^\star$ that dominates $x(t)$ in both the surplus order and the elasticity order. Crucially, since $x^{+}(t)$ and $x^{-}(t)$ generate less surplus compared to $\hat{x}$, they are more elastic than $\hat{x}$ by the ``anti-chain'' property of $\X^\star$, and hence must also dominate $x(t)$ in terms of the elasticity order. This implies that the lottery $a(t)$ generates a demand curve that is more elastic than the demand curve generated by $x(t)$, and moreover crosses the demand curve of $x(t)$ exactly at the quantile corresponding to type $t$ (by the indifference of type $t$). Note that pointwise elasticity ordering implies \textit{\textbf{single-crossing}} in the demand curves, and hence the demand curve of $a(t)$ must single-cross that of $x(t)$ from below. Then, under this reconstruction, for any type $t' > t$, the deviation payoff of mimicking type $t$ must have weakly decreased---indeed, the higher type $t'$ is located toward the left on the quantity axis, where the demand curve of $a(t)$ lies weakly below that of $x(t)$. Therefore, the reconstruction continues to satisfy all downward IC constraints. 

Now, to see that the reconstruction indeed yields $\text{Ran}(a)$ that consists of stochastically ordered lotteries, note that 
\[\text{Ran}(a) \subseteq \mathcal{A} :=  \Big\{a\in \Delta(\X^\star): a \in \Delta\big(\{x_{j-1}, x_j\}\big) \text{ for some $x_j \in \X^\star$}\Big\}\,,\]
which is totally ordered by $\preceq_{\text{st}}$.

\subsection{Generalized Downward Sufficiency Theorem}

Note that for the allocations on the frontier $\X^\star$, the agent's preferences satisfy strict increasing differences. Moreover, since the set of lotteries we reconstructed in \textbf{Step 1} can be totally ordered by stochastic dominance $\preceq_{\text{st}}$, the agent's preferences continue to satisfy strict increasing differences: Indeed, for any $a \preceq_{\text{st}} a'$ and $a \neq a'$, we can write  
\[v_t(a, t) = \E_{x\sim a}[v_t(x, t)] < \E_{x\sim a'}[v_t(x, t)] = v_t(a', t)\,,\]
where the strict inequality follows from the fact that $v_t(x_s, t)$ is strictly increasing in $s$ given the strict increasing differences on the frontier. Therefore, we may view the set of lotteries $\mathcal{A}$ used in the reconstruction step as a one-dimensional allocation space.  

A key technical result, which may be of independent interest, is the following downward sufficiency theorem that generalizes that in \citet{yang2022costly} to allow for any redistributive welfare weights:

\begin{theorem}[Downward Sufficiency]\label{thm:downward}
Let $\mathcal{A}$ and $\mathcal{T}$ be compact subsets of $\R$, and suppose $v(a, t)$ is nondecreasing in $t$, and is continuous and has strict increasing differences on $\mathcal{A} \times \mathcal{T}$. Then, for any type distribution and any redistributive welfare weights, any mechanism $(a, p): \mathcal{T} \rightarrow \mathcal{A} \times \R$ satisfying the downward IC and IR constraints can be weakly improved upon by a mechanism $(\tilde{a}, \tilde{p}): \mathcal{T} \rightarrow \mathcal{A} \times \R$ that satisfies \emph{all} IC and IR constraints.
\end{theorem}

The proof is in the appendix. We sketch the proof for the finite-type case here.

As discussed in \citet{yang2022costly}, the above result does not hold if one imposes only \textit{\textbf{local}} downward incentive constraints; however, perhaps surprisingly, the set of \textit{\textbf{global}} downward incentive constraints together is sufficient for determining the optimal solution in one-dimensional screening problems even with redistributive welfare weights.

Suppose $\mathcal{T} = \{t_1, \dots, t_n\}$. Let $a_i$ and $p_i$ denote the allocation and payment assigned to type $t_i$. Let $\mu(t_i)$ and $\lambda(t_i)$ denote the type distribution and welfare weights, respectively. Without loss, suppose that $\mu$ has full support. With finite types, our problem is: 
\[\max_{(a,\, p)}\sum_i \mu(t_i) \Big(\lambda(t_i) (v(a_i, t_i) - p_i) + p_i\Big)  \tag{\textbf{Downward-IC}}\]
subject to the downward IC constraints, $v(a_i, t_i) - p_i \geq v(a_j, t_i) - p_j$ for all $i > j$, and the IR constraints, $v(a_i, t_i) - p_i \geq 0$ for all $i$. The existence of a solution can be shown by compactness arguments (recall that $\E[\lambda(t)] \leq 1$ and $\lambda(t)$ is nonincreasing); see the appendix. As shown in \citet{yang2022costly}, every allocation rule $a \in \mathcal{A}^n$ is implementable by some transfers $p \in \R^n$ to be downward incentive compatible, i.e., feasible for the program (\textbf{Downward-IC}). Let IC$[k \rightarrow j]$ denote the IC constraint of type $t_k$ mimicking type $t_j$.  

We now prove \Cref{thm:downward} via two lemmas. 

\begin{lemma}\label{lem:mon}
For any type distribution $\mu$ and any redistributive welfare weights $\lambda$, there exists an optimal solution to (\emph{\textbf{Downward-IC}}) such that $a_1 \leq a_2 \leq \cdots \leq a_n$.
\end{lemma}

The proof is in the appendix. The intuition is as follows. Suppose for contradiction that no optimal allocation rule is monotone. Fix an optimal solution whose first violation of monotonicity, where adjacent types $t_i < t_{i+1}$ are assigned options with $a_{i} > a_{i+1}$, occurs the latest. Downward IC and increasing differences imply that either type could feasibly be given the other's assignment while preserving all downward IC constraints. Now, optimality implies that modifying type $t_{i+1}$'s assignment to $(a_i, p_i)$ is not an improvement. Thus, the potential revenue gain $p_i - p_{i+1}$ cannot justify the weighted welfare \textit{loss} of type $t_{i+1}$ (the modification reduces type $t_{i+1}$'s payoff due to downward IC). 
But then, by increasing differences, the \textit{\textbf{mirror-image modification}} that moves type $t_i$'s assignment down to $(a_{i+1}, p_{i+1})$ would give type $t_i$ a larger payoff gain than the payoff loss incurred by type $t_{i+1}$ under the first modification; moreover, by redistributiveness, this gain carries a weakly higher welfare weight and hence offsets the potential revenue loss, making the modification weakly improving. Now, iterating it then propagates the contract $(a_{i+1}, p_{i+1})$ leftward until the allocation rule becomes monotone up to type $t_{i+1}$, yielding an optimal solution whose first violation of monotonicity occurs strictly later---a contradiction. 

Now, the second lemma shows that, fixing a monotone optimal allocation rule from \Cref{lem:mon}, the payments can be chosen so that every local downward IC constraint binds. 

\begin{lemma}\label{lem:binding}
For any type distribution $\mu$ and any redistributive welfare weights $\lambda$, there exists an optimal solution to (\emph{\textbf{Downward-IC}}) with monotone $a$ and binding \emph{IC}$[(i+1) \rightarrow i]$ for all $i$.
\end{lemma}

The proof is in the appendix. If the objective is profit maximization only, then \Cref{lem:binding} is standard given \Cref{lem:mon}  because, given the monotonicity of $a$, any local downward IC constraint can be tightened while extracting weakly more revenue from
each type. With general welfare weights, however, raising payments can hurt: at any type with $\lambda(t_i) > 1$, the principal values the agent's payoff more than revenue. The proof is by a  \textit{\textbf{revenue-neutral perturbation}}. Fixing the monotone optimal allocation rule, consider, among the optimal payment rules, one that minimizes the sum of the slacks of all local downward IC constraints, and suppose for contradiction that some IC$[(i+1) \rightarrow i]$ remains slack. Then, given monotonicity of $a$ and increasing differences, every non-local downward IC constraint that crosses the pair $(i+1, i)$ must be slack by a uniform margin. Now, charge the types above $t_{i}$ slightly more and give the revenue back to all types as a lump sum. All downward IC constraints would be preserved under this revenue-neutral perturbation. Moreover, the net utility is shifted from higher types to lower types, which is a weak improvement by the redistributiveness of $\lambda$. The perturbed payment rule is thus also optimal. However, its total local slack is strictly smaller---a contradiction. 

Finally, by a standard argument, the mechanism $(a, p)$ from \Cref{lem:binding} must satisfy all upward IC constraints as well, concluding the proof sketch of \Cref{thm:downward}.

\subsection{Purification Lemma}

The final step is to purify the stochastic mechanism in \textbf{Step 2} to a deterministic mechanism. It turns out that whenever the stochastic mechanism uses ordered lotteries in the stochastic dominance sense and the principal has quasilinear preferences, there exists a purification. Formally, say that the principal has \textit{\textbf{generalized quasilinear preferences}} if the (ex post) objective value of assigning allocation $x$ to type $t$ with payment $p$ is given by 
$V(x, t) + W(t) \cdot p$ for some continuous $V$ and $W$. Clearly, our principal is a special case.

\begin{lemma}[Purification]\label{lem:purification}
Let $(\Y, \preceq)$ be any compact set of ordered allocations for which the agent's preferences satisfy strict increasing differences. Consider any IC and IR mechanism $(a, p): \mathcal{T}\rightarrow \mathcal{A} \times \R$ where $\mathcal{A} \subseteq \Delta(\Y)$ is totally ordered by $\preceq_{\emph{st}}$. Then, for any generalized quasilinear preferences of the principal, there exists a deterministic IC and IR mechanism $(x^\dagger, p^\dagger): \mathcal{T} \rightarrow \Y \times \R$ such that $(x^\dagger, p^\dagger)$ yields a weakly higher objective value than $(a, p)$. 
\end{lemma}

The proof is in the appendix. We sketch it here. By the Envelope theorem, the payments under a fully IC mechanism are pinned down by the allocation rule up to a constant. Then, since the principal has generalized quasilinear preferences, the objective is a \textit{\textbf{linear}} functional of the stochastic allocation rule (up to a constant). Since the lotteries in $\mathcal{A}$ are totally ordered by $\preceq_{\text{st}}$ and the agent's preferences have strict increasing differences, incentive compatibility implies that $a(\,\cdot\,)$ must be stochastically monotone as the type increases. By \citet{Strassen1965}, this stochastic allocation rule admits a \textit{\textbf{monotone coupling}}, i.e., it can be written as a mixture of deterministic monotone allocation rules. The objective is then equal to the \textit{\textbf{average}} of the objective values of a collection of deterministic monotone allocation rules. Since some realization must attain at least the average, taking that realized monotone allocation rule, which is implementable, then gives a weakly improving deterministic mechanism (with payments again given by the Envelope theorem).

\subsection{Proof Sketches for \Cref{thm:main2} and \Cref{thm:main3}}

\paragraph{Proof Sketch for \Cref{thm:main2}.}\hspace{-2mm}For the proof sketch, we assume a finite $\X$. Fix any stochastic mechanism $(a, p): \mathcal{T} \rightarrow \Delta(\X) \times \R$. Without loss, we can write $x(t, \varepsilon) \sim a(t)$ for all $t$, where $\varepsilon$ is a randomization device. We reconstruct an alternative stochastic mechanism $(\tilde{a}, \tilde{p}): \mathcal{T} \rightarrow \Delta(\X^\star) \times \R$ as follows. We keep the transfer rule unchanged. Fixing any $t$, for each $\varepsilon$, we construct $\tilde{a}(t, \varepsilon)$ from $x(t, \varepsilon)$ exactly as in \Cref{lem:reconstruct}, and then assign type $t$ the compound lottery resulting from the lottery $\tilde{a}(t, \varepsilon)$ and randomization $\varepsilon$. Since, for every realized $\varepsilon$, the reconstruction preserves each type's payoff from truthful reporting and weakly lowers each type's payoff from downward deviations, the mechanism $(\tilde a, \tilde p)$ must satisfy all downward IC constraints (averaging over $\varepsilon$). However, this stochastic allocation rule need not use stochastically ordered lotteries. Thus, we perform another round of modification as follows. For each type $t$, we replace $\tilde{a}(t)$ with a lottery $\hat{a}(t) \in \mathcal{A}$ supported on two adjacent elements of $\X^\star$ such that $v(\hat{a}(t), t) = v(\tilde{a}(t), t)$. Such a lottery exists by the intermediate value theorem. Note that the lottery $\hat{a}(t)$ is a \textit{\textbf{mean-preserving contraction}} of $\tilde{a}(t)$ in the \textit{\textbf{utility}} space of type $t$. Under the ordered incremental elasticity condition, the utility function of any higher type over the frontier elements is a \textit{\textbf{convex}} transformation of the utility function of type $t$. Thus, the replacement weakly decreases the deviation payoffs of all higher types; in particular, the demand curve generated by $\hat{a}$ single-crosses that generated by $\tilde{a}$ from below at the quantile of type $t$. Then, just like in the proof of \Cref{thm:main}, all downward IC constraints continue to hold under the modified mechanism $(\hat{a}, \tilde{p})$. Now, the newly modified mechanism uses lotteries in $\mathcal{A}$, which is totally ordered by $\preceq_{\text{st}}$. The rest of the proof follows identically to that of \Cref{thm:main}.

\paragraph{Proof Sketch for \Cref{thm:main3}.}\hspace{-2mm}To see why a generalized frontier is an optimal menu, fix any mechanism $(x, p)$ and fix any type $t$. By the definition of a covering, there exists some lottery $\hat{a}(t) \in \mathcal{A}$ whose demand curve weakly single-crosses that generated by the assigned allocation $x(t)$ from below at the quantile of type $t$---in particular, type $t$ is indifferent between $\hat{a}(t)$ and $x(t)$, whereas the deviation payoff of any higher type weakly decreases. Then, as in the proof of \Cref{thm:main}, the mechanism $(\hat{a}, p)$ is downward incentive compatible and individually rational, and it uses only lotteries that are stochastically ordered and supported on the elements in $\X^\star$. It follows immediately by \textbf{Step 2} and \textbf{Step 3} of the proof of \Cref{thm:main} that it can be further weakly improved by a deterministic mechanism $(\tilde{x}, \tilde{p})$ that is fully incentive compatible and uses only allocations in the generalized frontier $\X^\star$. Now, to see why the same ordered incremental elasticity condition implies stochastic optimality, note that just like in the proof of \Cref{thm:main2}, we can reconstruct a stochastic mechanism for each realization of the randomization device $\epsilon$ in the original mechanism---since at each type $t$, for each $\epsilon$, the reconstructed demand curve single-crosses the original demand curve at the \textit{\textbf{same}} crossing point (the quantile of type $t$), the pointwise mixture of the reconstructed demand curves must also single-cross the pointwise mixture of the original demand curves (mixing over $\epsilon$) at the quantile of type $t$. The rest then follows immediately by the proof of \Cref{thm:main2}.

\section{Applications}\label{sec:application}

\subsection{Nested Bundling: Robustness and Redistribution}\label{subsec:bundling}

Suppose that there are $n$ goods $\{1, \dots, n\} =: G$. A \textit{\textbf{bundle}} $b \in 2^{G}$ is a set of goods. Let $\mathcal{X} = 2^{G}$ be the allocation space. A  \textit{\textbf{menu}} $B \subseteq 2^{G}$ is then a set of bundles. A consumer's valuation $v(b, t)$ is monotone in $b$ in the set-inclusion order and strictly increasing in $t$. 

A menu is \textit{\textbf{nested}} if it can be totally ordered by set inclusion. As argued in \citet{yang2023nested}, offering a nested menu (\textit{\textbf{nested bundling}}) where more expensive bundles include all the goods of the less expensive ones appears to be common in practice.\footnote{Examples include tiered pricing for consumer AI services such as ChatGPT and Claude.} An immediate consequence of \Cref{thm:main} and \Cref{thm:main2} is the following new condition for when nested bundling is optimal: 
\begin{prop}[Robust Nesting]\label{prop:nesting}
For any nested menu $B$, if: 
\begin{itemize}
    \item[(i)] for any nonempty $b_1, b_2 \in B$ with $b_1 \subset b_2$
    \[\frac{\d}{\d t} \log v(b_1, t) < \frac{\d}{\d t} \log v(b_2, t) \text{ for all $t$} \]
    \item[(ii)] for any $b_1 \notin B$, there exists $b_2 \in B$ such that $b_1 \subset b_2$, and  
    \[\frac{\d}{\d t} \log v(b_1, t) \geq \frac{\d}{\d t} \log v(b_2, t) \text{ for all $t$} \]
\end{itemize}
then the menu $B$ is optimal for all $F$ and redistributive $\lambda$. If, in addition, we have: 
\begin{itemize}
    \item[(i')] Let $B = \{b_i\}$ where $\varnothing = b_0 \subset b_1  \subset \cdots \subset b_m$. For any $i = 1, \dots, m-1$, 
    \[\frac{\d}{\d t} \log \Big(v(b_{i}, t) - v(b_{i-1}, t)\Big) < \frac{\d}{\d t} \log\Big(v(b_{i+1}, t) - v(b_{i}, t)\Big)  \text{ for all $t$}\]
\end{itemize}
then the menu $B$ is optimal among stochastic mechanisms for all $F$ and redistributive $\lambda$.
\end{prop}
\begin{proof}
Under conditions \textit{(i)} and \textit{(ii)}, $B$ is a surplus-elasticity frontier and hence optimal by \Cref{thm:main}. If in addition we have condition \textit{(i')}, then $B$ is a strong frontier, and hence optimal among all stochastic mechanisms by \Cref{thm:main2}.    
\end{proof}

To illustrate \Cref{prop:nesting}, consider the following example: 
\begin{ex}\label{ex:bundling}
Suppose there are two goods $\{1, 2\}$. The bundle value is given by 
\[v(\{1, 2\}, t) = \kappa \cdot \Big(v(\{1\}, t) + v(\{2\}, t)\Big)\,,\]
where $\kappa > 0$ is a constant such that $v(\{1, 2\}, t) > \max\big\{v(\{1\}, t),\, v(\{2\}, t)\big\}$ for all $t$. The two goods are complements if $\kappa > 1$, substitutes if $\kappa < 1$, and additive if $\kappa = 1$. 
Suppose that good $2$ has a pointwise more elastic demand curve than good $1$, i.e., $\frac{\d}{\d t}\log v(\{2\}, t) < \frac{\d}{\d t}\log v(\{1\}, t)$ for all $t$. By the mediant inequality, we have that for all $t$ 
\[\frac{\d}{\d t}\log v(\{2\}, t) < \frac{\d}{\d t} \log v(\{1, 2\}, t) < \frac{\d}{\d t}\log v(\{1\}, t)\,.\]
Therefore, the item $\{1\}$ always has a less elastic demand curve than the bundle $\{1, 2\}$. It follows immediately by \Cref{prop:nesting} that nested bundling is optimal with an optimal menu given by $\X^\star = \big\{\{2\}, \{1, 2\}\big\}$, regardless of the parameter $\kappa$ (i.e., whether the two goods are complements or substitutes), regardless of type distributions $F$, and regardless of redistributive welfare weights $\lambda$. Moreover, it can be verified that $\X^\star = \big\{\{2\}, \{1, 2\}\big\}$ is a strong frontier, i.e., it satisfies condition $(i')$ of \Cref{prop:nesting}---therefore, $\X^\star =  \big\{\{2\}, \{1, 2\}\big\}$ is optimal even among all stochastic mechanisms.\footnote{In fact, by the mediant inequality, any surplus-elasticity frontier containing two elements (besides the outside option) is always a strong frontier and hence optimal among stochastic mechanisms by \Cref{thm:main2}.} \qed 
\end{ex}

The \textit{\textbf{robust nesting condition}} in \Cref{prop:nesting}, compared to the nesting condition derived in \citet{yang2023nested}, provides two additional conceptual insights.  

First, the robust nesting condition, by virtue of \Cref{thm:main}, shows that once there exists an appropriate frontier, the menu must be optimal for \textit{all} type distributions and \textit{all} redistributive welfare weights. Thus, unlike in \citet{yang2023nested} where the demand conditions depend on the type distributions and are specific to profit maximization, \Cref{prop:nesting} provides a clean decomposition between menu design and optimal pricing analogous to the one-good case of \citet{Myerson1981}: \textit{(i)} once the seller can identify the frontier from the demand system, she can already design an optimal menu just like in \Cref{ex:bundling}, and \textit{(ii)} as the seller obtains more information about the type distribution (e.g., via price experimentation), she can further refine the prices of different bundles in the menu. This is especially important in markets with new products---the sellers often lack detailed demand information, and often find it much easier to do price experimentation than menu experimentation.\footnote{For such markets with new products, \citet{pernoud2025bundling} provide a \textit{microfoundation} for why comonotonic type distributions as assumed by our model arise endogenously via consumer learning.}

Second, the robust nesting condition, as an immediate consequence of \Cref{thm:main}, also illuminates the nesting condition in \citet{yang2023nested} by organizing the key economic forces in the comparison of surplus and elasticity. \citet{yang2023nested} considers the partial order defined by \textit{(i)} set inclusion and \textit{(ii)} sold-alone quantities, and shows that if the undominated elements with respect to this partial order are nested, then this nested menu is optimal for profit maximization under quasi-concavity assumptions. \Cref{prop:nesting} shows that the dominance criterion can be broadly viewed as comparing demand curves in terms of their levels and their elasticities---whenever an appropriate frontier exists, it becomes optimal. This unifying perspective helps explain the perhaps surprising feature that the complementarity/substitutability patterns over the items for a given consumer are \textit{not} crucial for the optimality of nested bundling as shown in \Cref{ex:bundling}---these patterns are already taken into account by the surplus comparisons. It also helps explain why the sold-alone quantities are important in \citet{yang2023nested}: A higher sold-alone quantity corresponds to a demand curve with a larger elastic region; the seller only offers smaller bundles that are more elastic in order to target the lower end of the consumers.

\paragraph{Redistributive Bundling.}\hspace{-2mm}Another benefit of \Cref{prop:nesting}, relative to all of the bundling literature, is that it covers settings with redistributive preferences. To illustrate this, we consider a special case where $B = \{\overline{b}\}$ and $\overline{b} := \{1, \dots, n\}$, i.e., pure bundling. Then, \Cref{prop:nesting} immediately recovers the ratio-monotonicity condition of \citet{haghpanah2021pure} and shows that their condition actually gives the optimality of pure bundling even with redistributive concerns. More generally, we can relax the one-dimensional types here as in \citet{haghpanah2021pure}. Say that \textit{\textbf{stochastic ratio-monotonicity}} holds if $\frac{1}{v^{\overline{b}}} \mathbf{v}$ is stochastically nondecreasing in $v^{\overline{b}}$ for the random vector $\mathbf{v}:=(v^{b})_{b \in 2^{G}}$ of bundle values. Let the welfare weights and bundle values $(\lambda, \mathbf{v})$ be jointly distributed. 
\begin{cor}\label{cor:pure}
Suppose $v^{\overline{b}}$ is a sufficient statistic for redistribution in that $\E[\lambda \mid \mathbf{v}] = \E[\lambda \mid v^{\overline{b}}]$ which is continuous, nonincreasing in $v^{\overline{b}}$. Then, under stochastic ratio-monotonicity, we have: 
\begin{itemize}
    \item[(i)] Pure bundling is optimal among all stochastic mechanisms\,;
    \item[(ii)] The optimal mechanism satisfies the following comparative statics: Letting $\lambda(s) := \E[\lambda \mid q = s]$ where $q$ is the quantile of grand-bundle value, if the welfare weights are modified according to the following weak majorization order: $\int_{z}^1 \lambda(s) \d s \leq \int_{z}^1 \widetilde{\lambda}(s) \d s$ for all $z \in [0, 1]$,
    then the optimal mechanism lowers the grand-bundle price in the strong set order\,. 
\end{itemize}
\end{cor}

\Cref{cor:pure} broadens the scope to which \citet{haghpanah2021pure}'s condition applies and shows that in many redistributive contexts, it could be optimal to simply offer \textit{\textbf{broad access}} to resources as a single option without further screening of the recipients. This is perhaps surprising because the principal cares more about exactly the types who tend to have \textit{\textbf{lower}} values for the grand bundle, and yet the principal does not provide any smaller bundle for them to consume---by \Cref{cor:pure}, the only instrument the principal uses is to lower the price for the grand bundle and expand the market access to help as the welfare weights $\lambda$ increase pointwise.  Perhaps also surprisingly, fixing the average welfare weight $\E[\lambda]$, \Cref{cor:pure} asserts that the grand-bundle price \textit{\textbf{increases}} as the designer shifts the welfare weights from the high types to the low types. The reason is that giving more rents to the lower types under a pure-bundling mechanism necessarily implies more rents to the higher types given a single posted price, working against the redistributive objective. Instead, the designer raises the price for the all-access bundle to obtain a higher revenue at the expense of more exclusion. Nevertheless, the designer never introduces any smaller bundle for the poorer agents to consume, because the grand bundle itself is the surplus-elasticity frontier under the ratio-monotonicity condition.

\subsection{Optimal Taxation and Ordeals}\label{subsec:ordeal}

Consider the optimal taxation problem with quasilinear preferences (\citealt{diamond1998optimal}). In addition to the tax instrument, suppose that the government can also utilize ordeal mechanisms to induce self-targeting. As shown in \citet{Nichols1982b}, ordeals that purely destroy surplus can increase the total welfare achievable under nonlinear taxation by helping with screening. When should we expect an ordeal to be useful? We apply the main results to study this question. 

In this environment, an agent of type $t$ has the payoff given by $l - \psi(l, y, t) -  T$, where $l$ is the \textit{\textbf{labor income}}, $y$ is the \textit{\textbf{ordeal}}, and $T$ is the \textit{\textbf{tax}}. The function $\psi$ is the \textit{\textbf{disutility}} function, which we assume is strictly decreasing in the ability type $t$. The government maximizes the weighted welfare objective
\[\E[\lambda (t) (l(t) - \psi(l(t), y(t), t) -  T(t))]\]
subject to IC and the \textit{\textbf{budget balance}} constraint: $ \E\big[T(t)\big] \geq E$, where $E \geq 0$ is an exogenous public expenditure.  We assume that $\lambda(t)$ is redistributive in the sense that $\lambda(t)$ is weakly decreasing in the ability type $t$. 

To see how this problem is nested in our framework, note that for any IC mechanism, uniformly shifting the tax by a constant preserves IC.  Therefore, as a normalization, for any IC mechanism, we can shift the tax by a constant so that all types have weakly positive utility, i.e., the net utility $U(t) \geq 0$ for all $t$. Thus, we can equivalently optimize over IC and IR mechanisms with the objective being $\E[\lambda (t) U(t)] - \E[\lambda (t)] \cdot  \big(E - \E[T(t)]\big)$, where the additional term comes from uniformly shifting the tax back so that the budget constraint binds. Now, the above objective, after dropping the constant $\E[\lambda (t)] E$, is equivalent to $\E[\tilde{\lambda}(t) U(t)] + \E[T(t)]$,
where $\tilde{\lambda}(t) = \lambda(t) / \E[\lambda(t)]$. Thus, this is a special case of our main model with the average welfare weight equal to $1$. In this case, the allocation space $\mathcal{X} \subset \R^2_+$ consists of elements $(l, y)$, and we assume that $\{(l, 0)\}_{l \in [\underline{l}, \overline{l}]} \subset \mathcal{X}$ and $\psi(l, 0, t) \leq \psi(l, y, t)$ for all $(l, y) \in \X$, so $y = 0$ represents the absence of costly screening. We assume that $\psi(l, 0, t)$ is strictly submodular in $(l, t)$, continuous in $l$, and satisfies $\psi(0, 0, t) = 0$.

Our next result shows that there are two key tests that determine whether an ordeal should be used in the presence of tax instruments: \textit{(i)} it should sort agents in the right direction, and \textit{(ii)} even when the sorting direction is right, the ordeal should increase the screening effectiveness of the tax instrument. If either fails, nonlinear taxes alone are optimal. Passing the second test, as we discuss, is equivalent to increasing the take-up elasticity with respect to a small reduction in tax liability, in a sense made precise below. 

\begin{prop}\label{prop:ordeal}
Suppose that either of the following two conditions holds: 
\begin{itemize}
    \item[(a)] Either, $\psi(l, y, t) - \psi(l, 0, t)$ is nonincreasing in $t$; 
    \item[(b)] or, $l-\psi(l, 0, t)$ is log-supermodular in $(l, t)$ and $l - \psi(l, y, t) \geq 0$ with $\frac{l - \psi(l, y, t)}{l - \psi(l, 0, t)} $ nondecreasing in $t$ for all $(l, y) \in \mathcal{X}$. 
\end{itemize}
Then, $\X^\star = \{(l, 0)\}_{l\in [\underline{l}, \overline{l}]}$ is an optimal menu.
\end{prop}
\begin{proof}
Under condition (a), note that fixing any mechanism $t \mapsto (l(t), y(t), T(t))$, we can replicate the same truth-telling payoff by setting $\tilde{y}(t) = 0$ and $\widetilde{T}(t) = T(t) + \big(\psi(l(t), y(t), t) - \psi(l(t), 0, t)\big)$. Moreover, under condition (a), note that after this reconstruction, the deviation payoff of any type $t$ misreporting a lower type  $\hat{t} < t$ must weakly decrease. Thus, the modified mechanism must be downward incentive compatible (as defined in \Cref{sec:proof}). Then, by the generalized downward sufficiency theorem (\Cref{thm:downward}), it follows immediately that there exists a fully IC mechanism that involves no costly screening and weakly improves on the original mechanism, proving the result. Under condition (b), this is an immediate consequence of \Cref{prop:costly}. 
\end{proof}

To illustrate \Cref{prop:ordeal}, consider the following example: 

\begin{ex}
For the baseline utility function, consider the standard quasilinear specification with isoelastic disutility of labor $l - \frac{(l/t)^{1 + 1/\varepsilon}}{1 + 1/\varepsilon}$, where $\varepsilon > 0$ is a constant. Consider an additive ordeal $l - \frac{(l/t)^{1 + 1/\varepsilon}}{1 + 1/\varepsilon} - \frac{\alpha y}{t^\beta}$, 
where $\alpha, \beta > 0$ are constants. Then, $\psi(l, y, t) - \psi(l, 0, t) = \alpha y / t^\beta$ is nonincreasing in $t$. Here the sorting is incorrect; condition (a) of \Cref{prop:ordeal} applies, and nonlinear taxes alone are optimal.  Now, consider a multiplicative ordeal $\exp\big(-\frac{\alpha y}{t^\beta}\big)\big(l - \frac{(l/t)^{1 + 1/\varepsilon}}{1 + 1/\varepsilon}\big)$, where $\alpha, \beta > 0$ are constants. Here the ordeal lowers surplus but does not increase the screening effectiveness; in particular, condition (b) of \Cref{prop:ordeal} applies, so nonlinear taxes alone remain optimal. By contrast, under the reversed multiplicative distortion $\exp\big(-\alpha y \cdot t^\beta\big)\big(l - \frac{(l/t)^{1 + 1/\varepsilon}}{1 + 1/\varepsilon}\big)$, ordeals can be useful; indeed, one can construct a two-type example here in which nonlinear taxes alone are suboptimal. \qed 
\end{ex}

\paragraph{Take-up Elasticity.}\hspace{-2mm}The comparison in condition (b) of \Cref{prop:ordeal} can be equivalently stated using take-up elasticities with respect to a small tax reduction. Fix an income level $l$ and focus on the extensive margin, at which the marginal agents move from unemployment to 
earning $l$. First, absent an ordeal, consider a small uniform reduction in tax liability for all income at least $l$, which increases the fraction of the population reaching $l$. Second, repeat the thought experiment while requiring the ordeal $y$ for earning at least $l$. Formally, let
\[\eta^{\text{take-up}}(l, y, q) := -\frac{\d}{\d \log T} \log \P\big(l - \psi(l, y, t) \geq T\big) \Big|_{T = T^\star(l, y, q)}\,,\]
where $T^\star(l, y, q)$ is defined by $\P\big(l - \psi(l, y, t) \geq T^\star(l, y, q)\big) = q$. Then, note that
\[\frac{l - \psi(l, y, t)}{l - \psi(l, 0, t)} \text{ is nondecreasing in $t$} 
\iff 
\eta^{\text{take-up}}(l, y, q) \leq \eta^{\text{take-up}}(l, 0, q) \text{ for all $q$}\,,\]
i.e., starting from the same take-up rate, the same small proportional tax reduction 
produces a smaller extensive-margin response with the ordeal than without it. \Cref{prop:ordeal} shows that for ordeals to be \textit{useful}, we should expect two properties: \textit{(i)} the additional disutility resulting from an ordeal should be increasing in the ability type, and \textit{(ii)} the ordeal should also \textit{\textbf{increase}} the take-up elasticity. The first property is consistent with the takeaway from \citet{yang2022costly} that an ordeal is only useful if the agents' preferences over the ordeal and the productive labor supply are suitably negatively correlated. \Cref{prop:ordeal} shows that this conclusion holds even in the optimal taxation context. Moreover, \Cref{prop:ordeal} shows that, for an ordeal to be useful, it is not enough for the agents' preferences to be negatively correlated---the reduction in surplus must be compensated via an increase in the screening effectiveness, captured by an appropriate elasticity measure.\footnote{The elasticity derived here is related to, but distinct from, that in \citet*{yang2023comparison}: it concerns the extensive-margin labor response to a small tax cut. If the reform instead took the form of a \textit{cash transfer} conditional on completing the ordeal, a \textit{low} take-up elasticity with respect to the cash benefit would be helpful. These observations are consistent: such inelasticity arises when ordeal costs increase steeply with ability, which makes utility net of screening costs increase slowly with ability and hence makes the extensive-margin labor response elastic with respect to a tax cut. Importantly, unlike \citet*{yang2023comparison}, the analysis here allows monetary transfers  and ordeals to be used \textit{\textbf{jointly}} as screening instruments.}

\subsection{Sequential Screening}\label{subsec:sequential}

Consider the sequential screening problem introduced in \citet{Courty2000}. There is one good. The agent has an \textit{\textbf{ex ante type}} $t$, which is correlated with their \textit{\textbf{ex post type}} $\theta$, the agent's realized value from consuming the good. The contracting happens at the stage where the agent observes their ex ante type $t$ but does not observe their ex post type $\theta$. Let $g(\theta \mid t)$ denote the conditional density of the ex post type $\theta$ given the ex ante type $t$. As in \citet{Courty2000}, we assume that $\theta \mid t$ is \textit{\textbf{FOSD-increasing}} in $t$. By the revelation principle, the principal can design a \textit{\textbf{direct dynamic mechanism}} $(t, \theta) \mapsto \big(q(t, \theta), p(t, \theta) \big)$, where $q$ is the allocation probability and $p$ is the payment, satisfying: 
\begin{align*}
\theta q(t, \theta) - p(t, \theta) &\geq  \theta q(t, \hat{\theta}) - p(t, \hat{\theta}) &&\text{ for all $t, \theta, \hat{\theta}$\,;} \\ 
\int \Big( \theta q(t, \theta) 
- p(t, \theta) \Big) g(\theta \mid t)  \d \theta &\geq \int \Big( \theta  q(\hat{t}, \theta) 
- p(\hat{t}, \theta)\Big) g(\theta \mid t)  \d \theta &&\text{ for all $t, \hat{t}$\,;}   \\  
\int \Big( \theta q(t, \theta) 
- p(t, \theta) \Big) g(\theta \mid t)  \d \theta &\geq 0 &&\text{ for all $t$\,,}  
\end{align*}
where the first constraint is the second-period IC constraint, the second constraint is the first-period IC constraint, and the last constraint is the first-period IR constraint.

In practice, although consumers are often uncertain about their values, a common mechanism is to simply offer the good at a posted price $p$ with no refund. We call such a mechanism a \textit{\textbf{nonrefundable posted price}}. When is it optimal? Applying our main results, we show that nonrefundable posted prices are optimal whenever the consumers who have higher expected willingness to pay also tend to have greater uncertainty about their ex post values on the log scale. Recall that $\preceq_{\text{disp}}$ denotes the \textit{\textbf{dispersive order}}, i.e., $X \preceq_{\text{disp}} Y$ if $F^{-1}_X(q') - F^{-1}_X(q) \leq F^{-1}_Y(q') - F^{-1}_Y(q)$ for all $q < q' \in [0, 1]$. We say that the types have \textit{\textbf{increasing dispersion}} if for all $t < t'$ 
\[\log(\theta) \mid t \preceq_{\text{disp}} \log(\theta) \mid t' \,.\]

\begin{prop}\label{prop:sequential}
Under increasing dispersion, a nonrefundable posted price is optimal. 
\end{prop}
\begin{proof}
We formulate this problem in our framework in the following way. As in \citet{Eso2007}, note that we can write $\theta = G^{-1}(\varepsilon \mid t)$ where $\varepsilon$ is a uniform $[0, 1]$ random variable independent of $t$ and $G(\theta \mid t)$ is the conditional CDF. In particular, the model is equivalent to the one where the second-period type is indexed by $\varepsilon$ instead of $\theta$. Under this reparameterization, the ex post utility function is given by $G^{-1}(\varepsilon \mid t) q(t, \varepsilon) - p(t, \varepsilon)$. In this parameterization, note that the second-period IC implies that $q(t, \,\cdot\,)$ is monotone in $\varepsilon$. The first-period IC constraints would involve double deviations under this parameterization, but we keep the following subset of IC and IR constraints: 
\begin{align*}
\int G^{-1}(\varepsilon \mid t) q(t, \varepsilon) \d \varepsilon - \tilde{p}(t) &\geq \int G^{-1}(\varepsilon \mid t) q(\hat{t}, \varepsilon) \d \varepsilon - \tilde{p}(\hat{t}) \, && \text{ for all $t, \hat{t}$;}\\
\int G^{-1}(\varepsilon \mid t) q(t, \varepsilon) \d \varepsilon - \tilde{p}(t) &\geq 0 && \text{ for all $t$,}
\end{align*}
where $\tilde{p}(t) = \E_\varepsilon[p(t, \varepsilon)]$. We relax all other constraints and hence directly choose $(q, \tilde{p})$, keeping $q(t, \,\cdot\,)$ monotone, to solve the above problem. Now, note that since $q(t, \,\cdot\,)$ is monotone, up to a measure-zero set of $\varepsilon$'s, we can write $q(t, \varepsilon) = \int \1_{\varepsilon \geq 1-k} \mu_t(\d k)$, where $\mu_t \in \Delta([0, 1])$. Then, 
\[\int G^{-1}(\varepsilon \mid t) q(t, \varepsilon) \d \varepsilon = \int_0^1 \Big(\int^{1}_{1-k} G^{-1}(\varepsilon \mid t)\d \varepsilon\Big) \mu_t(\d k) \,.\]
Letting $u(k, t):= \int^{1}_{1 - k} G^{-1}(\varepsilon \mid t)\d \varepsilon$, we can write the above as \[\int_0^1 u(k, t) \mu_t(\d k)\,.\] Therefore, $\mu_t$ is equivalent to a lottery over the allocations $k \in [0, 1]$ with the utility function given by $u(k, t)$, which is strictly increasing in $k$ and weakly increasing in $t$ (since $\theta \mid t$ is FOSD-increasing in $t$). Thus, the relaxed problem is equivalent to a screening problem with allocation space $\X = [0, 1]$, utility function $u(k, t)$, and randomization allowed.\footnote{Alternatively, one can write the \textit{original} problem---rather than the relaxed problem---as a static screening problem with randomization, using ``strike prices'' as the pure allocations (\citealt{krahmer2017sequential}) instead of the ``$\varepsilon$-thresholds'' used here. However, that formulation does not directly fit our framework: for any given contract, the expected payment varies with the agent's \textit{true} type, which effectively introduces a type-dependent cost term into the principal's objective.} The result follows from \Cref{prop:stochastic} by noting that $\log u_{k}(k, t) = \log G^{-1}(1 - k \mid t)$ is weakly submodular if $\log(\theta) \mid t$ is weakly increasing in the dispersive order. Indeed, by the proof of \Cref{prop:stochastic}, in the weakly submodular case, $\X^\star := \{1\} \subseteq [0, 1]$ is a strong frontier and hence offering $k=1$ alone, at some price $p^\star$ (depending on the type distribution and welfare weights), is optimal in the relaxed problem, even if lotteries in $\Delta([0, 1])$ are allowed.\footnote{Formally, \Cref{prop:stochastic} assumes that the value function is differentiable and strictly increasing in $t$, but neither is needed in the weakly log-submodular case by inspecting the proof.} Offering a nonrefundable posted price $p^\star$ yields the same expected payoff for every type $t$ and the same revenue as in the relaxed problem---thus, it is optimal. 
\end{proof}

By virtue of our main results, \Cref{prop:sequential} holds for \textit{all} type distributions of $t$ and \textit{all} redistributive welfare weights. Even though sequential screening is a well-studied problem, we are not aware of any primitive condition under which a nonrefundable posted price is optimal. One reason is methodological: much of the dynamic mechanism design literature relies on the Myersonian approach that relaxes the global integral monotonicity constraint (\citealt*{Pavan2014}), and therefore mechanisms with significant bunching, such as a nonrefundable posted price, can be difficult to obtain without explicitly tracking the global incentive constraints.\footnote{For recent work that does confront global incentive constraints in dynamic settings, see \citet{battaglini2019optimal}, \citet{krasikov2026dynamic}, and \citet{li2025stochastic}.} In contrast, as shown in \Cref{sec:proof}, our proof method is non-Myersonian and deals with global incentive constraints.

\Cref{prop:sequential} asserts that nonrefundable ticket sales are optimal when higher ex ante types not only have higher expected willingness to pay (FOSD) but also have greater uncertainty about their ex post values in the sense of the log-dispersion order. Perhaps surprisingly, this holds even though the principal could instead offer refunds to better target high types and extract surplus through a higher up-front fee. Intuitively, the key friction is leakage: since ex post values are FOSD-increasing in ex ante types, a refund would attract lower types, leading to too many refunds and eroding revenue. Importantly, a refund with a high up-front fee actually \textit{\textbf{reduces}} the ex post surplus by inducing agents with lower ex post values to forgo consumption; by our main results, it can be justified only by increasing the screening effectiveness. Yet under the log-dispersion order---which captures the elasticity of how ex post values vary with ex ante types---the screening effectiveness cannot be increased; hence, a nonrefundable posted price is optimal. 

To illustrate our log-dispersion condition, consider the following example: 

\begin{ex}\label{ex:sequential}
Let $\varepsilon$ be some random variable independent of $t$ and 
\[\log \theta = \alpha \cdot t  + \sigma(t) \cdot \varepsilon\,, \tag{\textbf{Log-linear}}\]
where $\alpha > 0$ and $\sigma(t) \geq 0$. It can be easily verified that the types have increasing dispersion if $\sigma(t)$ is nondecreasing in $t$. 
As another example, consider  
\[\theta = \alpha \cdot t  + \sigma(t) \cdot \varepsilon \,, \tag{\textbf{Linear}}\]
where $\alpha > 0$ and $\sigma(t) \geq 0$.  It can be easily verified that the types have increasing dispersion if $\sigma(t)/t$ is nondecreasing in $t$ since the dispersion comparison in \Cref{prop:sequential} measures dispersion on the log scale.  We note that one of the most commonly used specifications in the dynamic mechanism design literature is the above linear form with $\sigma(t)$ equal to a constant---in that case, the types have \textit{\textbf{decreasing}} dispersion instead of \textit{\textbf{increasing}} dispersion on the log scale. In that case, it is well known that under a monotone hazard rate assumption on the distribution of $t$, the optimal mechanism is a continuum menu of option contracts (\citealt{Courty2000}). \qed 
\end{ex}
\begin{figure}[t]
    \centering
    \begin{minipage}{0.49\linewidth}
        \centering
        \includegraphics[width=0.82\textwidth]{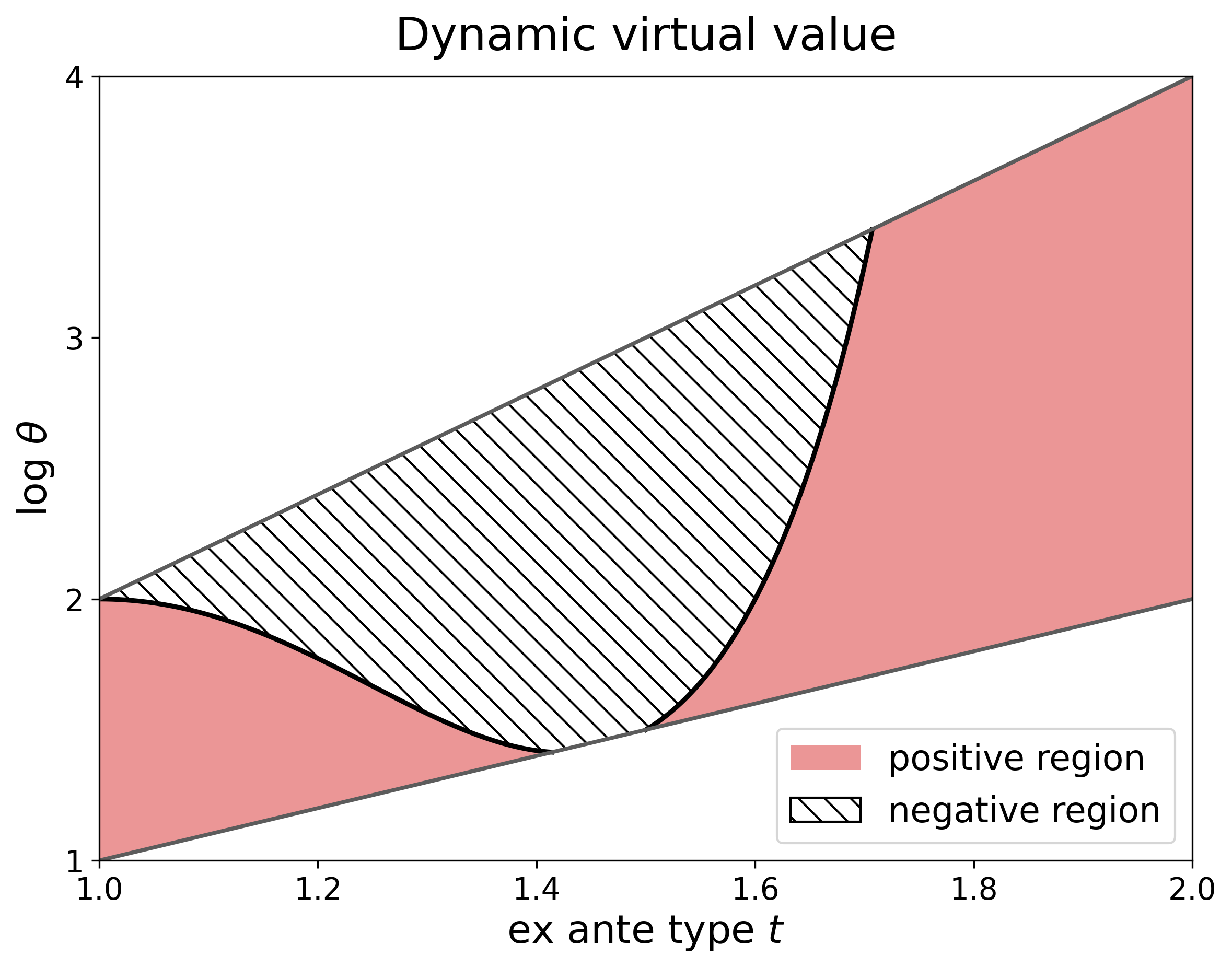}
        \subcaption{Sign of the Dynamic Virtual Value\label{fig:sequentiala}}
    \end{minipage}
    \hfill
    \begin{minipage}{0.49\linewidth}
        \centering
        \includegraphics[width=0.82\textwidth]{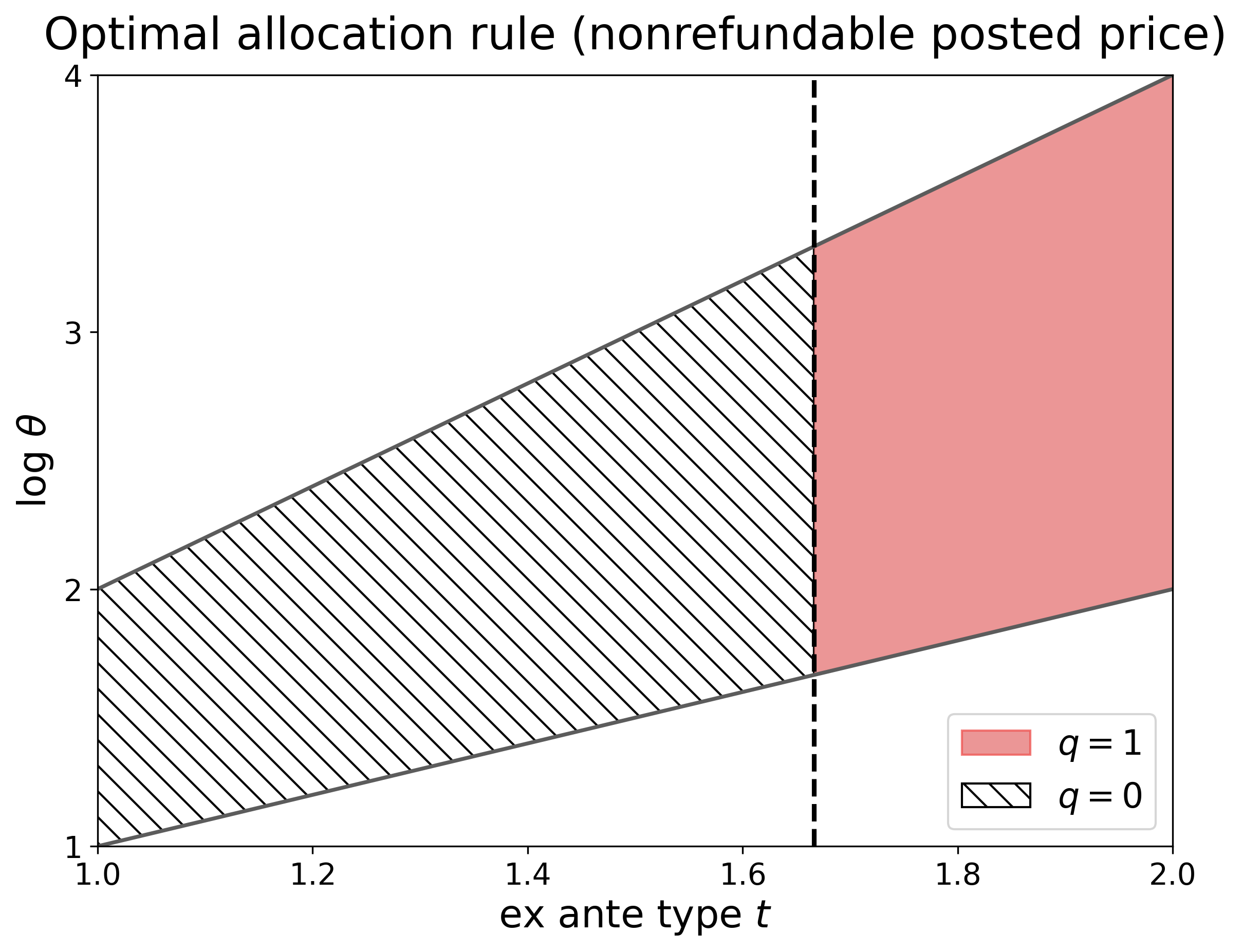}
        \subcaption{Optimal Allocation Rule\label{fig:sequentialb}}
    \end{minipage}
    \caption{Illustration of the Failure of the Myersonian Approach. Panel (a) depicts the sign of the dynamic virtual value for the log-linear example, which shows that the pointwise solution cannot be implemented. Panel (b) depicts the allocation rule under the optimal mechanism, which is a nonrefundable posted price by \Cref{prop:sequential}. \label{fig:sequential}}
\end{figure}
\vspace{-5mm}
\paragraph{Comparison to the Myersonian Approach.}\hspace{-2mm}To further illustrate how our approach differs from the Myersonian approach in the dynamic mechanism design literature, consider a log-linear specification with increasing dispersion in \Cref{ex:sequential}, $\log(\theta) = t + t \cdot \varepsilon$, where $\varepsilon \sim U[0, 1]$ and $t$ is supported on $[1, 2]$ with density $f(t) = 6(t-1.5)^2 + 0.5$. Suppose that the objective is profit maximization. Recall that the \textit{\textbf{dynamic virtual value}} is given by $\phi(t, \theta) := \theta + \frac{1 - F(t)}{f(t)} \frac{\frac{\partial}{\partial t}G(\theta \mid t)}{g(\theta \mid t)}$. The Myersonian approach maximizes the virtual surplus $\phi(t, \theta) q(t, \theta)$ pointwise, i.e., allocates exactly when $\phi(t, \theta) \geq 0$, and then checks whether the resulting allocation rule is implementable (\citealt*{Pavan2014}). \Cref{fig:sequentiala} depicts the sign of $\phi(t, \theta)$ and shows that the pointwise solution \textit{fails} the usual monotonicity conditions and cannot be implemented.\footnote{In fact, it fails both the second-period monotonicity and integral monotonicity conditions.} \Cref{fig:sequentialb} depicts the allocation rule under the optimal mechanism, which is a nonrefundable posted price by \Cref{prop:sequential} and has significant bunching. Despite the failure of the Myersonian approach and the absence of applicable ironing methods in this setting, our approach identifies the optimal mechanism using an appropriate screening frontier.

\subsection{Selling Information}\label{subsec:info}

Consider a finite state space $\Omega$.  As in \citet*{bergemann2018design}, the seller designs and sells \textit{\textbf{Blackwell experiments}} about the state. Here, however, the prior $\mu_0\in\Delta(\Omega)$ is commonly known and assumed uniform, while the buyer’s private information is their decision problem, indexed by type $t$. Without loss of generality, we identify an experiment with a \textit{\textbf{posterior distribution}} $\gamma \in \Delta(\Delta(\Omega))$  that satisfies Bayes plausibility $\int \mu \d \gamma (\mu) = \mu_0$; we identify a decision problem with a convex \textit{\textbf{indirect utility function}} of the posterior belief $u(\,\cdot\,, t): \Delta(\Omega) \rightarrow \R$. A type-$t$ buyer's value for an experiment $\gamma$ is then $v(\gamma, t) = \int u(\mu, t) \d \gamma (\mu) - u(\mu_0, t)$. Without loss of generality, we normalize $u(\mu_0, t) = 0$. 

We assume that types are vertically ordered in that higher types value every nontrivial experiment strictly more. We say that the environment is \textit{\textbf{symmetric}} if $u(\mu, t)$ is permutation symmetric with respect to $\mu$. In a symmetric environment with a uniform prior, vertical ordering is equivalent to $u(\mu, t)$ being strictly increasing in $t$ for every $\mu \neq \mu_0$.\footnote{In general, vertical ordering is weaker than the monotonicity of $u(\mu, \,\cdot\,)$---see \citet{whitmeyer2026making} for a characterization of when one decision problem has ``greater value for information'' than another.}

An experiment is a \textit{\textbf{truth-or-noise}} experiment if the signal reports the true state with probability $\alpha$ and otherwise is an independent draw from the prior $\mu_0$ (\citealt{lewis1994supplying}): formally, for any $\omega \in \Omega$ and $\alpha \in [0, 1]$, let $\mu^\omega_{\alpha} := (1 - \alpha) \mu_0 + \alpha \delta_\omega$ and $\gamma_\alpha := \sum_{\omega} \mu_0(\omega) \delta_{\mu^\omega_\alpha}$. The truth-or-noise experiments are very simple: indexed by the truth probability $\alpha$, Blackwell-monotone in $\alpha$, and sparsely supported. When are they optimal?

We now identify the frontiers in this screening problem in two cases. It turns out that, in a symmetric environment, \textit{(i)} if a higher type has a decision problem that depends ``less sensitively'' on the state, then selling full information is optimal; \textit{(ii)} in contrast, if a higher type has a decision problem that depends ``more sensitively'' on the state, then a menu of truth-or-noise signals is optimal. More precisely, we compare how convex the indirect utility functions are in the Arrow--Pratt sense: $u(\,\cdot\,, t)$ is \textit{\textbf{increasingly convex}} in $t$ if, for any $ t< t'$, there exists a nondecreasing convex function $g$ such that $u(\,\cdot\,, t') = g \circ u(\,\cdot\,, t)$, and  \textit{\textbf{strictly increasingly convex}} if $g$ is strictly convex; similarly, $u(\,\cdot\,, t)$ is \textit{\textbf{increasingly concave}} in $t$ if the same holds with $g$ nondecreasing and concave.

\begin{prop}\label{prop:info}
Consider a symmetric environment. If $\text{$u(\,\cdot\,, t)$ is increasingly \emph{concave} in $t$}$, then selling full information is optimal. If $\text{$u(\,\cdot\,, t)$ is strictly increasingly \emph{convex} in $t$}$, then selling truth-or-noise experiments $\{\gamma_\alpha\}_{\alpha\in[0, 1]}$ is optimal. 
\end{prop}

\begin{proof} 
Let $\X$ be the set of all Bayes-plausible posterior distributions. To prove the first part, suppose $u$ is increasingly concave in $t$. We show that the fully informative experiment $\{\gamma^\star\}$ is the surplus-elasticity frontier. Clearly, it generates a demand curve that is pointwise maximal. Now we check whether the demand curve it generates is more elastic compared to the one generated by any other nontrivial experiment $\gamma$.  Fix any two types $t_1 < t_2$. Then, 
$u(\mu, t_2) = g(u(\mu, t_1))$ for a concave function $g$. Moreover, since $u(\mu_0, t) = 0$ by our normalization, $g(0) = 0$. Together, these imply that $g(z)/z$ is nonincreasing in $z$ for all $z > 0$.  By symmetry, $u(\delta_\omega, t) = u(\delta_{\omega'}, t)$ for all $\omega, \omega'$. Therefore, the fully informative experiment yields a deterministic payoff for type $t_1$, say $u^\star_1$. Clearly, $u^\star_1 \geq u(\mu, t_1)$ for all $\mu$. Then, by the previous observation, $(g(u^\star_1) / u^{\star}_1) \cdot  u(\mu, t_1) \leq g(u(\mu, t_1))$.\footnote{Note that by symmetry and Jensen's inequality, we also have $u(\mu, t) \geq 0$ for all $\mu$ and $t$.} Thus, 
\[\frac{v(\gamma, t_2)}{v(\gamma, t_1)} =  \frac{\int g(u(\mu, t_1)) \d \gamma }{\int u(\mu, t_1) \d \gamma}  \geq  \frac{g(u^\star_1) /u^\star_1 \int u(\mu, t_1) \d \gamma }{\int u(\mu, t_1) \d \gamma} = \frac{g(u^\star_1)}{u^\star_1} = \frac{v(\gamma^\star, t_2)}{v(\gamma^\star, t_1)}\,,\]
which proves the result by \Cref{thm:main}.  

To prove the second part, suppose $u$ is strictly increasingly convex in $t$. We show that $\X^\star := \{\gamma_\alpha\}$ is a generalized frontier. Note that for any truth-or-noise experiment $\gamma_\alpha$,  we have by construction $v(\gamma_\alpha, t) =\sum_{\omega}\mu_0(\omega) u\big( (1 - \alpha) \mu_0 + \alpha \delta_\omega, t\big)$. By symmetry, for any $\omega, \omega'$, 
\[u\big( (1 - \alpha) \mu_0 + \alpha \delta_\omega, t\big) = u\big( (1 - \alpha) \mu_0 + \alpha \delta_{\omega'}, t\big) =: h(\alpha, t)\,.\]
Note that $h(\alpha, t)$ is strictly increasing in $\alpha$, and in $t$ for $\alpha > 0$, with strict increasing differences: for $s > t$, there exists $g$ with $g(z) \geq z$ and $g(z)/z$ strictly increasing such that 
\begin{align*}
  h(\alpha, s) - h(\alpha, t) &= \big(g(h(\alpha, t))/h(\alpha, t) - 1\big) \cdot h(\alpha, t)   \\
  &< \big(g(h(\alpha', t))/h(\alpha', t) - 1\big) \cdot h(\alpha', t) = h(\alpha', s) - h(\alpha', t)  
\end{align*}
for any $\alpha' > \alpha > 0$. Now, fix any $\gamma$ and fix any type $t$. By the intermediate value theorem, there exists some $\alpha^\star$ such that $v(\gamma_{\alpha^\star}, t) = h(\alpha^\star,t) = \int u(\mu, t) \d \gamma(\mu)$. Then, note that for any $s > t$, since $u$ is increasingly convex, there exists some convex $g$ such that 
\[v(\gamma_{\alpha^\star}, s) = h(\alpha^\star, s) = g(h(\alpha^\star, t))\leq \int g(u(\mu, t)) \d \gamma(\mu) = v(\gamma, s)\,, \]
where the inequality uses that $g$ is convex. Clearly, the opposite inequality holds for $s < t$. Since this holds for all $t$, it implies that $\X^\star$ is a generalized frontier. Indeed, any $\gamma$ is covered by $\mathcal{A}$ defined as the collection of Dirac measures supported on each of $\{\gamma_\alpha\}_\alpha$. This completes the proof by \Cref{thm:main3}. 
\end{proof}

As the proof of \Cref{prop:info} shows, in the increasingly concave case, full information both maximizes surplus and generates the most elastic demand curve, so the frontier is a singleton; in the increasingly convex case, the truth-or-noise signals span a generalized frontier, where a higher truth probability creates more surplus but a less elastic demand curve. To illustrate \Cref{prop:info}, consider the following example: 
\begin{ex}\label{ex:info}
Fix any finite state space $\Omega$. There are two primitive decision problems represented by two symmetric convex functions $u_A(\mu)$ and $u_B(\mu)$, with $u_A(\mu_0) = u_B(\mu_0) = 0$. Suppose $u_B \geq u_A$ (and strictly so for $\mu \neq \mu_0$), i.e., problem $B$ has a higher stake than $A$. The agent will face one of the two decision problems and privately knows the probability $t \in [0, 1]$ of facing the higher-stake problem $B$, independently of the state. So a type-$t$ agent's indirect utility is $u(\mu, t) = (1 - t) \cdot u_A(\mu) + t \cdot  u_B(\mu)$. Consider two cases: 
\begin{itemize}
    \item[$(a)$] $u_B$ is a concave transformation of $u_A$. Then $u(\,\cdot\,,t)$ is increasingly concave in $t$. By \Cref{prop:info}, selling full information is optimal. To illustrate, consider $|\Omega| = 2$ and so $\mu \in [0, 1]$, with $u_A = |2\mu - 1|^2$ and $u_B = 2 \cdot |2\mu - 1|$. \Cref{fig:infoa} depicts $u_A$ and $u_B$. 
    \item[$(b)$] $u_B$ is a strictly convex transformation of $u_A$. Then $u(\,\cdot\,,t)$ is strictly increasingly convex in $t$. By \Cref{prop:info}, selling a menu of truth-or-noise signals is optimal. To illustrate, consider $|\Omega| = 2$ and so $\mu \in [0, 1]$, with $u_A = |2\mu - 1|$ and $u_B = |2\mu - 1| + |2\mu - 1|^2$. \Cref{fig:infob} depicts $u_A$ and $u_B$. More concretely, suppose the objective is profit maximization and the type distribution is uniform. Then, the optimal mechanism uses the continuum of truth-or-noise signals $\{\gamma_\alpha\}_{\alpha\in[0.5,\, 1]}$, with allocation rule
    \[\alpha^\star(t) = \begin{cases}
       \frac{1}{2(1-2t)} &\text{ if $t \in [0,\, 0.25]$}\\
       1 &\text{ if $t \in (0.25,\, 1]$\,,}
    \end{cases}\]
    where $\alpha^\star(t)$ is the truth probability of the signal for type $t$. \qed
  \end{itemize}
\end{ex}

\begin{figure}[t]
    \centering
    \begin{minipage}{0.49\linewidth}
        \centering
        \includegraphics[width=0.9\textwidth]{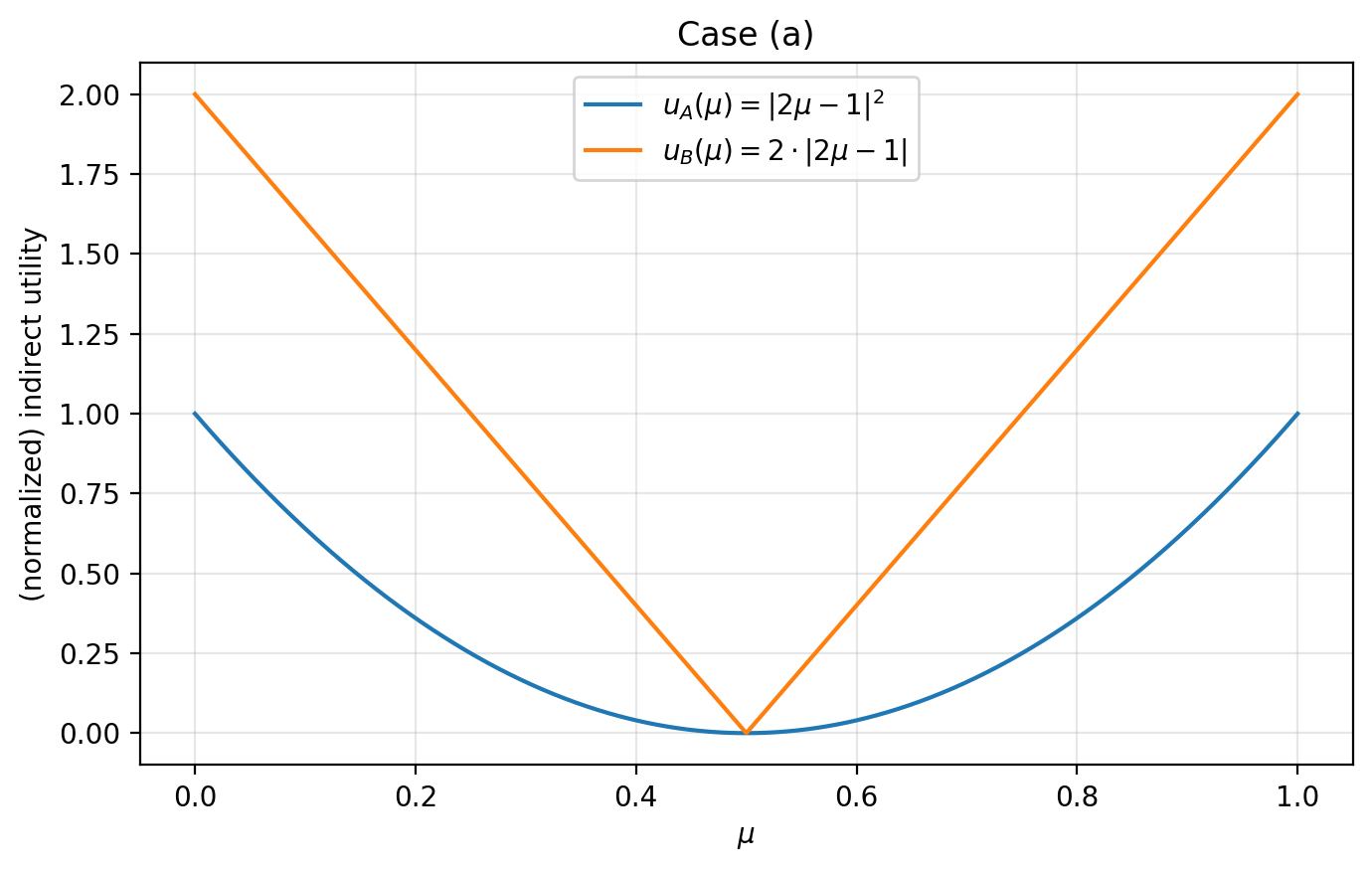}
        \subcaption{Higher-stake problem is less sensitive\label{fig:infoa}}
    \end{minipage}
    \hfill
    \begin{minipage}{0.49\linewidth}
        \centering
        \includegraphics[width=0.9\textwidth]{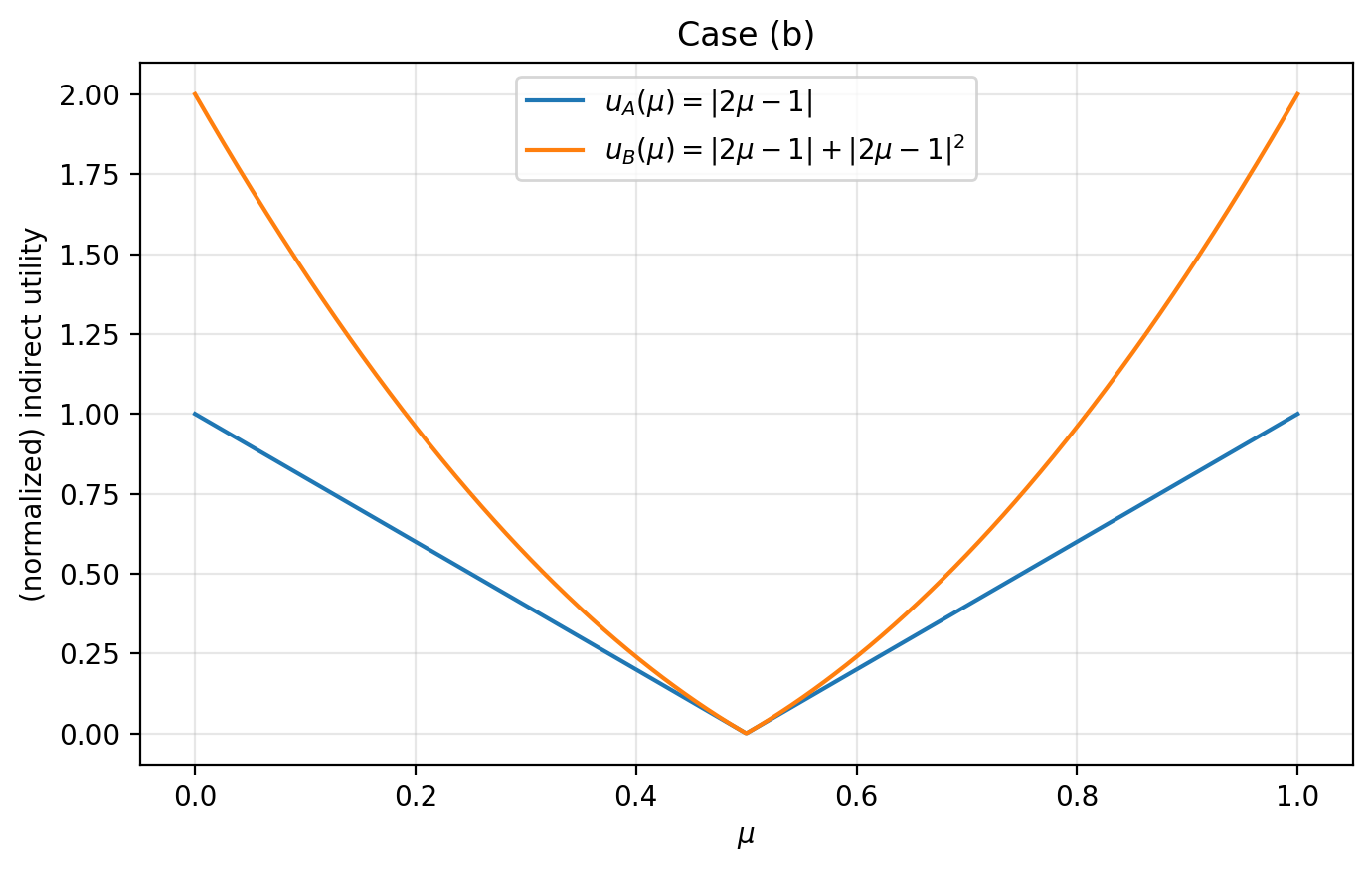}
    \subcaption{Higher-stake problem is more sensitive\label{fig:infob}}
    \end{minipage}
    \caption{Illustration of Cases (a) and (b) in \Cref{ex:info}. In Case (a), the higher-stake problem depends less sensitively on the state and hence selling full information is optimal by \Cref{prop:info}; in Case (b), the higher-stake problem depends more sensitively on the state and hence selling a menu of truth-or-noise signals is optimal by \Cref{prop:info}.\label{fig:info}}
\end{figure}

As illustrated in \Cref{ex:info}, the optimal pricing and allocation rule of the truth-or-noise signals depend on the details of the environment, but the same menu is optimal for \textit{all} type distributions and redistributive welfare weights.

\subsection{Regulating a Data-Rich Monopolist}\label{subsec:regulation}

Our last application considers a model of monopoly regulation just like in \citet{Baron1982RegulatingCosts}, but allows the monopolist to observe rich consumer data. 

There is a monopolist and a unit mass of consumers. Consumers are indexed by their \textit{\textbf{values}} $\theta \in [\underline{\theta}, \overline{\theta}]$ for the monopolist's good, distributed according to $G$. The monopolist has a private \textit{\textbf{cost type}} $t \in [\underline{t}, \overline{t}]$, distributed according to $F$, independent of $\theta$. The cost for the monopolist to serve consumer type $\theta$ can be \textit{\textbf{interdependent}}, denoted by $c(\theta, t)$. We assume $c(\theta, t)$ is strictly decreasing in $t$ so that a  higher $t$ represents a lower cost type. Following \citet{strack2025non}, we model the  \textit{\textbf{data-rich monopolist}} as observing the consumers' valuations $\theta$ and engaging in personalized pricing in the absence of any regulation; in practice, $\theta$ can be thought of as a rich enough tag such that the residual uncertainty over consumer valuations is small conditional on the tag. In the absence of any regulation, the monopolist would extract the \textit{\textbf{total surplus}} $S(\theta, t):= \theta - c(\theta, t)$,
which we assume is positive and continuously differentiable. As in \citet{strack2025non}, the monopolist can use any \textit{\textbf{pricing rule}} $p(\theta, r)$ that specifies a posted price for consumer $\theta$ depending on the realization of a randomization device $r$. Given a pricing rule $p$, the monopolist's profit is $\Pi(p, t) := \E_{\theta, r}\big[\big(p(\theta, r) - c(\theta, t) \big) \1_{\theta \geq p(\theta, r)} \big]$, and the consumer $\theta$'s surplus is 
$\text{CS}(p, \theta) := \E_r\big[\big(\theta - p(\theta, r)\big)_+\big]$. Let $\mathcal{P}$ denote the set of all pricing rules. 

A regulator contracts with the monopolist on the pricing rules. The regulator does not observe the monopolist's cost type $t$ and therefore posts a menu of pricing rules and associated subsidies/taxes $\{(p, T)\}$. By the revelation principle, without loss of generality, the regulator uses a  \textit{\textbf{(direct, incentive-compatible) mechanism}} $(p, T): [\underline{t}, \overline{t}] \rightarrow \mathcal{P} \times \R$ that satisfies the IC and IR constraints of the monopolist
\begin{align*}
 \Pi(p(t), t) - T(t)  &\geq  \Pi(p(\hat{t}), t) - T(\hat{t})  &&\text{ for all $t, \hat{t}$\,;} \\
 \Pi(p(t), t) - T(t)  &\geq  0  &&\text{ for all $t$\,.} 
\end{align*}
Following \citet{Baron1982RegulatingCosts}, the regulator maximizes a weighted sum of consumer and producer surplus: $\E[\text{CS}(p(t), \theta) + T(t)] + \alpha \cdot \E[\Pi(p(t), t) - T(t)]$,
where $\alpha \in [0, 1]$. This objective balances the budget as in \citet{Baron1982RegulatingCosts} since any tax imposed on the monopolist is rebated to the consumers and vice versa. 

We say that a mechanism is a \textit{\textbf{uniform pricing regulation}} if the assigned pricing rules $p(\theta, r; t)$ do not depend on any consumer data $\theta$ or randomization device $r$ for any $t$. That is, the regulator simply imposes uniform pricing and taxes/subsidizes the monopolist depending on the uniform price. Note that if the regulator were to use a uniform pricing regulation, then the determination of the optimal taxes/subsidies exactly collapses to the model of \citet{Baron1982RegulatingCosts}. However, when are uniform pricing regulations optimal? Our next result answers this question. Before stating it, we introduce another class of mechanisms. We say that a mechanism is a \textit{\textbf{discriminatory pricing regulation}} if the assigned pricing rules are deterministic and specify either full price discrimination or no trade, i.e., $p(\theta, r; t) \in \big\{\theta,\,\, \overline{\theta} + 1\big\}$, for all types $t$. In this case, the consumer surplus is entirely generated via lump-sum transfers by the regulator. 

Our next result shows that, in the presence of tax instruments, the optimal regulation of consumer data depends on the log modularity of the total surplus function, which, as we discuss, captures the severity of adverse selection in the market.
\begin{prop}\label{prop:regulation}
If $S(\theta, t)$ is log-submodular, then uniform pricing regulations are optimal. If $S(\theta, t)$ is log-supermodular, then discriminatory pricing regulations are optimal. 
\end{prop}
\begin{proof}
Fix any mechanism $(p, T)$. Let $\widetilde{T}(t) = \E_{\theta}[\text{CS}(p(t), \theta)] + T(t)$. 
Then, note that the mechanism using the modified transfer rule satisfies all IC and IR constraints with respect to the monopolist's modified utility that takes into account the consumer surplus: $\Pi(p, t) + \E_{\theta}[\text{CS}(p, \theta)] - \widetilde{T}$. After this relabeling, the objective then becomes 
\[\E\Big[\text{CS}(p(t), \theta) + T(t)\Big] + \alpha \cdot \E\Big[\Pi(p(t), t) - T(t)\Big] = \E\Big[\widetilde{T}(t)\Big] + \alpha \cdot \E\Big[\Pi(p(t), t) + \E_{\theta}[\text{CS}(p(t), \theta)]  - \widetilde{T}(t)\Big]\,,\]
which is a special case of the objective in our main model with weight $\alpha \in [0, 1]$ on the agent's net utility (i.e., the monopolist's modified utility). Now, note that 
\[\Pi(p, t) + \E_{\theta}[\text{CS}(p, \theta)]  = \int \big(\theta - c(\theta, t) \big) \E_r[\1_{\theta \geq p(\theta, r)}]  \d G(\theta) =\int S(\theta, t) x(\theta)  \d G(\theta)\,,\]
where $x:[\underline{\theta}, \overline{\theta}] \rightarrow [0, 1]$ is defined by $x(\theta) := \E_r[\1_{\theta \geq p(\theta, r)}]$. Moreover, conversely, any $x:[\underline{\theta}, \overline{\theta}] \rightarrow [0, 1]$ can be written as $\E_r[\1_{\theta \geq p(\theta, r)}]$ for some pricing rule $p$. Therefore, by \Cref{prop:function}, the menu $\X^\star =\big \{\1_{[k, \overline{\theta}]}\big\}_{k}$ is optimal if $S$ is log-submodular, and $\X^\star =\big \{\1_{[\underline{\theta}, k]}\big\}_{k}$ is optimal if $S$ is log-supermodular. The first part of the result follows immediately by noting that in the log-submodular case, the menu is equivalent to a uniform pricing regulation (with the $k$'s being the uniform prices). Now, to see the second part, let $k(t)$ be the assignment for type $t$ given the optimal prices for the menu $\X^\star =\big \{\1_{[\underline{\theta}, k]}\big\}_{k}$, and write $p(\theta, r; t) = \1_{\theta \leq k(t)} \theta + \1_{\theta > k(t)} (\overline{\theta} + 1)$, which then gives a discriminatory pricing regulation that implements the menu $\X^\star$.
\end{proof}

To illustrate \Cref{prop:regulation}, consider the following example:
\vspace{-2mm}
\begin{ex}\label{ex:regulation}
Suppose $c(\theta, t) = h(\theta)w(t)$. Then, $\partial_{\theta t} \log S(\theta, t) =\frac{\theta h'(\theta) - h(\theta)}{(\theta - h(\theta) w(t))^2} \cdot (-w'(t))$\,.
Since $-w'(t) > 0$, the above expression has the same sign as $\theta h'(\theta) - h(\theta)$, i.e., whether $h(\theta)/\theta$ is increasing or decreasing in $\theta$. By \Cref{prop:regulation}, if $h(\theta)/\theta$ is decreasing in $\theta$, then uniform pricing regulations are optimal; if $h(\theta)/\theta$ is increasing in $\theta$, then discriminatory pricing regulations are optimal---regardless of the type distribution $F$, the value distribution $G$, and the welfare weight $\alpha$. Moreover, note that $h(\theta)/\theta$ would be decreasing---and hence uniform pricing regulations would be optimal---if either $h(\theta)$ is decreasing or $h(\theta)$ does not increase as fast as $\theta$. 
\qed
\end{ex}

\vspace{-6mm}
\paragraph{Data Regulation and Severity of Adverse Selection.}\hspace{-2mm}As \Cref{ex:regulation} illustrates, \Cref{prop:regulation} shows that uniform pricing regulations are optimal in markets with either \textit{\textbf{advantageous selection}} (where $h(\theta)$ is decreasing) or \textit{\textbf{mild adverse selection}} (where $h(\theta)$ is increasing but $h(\theta)/\theta$ is decreasing). In such markets, consumer data has no regulatory value: the optimal regulation coincides with the classic analysis of \citet{Baron1982RegulatingCosts}, even though the monopolist could personalize prices. In markets with \textit{\textbf{severe adverse selection}} (where $h(\theta)/\theta$ is increasing), uniform pricing regulations generally fail to be optimal: in the strict log-supermodular case, the optimal pricing rules identified by \Cref{prop:regulation}, if nondegenerate, serve low-value consumers while excluding high-value ones, and hence cannot be implemented via any uniform pricing regulation.\footnote{In contrast, for any uniform pricing regulation, there exists an outcome-equivalent discriminatory pricing regulation. However, without the conditions identified in \Cref{prop:regulation}, even discriminatory pricing regulations need not be optimal since the optimum can require randomized pricing rules.}

More concretely, consider \Cref{ex:regulation} with $w(t) = 1 - 0.15t + 0.1t^2 - 0.75t^3$, $t \sim U[0, 1]$, $\theta \sim U[0, 1]$, and $\alpha = 0.4$.  First, consider $h(\theta)= \theta - 0.1\theta^2$: since $h(\theta)/\theta$ is decreasing, uniform pricing regulation is optimal by \Cref{prop:regulation}; \Cref{fig:regulationa} depicts the optimal allocation rule. Second, consider $h(\theta)= \theta^{1.05}$: since $h(\theta)/\theta$ is increasing, discriminatory pricing regulation is optimal by \Cref{prop:regulation}; \Cref{fig:regulationb} depicts the optimal allocation rule. In the second case, even though the regulator could impose uniform pricing and subsidize the monopolist to prevent unraveling (as the surplus is always positive), doing so is costly precisely because the supermodularity of surplus is at the log level, capturing how the monopolist's information rents per unit of surplus rise with consumer value. Since a uniform price necessarily serves the highest-value consumers, the regulator instead lets the monopolist utilize the consumer data to trade with the low-value, low-cost consumers, curtailing the monopolist's information rents at minimal surplus cost.\footnote{In fact, if $S(\theta, t)$ is log-supermodular and also decreasing in $\theta$, then the optimal discriminatory pricing regulation can be implemented by a nonlinear tax schedule depending only on \textit{total trade volume}, allowing the monopolist to fully utilize consumer data to both price discriminate and decide whom to serve.}

\begin{figure}[t]
    \centering
    \begin{minipage}{0.49\linewidth}
        \centering
        \includegraphics[width=0.82\textwidth]{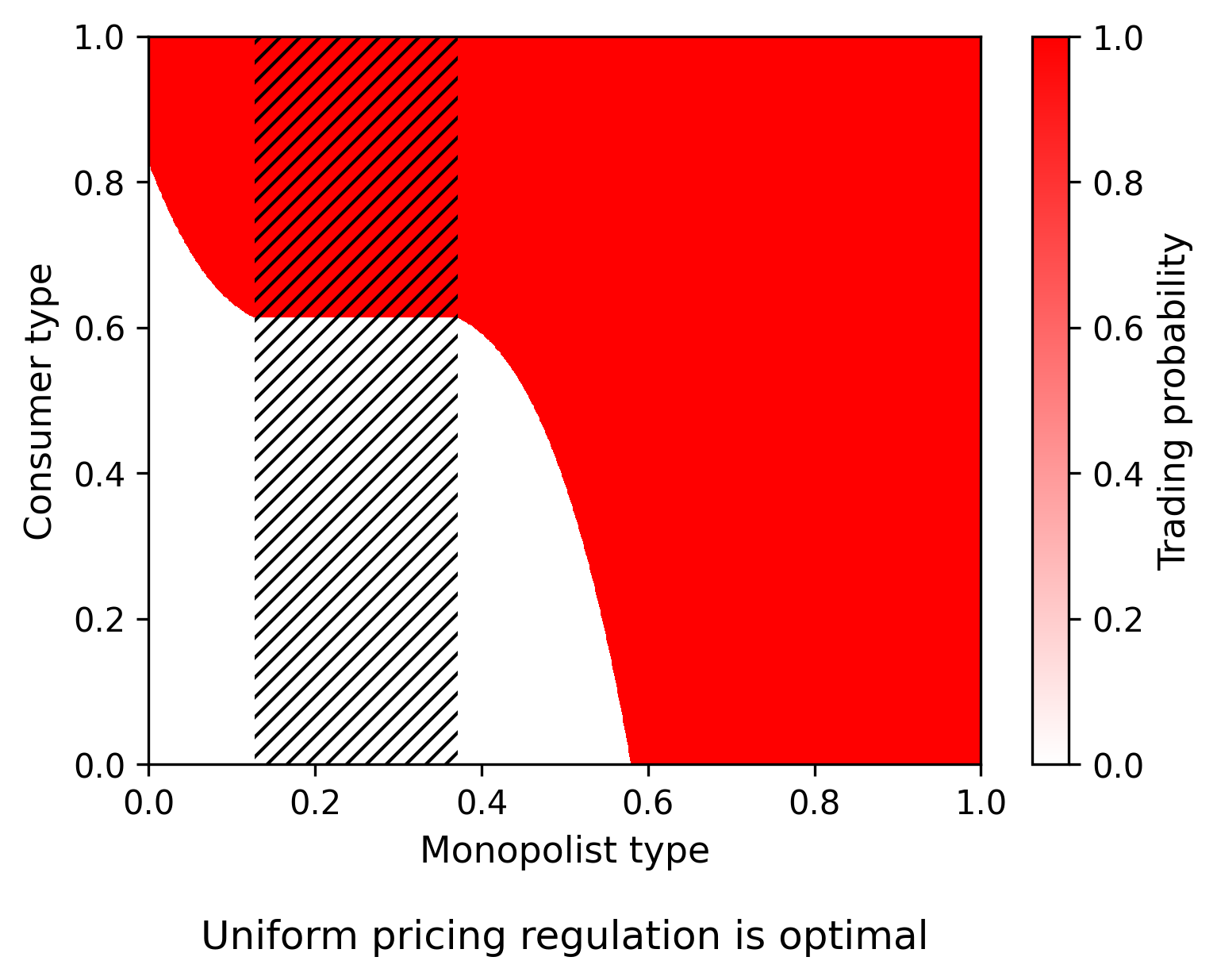}
        \subcaption{Uniform pricing regulation is optimal\label{fig:regulationa}}
    \end{minipage}
    \hfill
    \begin{minipage}{0.49\linewidth}
        \centering
        \includegraphics[width=0.82\textwidth]{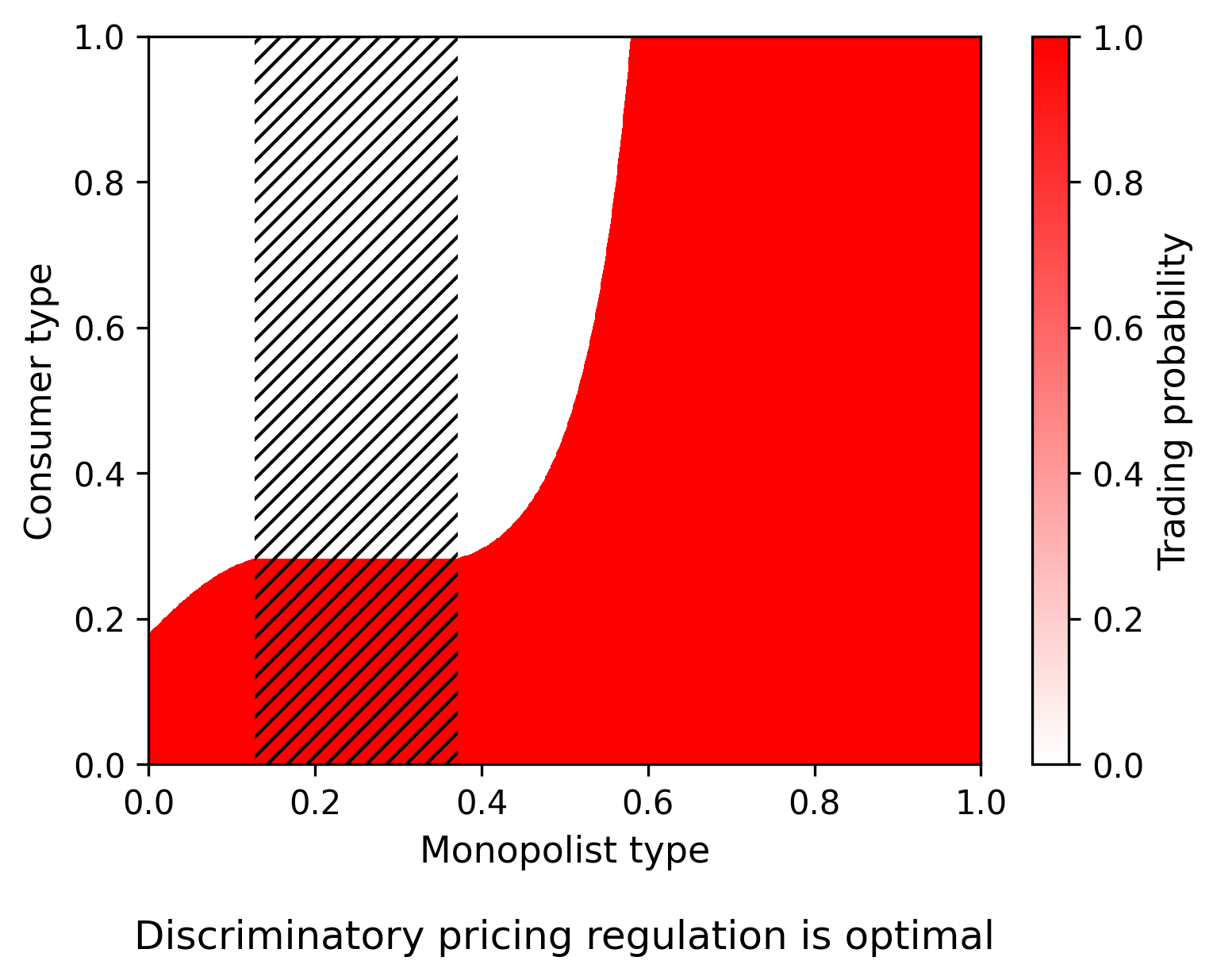}
       \subcaption{Discriminatory pricing regulation is optimal\label{fig:regulationb}}
    \end{minipage}
     \caption{Illustration of Optimal Allocation Rules in \Cref{ex:regulation}. Panel (a) depicts a \textit{mild adverse selection} case ($h(\theta)= \theta - 0.1\theta^2$), where uniform pricing regulation is optimal; panel (b) depicts a \textit{severe adverse selection} case ($h(\theta)= \theta^{1.05}$), where discriminatory pricing regulation is optimal (\Cref{prop:regulation}). The colored regions depict the seller and buyer types for which trade occurs; the shaded regions indicate bunching in the allocation rule.\label{fig:regulation}}
\end{figure}

\section{Conclusion}\label{sec:conclude}
We study a general screening problem with an arbitrary allocation space and comonotonic values. A surplus-elasticity frontier---balancing available surplus against information rents per unit of surplus---is optimal for every type distribution and redistributive welfare weight, separating menu design from optimal pricing. Under an incremental-elasticity condition, this result holds among stochastic mechanisms, and a generalized notion of the frontier characterizes robust optimality. Across applications to bundling, taxation, sequential screening, selling information, and monopoly regulation, the same surplus-elasticity logic yields robust conditions for the optimality of simple mechanisms.

\renewcommand*{\bibfont}{\linespread{0.97}\scriptsize\selectfont}
\setlength\bibsep{0pt}
\bibliographystyle{ecta} 
\bibliography{references}

\newpage 
\appendix
\crefalias{section}{appendix}
\linespread{1.2}\selectfont

\makeatletter
\g@addto@macro\normalsize{%
  \setlength\abovedisplayskip{5pt plus 2pt minus 2pt}%
  \setlength\belowdisplayskip{4pt plus 2pt minus 2pt}%
  \setlength\belowdisplayshortskip{3pt plus 1pt minus 1pt}%
}
\makeatother
\normalsize

\noindent\title{\qquad \qquad \centering \LARGE \textbf{Online Appendix to ``Screening Frontiers''}}

\[\text{\centering \Large Frank Yang}\]

\section{Omitted Proofs}\label{app:proof}
\subsection{Proof of \texorpdfstring{\Cref{thm:main}}{}}

The proof follows closely the outline in \Cref{sec:proof}. We divide the proof into five steps. \Cref{subsec:prelim} derives some preliminary inequalities. \Cref{subsec:reconstruct} proves the reconstruction lemma. \Cref{subsec:downward} proves the generalized downward sufficiency theorem. \Cref{subsec:lottery} proves the purification lemma. \Cref{subsec:complete} completes the proof first for a finite allocation space and then extends the result to an arbitrary allocation space. 

\subsubsection{Preliminary Inequalities}\label{subsec:prelim}

\begin{lemma}\label{lem:inequality}
Let $\X^\star =\{\varnothing, x_1, \dots, x_m\}$ be a surplus-elasticity frontier, ordered by the surplus order. For any $x \not \in \X^\star$, let $x_k \in \X^\star$ be such that $x_k \neq \varnothing$ and 
\[x \preceq_{\emph{elasticity}} x_k \,.\]
Then, for all $1 \leq i < j \leq k$, we have 
\[  \frac{\d}{\d t}   \log(v(x, t) - v(x_i, t))  >  \frac{\d}{\d t}   \log(v(x_j, t) - v(x_i, t))\]
for all 
\[t \in \mathcal{T}':=\big\{s: v(x_j, s) >  v(x, s) > v(x_i, s)\big\}\,.\]
Moreover, if $\mathcal{T}' \neq \emptyset$, then it must be an interval. 
\end{lemma}
\begin{proof}

Since $\X^\star$ is a surplus-elasticity frontier, for all $i < j \leq k $, we have 
    \[\frac{\d}{\d t}  \log(v(x, t))\geq \frac{\d}{\d t}  \log(v(x_{k}, t)) \geq  \frac{\d}{\d t}  \log(v(x_j, t)) > \frac{\d}{\d t}  \log(v(x_i, t)) \text{ for all $t$}\,.\]
By the mediant inequality, this implies the following two inequalities: \textit{(i)} for all $i < j \leq k$,
\[ \frac{\d}{\d t}  \log(v(x_j, t) - v(x_i, t)) > \frac{\d}{\d t}  \log(v(x_j, t)) > \frac{\d}{\d t}  \log(v(x_i, t)) \geq 0\text{ for all $t$\,;}\]
and \textit{(ii)} for all $i < j \leq k$,
\[ \frac{\d}{\d t}  \log(v(x, t)) \geq \frac{\d}{\d t}  \log(v(x_j, t)) \geq  \frac{\d}{\d t}  \log(v(x_j, t) - v(x, t)) \text{ for all $t$ s.t. $v(x_j, t) > v(x, t)$ \,.}\]
Combining these two inequalities gives \textit{(iii)} for all $i < j \leq k$ 
\[ \frac{\d}{\d t}  \log(v(x_j, t) - v(x_i, t))   > \frac{\d}{\d t}  \log(v(x_j, t)) \geq  \frac{\d}{\d t}   \log(v(x_j, t) - v(x, t)) \text{ for all $t$ s.t. $v(x_j, t) > v(x, t)$\,.} \]
Finally, by combining the mediant inequality and inequality \textit{(iii)}, we have 
\[  \frac{\d}{\d t}   \log(v(x, t) - v(x_i, t)) > \frac{\d}{\d t}  \log(v(x_j, t) - v(x_i, t))  >  \frac{\d}{\d t}   \log(v(x_j, t) - v(x, t))\]
for all
\[t \in \mathcal{T}'=\big\{s: v(x_j, s)>  v(x, s)> v(x_i, s)\big\}\,,\]
proving the inequality. Now, note that 
\[\frac{\d}{\d t}  \log(v(x, t))\geq \frac{\d}{\d t}  \log(v(x_{j}, t)) \text{ for all $t$}\]
implies that for all $t < t'$, 
\[ v(x, t) \geq v(x_{j}, t) \implies  v(x, t') \geq v(x_{j}, t') \,.\]
Therefore, the set $\{s: v(x_j, s)>  v(x, s)\}$ is of the form $[\underline{t}, t')$ or $[\underline{t}, \overline{t}]$. Similarly, the set $\{s: v(x, s)>  v(x_i, s)\}$ is of the form $(t'', \overline{t}]$ or $[\underline{t}, \overline{t}]$. Therefore, if $\mathcal{T}'\neq\emptyset$, then $\mathcal{T}'$ must be an interval. 
\end{proof}

\subsubsection{Proof of \Cref{lem:reconstruct}}\label{subsec:reconstruct}
Note that by construction, for all $t$ where $x(t) \neq \varnothing$, we have 
\begin{align*}
v(a(t), t) &= \frac{v(x(t), t) - v(x^{-}(t), t)}{ v(x^{+}(t), t) -  v(x^{-}(t), t)} \cdot  v(x^{+}(t), t) + \Big[1 - \frac{v(x(t), t) - v(x^{-}(t), t)}{ v(x^{+}(t), t) -  v(x^{-}(t), t)}\Big] \cdot v(x^{-}(t), t)    \\
&=v(x(t), t)\,.
\end{align*}
Thus, by construction $v(a(t), t) - p(t) = v(x(t), t) - p(t)$ for all $t$. Hence, if type $t$ reports truthfully, then its payoff is unchanged in the reconstruction compared to that in $(x, p)$. Because the mechanism $(x, p)$ satisfies the downward IC constraints, to show that $(a, p)$ satisfies the downward IC constraints, it suffices to show that the deviating payoff of type $t$ misreporting to be $\hat{t} < t$ is weakly lower under $(a, p)$ compared to the deviating payoff under $(x, p)$. 

Since the payment rule is unchanged, it suffices to show that for $\hat{t} < t$, 
\[v(a(\hat{t}), t) = \alpha(\hat{t}) v(x^{+}(\hat{t}), t) + (1 - \alpha(\hat{t})) v(x^{-}(\hat{t}), t)  \leq v(x(\hat{t}), t)\,.\label{eq:dIC}\tag{A.1}\]
To ease notation, fix $\hat{t} < t$ and write $\hat{x}$, $\hat{x}^+$, and $\hat{x}^-$ for $x(\hat{t})$, $x^{+}(\hat{t})$, and $x^{-}(\hat{t})$ respectively. Note that if $\hat{x} \in \X^\star$, then $\hat{x}=\hat{x}^+$ and $\alpha(\hat{t}) = 1$, and hence \eqref{eq:dIC} clearly holds by construction. From now on, suppose $\hat{x} \not\in \X^\star$ (and in particular, $\hat{x} \neq \varnothing$). Then, by construction, $\hat{x}^- \prec_{\text{surplus}} \hat{x}^+$. We consider two cases. First, consider the case where $\hat{x}^- = \varnothing$. Then, $\hat{x}^+ = x_1$ and
\[\alpha(\hat{t}) = \frac{v(\hat{x}, \hat{t})}{v(x_1, \hat{t})}\,.\]
By property \textit{(i)} of the surplus-elasticity frontier, there exists
$x_k \in \X^\star$ such that  $x_k \neq \varnothing$ and
\[
\frac{\d}{\d s} \log v(\hat{x}, s) \geq \frac{\d}{\d s} \log v(x_k, s)\,.
\]
By property \textit{(ii)} of the surplus-elasticity frontier, we have $x_1 \preceq_{\text{surplus}} x_k$ and 
\[
\frac{\d}{\d s} \log v(x_k, s) \geq \frac{\d}{\d s}\log v(x_1, s)\,.
\]
Thus, $\frac{v(\hat{x}, \,\cdot\,)}{v(x_1, \,\cdot\,)}$ is nondecreasing. Since $\hat{t} < t$, we have 
\begin{align*}
v(a(\hat{t}), t) &= \alpha(\hat{t}) v(x_1, t) \\
                 &= \frac{v(\hat{x}, \hat{t})}{v(x_1, \hat{t})} v(x_1, t)\\
                 &\leq \frac{v(\hat{x}, t)}{v(x_1, t)} v(x_1, t) =v(\hat{x}, t)\,,
\end{align*}
proving \eqref{eq:dIC} for this case.

Now, consider the case where $\hat{x}^- \neq \varnothing$. Then, we can write the deviating payoff as 
\begin{align*}
\alpha(\hat{t}) v(x^{+}(\hat{t}), t) + (1 - \alpha(\hat{t})) v(x^{-}(\hat{t}), t) &= \alpha(\hat{t})\big(v(\hat{x}^+, t) - v(\hat{x}^-, t)\big) +v(\hat{x}^-, t)    \\
 &= \frac{v(\hat{x}, \hat{t}) - v(\hat{x}^{-}, \hat{t})}{ v(\hat{x}^{+}, \hat{t}) -  v(\hat{x}^{-}, \hat{t})}\big(v(\hat{x}^+, t) - v(\hat{x}^-, t)\big) +v(\hat{x}^-, t)\,.   
\end{align*}
Then, by properties \textit{(i)} and \textit{(ii)} of a surplus-elasticity frontier and the definition of $\hat{x}^+$, we must have $\hat{x}^+ \preceq_{\text{surplus}} x_k \in \X^\star$ for some $x_k \succeq_{\text{surplus}} \hat{x}$ such that  
\[\frac{\d}{\d s}\log v(\hat{x}, s) \geq \frac{\d}{\d s}\log v(x_k, s) \geq \frac{\d}{\d s} \log v(\hat{x}^+, s) > \frac{\d}{\d s} \log v(\hat{x}^-, s)\text{ for all $s$}\,, \]
where the first inequality follows from $\hat{x} \preceq_{\text{elasticity}} x_k$, the second inequality follows from $x_k \preceq_{\text{elasticity}} \hat{x}^+$, and the last inequality follows from $\hat{x}^+ \prec_{\text{elasticity}} \hat{x}^-$. This implies that for all $s < s'$, we have 
\[v(\hat{x}, s)\geq v(\hat{x}^+, s) \implies v(\hat{x}, s')\geq v(\hat{x}^+, s')\,. \label{eq:sc} \tag{A.2}\]
Note that if $v(\hat{x}, t) \geq v(\hat{x}^+, t)$, then \eqref{eq:dIC} clearly holds as 
\[\alpha(\hat{t})\big(v(\hat{x}^+, t) - v(\hat{x}^-, t)\big) +v(\hat{x}^-, t)\leq v(\hat{x}^+, t) \leq v(\hat{x}, t)\,.\]
Henceforth, suppose $v(\hat{x}, t) < v(\hat{x}^+, t)$. Then, by \eqref{eq:sc}, we have $v(\hat{x}, \hat{t}) < v(\hat{x}^+, \hat{t})$. By definition of $\hat{x}^-$, we have $v(\hat{x}^-, \hat{t})< v(\hat{x}, \hat{t}) < v(\hat{x}^+, \hat{t})$. But for any $s < s'$ we also have 
\[v(\hat{x}, s) > v(\hat{x}^-, s) \implies v(\hat{x}, s')> v(\hat{x}^-, s')\,. \label{eq:sc2} \tag{A.3}\]
Thus, we have $v(\hat{x}^-, t) < v(\hat{x}, t) < v(\hat{x}^+, t)$. Therefore, we have 
\[\{\hat{t}, t\} \subseteq \mathcal{T}':=\big\{s\in \mathcal{T}: v(\hat{x}^-, s) < v(\hat{x}, s) < v(\hat{x}^+, s)\big\}\,.\]
Since $\mathcal{T}' \neq \emptyset$, by \Cref{lem:inequality}, $\mathcal{T}'$ must be an interval. Moreover, by \Cref{lem:inequality}, we have 
\[ \frac{\d}{\d s}   \log(v(\hat{x}, s) - v(\hat{x}^-, s))  >  \frac{\d}{\d s}   \log(v(\hat{x}^+, s) - v(\hat{x}^-, s))\]
for all $s \in \mathcal{T}'$, which implies that 
\[ \frac{\d}{\d s}   \log\Bigg[\frac{v(\hat{x}, s) - v(\hat{x}^-, s)}{v(\hat{x}^+, s) - v(\hat{x}^-, s)}\Bigg]  >  0\]
for all $s \in \mathcal{T}'$. Thus, the function 
\[g(s) := \frac{v(\hat{x}, s) - v(\hat{x}^-, s)}{v(\hat{x}^+, s) - v(\hat{x}^-, s)}\]
is strictly increasing on the interval $\mathcal{T}'$. But then since $\hat{t} < t \in \mathcal{T}'$, we have 
\begin{align*}
\alpha(\hat{t}) v(x^{+}(\hat{t}), t) + (1 - \alpha(\hat{t})) v(x^{-}(\hat{t}), t) 
 &= \frac{v(\hat{x}, \hat{t}) - v(\hat{x}^{-}, \hat{t})}{ v(\hat{x}^{+}, \hat{t}) -  v(\hat{x}^{-}, \hat{t})}\big(v(\hat{x}^+, t) - v(\hat{x}^-, t)\big) +v(\hat{x}^-, t)\ \\
 &\leq \frac{v(\hat{x}, t) - v(\hat{x}^{-}, t)}{ v(\hat{x}^{+}, t) -  v(\hat{x}^{-}, t)}\big(v(\hat{x}^+, t) - v(\hat{x}^-, t)\big) +v(\hat{x}^-, t) \\
  &= v(\hat{x}, t) - v(\hat{x}^{-}, t) +v(\hat{x}^-, t) = v(\hat{x}, t)\,, 
\end{align*}
which proves \eqref{eq:dIC}. The claim follows. 

\subsubsection{Proof of \Cref{thm:downward}}\label{subsec:downward}

\paragraph{Proof of \Cref{lem:mon}.}\hspace{-2mm}For any optimal solution $a$ to the (\textbf{Downward-IC}) program, let $o(a) \in \{1, \dots, n\}$ be the first index at which $a_i > a_{i+1}$; if no such index exists, we set $o(a) = n$. This is the first index at which the allocation rule starts to decrease. We claim that there must exist an optimal solution $a$ with $o(a) = n$. Suppose for contradiction that all optimal solutions satisfy $o(a) < n$. 

Fix any optimal solution $a$ with the largest index $o(a)$ among those. Let $i = o(a) < n$. First, we consider what happens if we replace $(a_{i+1}, p_{i+1})$ that is assigned to type $t_{i+1}$ with $(a_i, p_i)$ while keeping everything else fixed. Call the new allocation and payment rules $(\hat{a}, \hat{p})$. By construction, we have 
\[a_1 \leq a_2 \leq \cdots \leq a_i \text{ and } a_i > a_{i+1}\,.\]
By definition $\{a_j: j < i\}$ were assigned to types $t_j$ below type $t_i$, and hence obtainable by type $t_i$; the fact that $t_i$ finds it downward incentive compatible to consume $(a_i, p_i)$ implies that type $t_i$ must prefer  $(a_i, p_i)$ over $(a_j, p_j)$ for all $j < i$.  At the same time, because $a_j \leq a_i$, by the increasing difference property of $v(a, t)$, for all $j < i$, we have 
\[v(a_i, t_{i+1}) - v(a_j, t_{i+1}) \geq v(a_i, t_{i}) - v(a_j, t_{i}) \geq p_i - p_j \,, \]
where the second inequality is due to IC$[i \rightarrow j]$ under $(a, p)$. It follows that IC$[(i+1) \rightarrow j]$ constraints are satisfied for all $j < i+1$ under the modified allocation and payment rule $(\hat{a}, \hat{p})$. Moreover, note that for any $k > (i+1)$, 
\[v(a_k, t_k) - p_k \geq v(a_{i}, t_k) - p_i = v(\hat{a}_{i+1}, t_k) - \hat{p}_{i+1}\,,\]
where the first inequality is due to IC$[k \rightarrow i]$ under $(a, p)$ and the equality holds by construction. Therefore, IC$[k \rightarrow (i+1)]$ constraints are satisfied for all $i+1 < k$ under the modified allocation and payment rule $(\hat{a}, \hat{p})$. Since we have only changed the allocation and payment of type $t_{i+1}$, it follows that $(\hat{a}, \hat{p})$ is downward incentive compatible (it also satisfies all IR constraints given that $v(a, t)$ is nondecreasing in $t$). Then, the fact that $(a, p)$ is optimal for (\textbf{Downward-IC}) implies that $(\hat{a}, \hat{p})$ cannot yield a strictly higher objective value for the principal. Thus, we must have 
\[\lambda(t_{i+1})\big(v(a_{i+1}, t_{i+1}) - p_{i+1}\big) + p_{i+1} \geq  \lambda(t_{i+1})\big(v(a_{i}, t_{i+1}) - p_{i}\big) + p_{i} \,,\]
or equivalently 
\[\lambda(t_{i+1})\Big(v(a_{i+1}, t_{i+1}) - v(a_{i}, t_{i+1}) + p_i - p_{i+1}\Big)  \geq   p_{i} -  p_{i+1} \,,\]
which implies that 
\[\lambda(t_{i})\Big(v(a_{i+1}, t_{i+1}) - v(a_{i}, t_{i+1}) + p_i - p_{i+1}\Big)  \geq   p_{i} -  p_{i+1} \,,\]
since $\lambda(t_i) \geq \lambda(t_{i+1})$ by assumption and  $v(a_{i+1}, t_{i+1}) - v(a_{i}, t_{i+1}) + p_i - p_{i+1} \geq 0$ given IC$[(i+1)\rightarrow i]$ under $(a, p)$.

Then, since $a_{i+1} < a_i$, by the strict increasing difference property of $v(a, t)$, we have 
\[v(a_{i+1}, t_{i+1}) - v(a_{i}, t_{i+1}) < v(a_{i+1}, t_{i}) - v(a_{i}, t_{i})\,.\]
Thus, since $\lambda(t_i) \geq 0$, we must have 
\[\lambda(t_{i})\Big(v(a_{i+1}, t_{i}) - v(a_{i}, t_{i}) + p_i - p_{i+1}\Big)  \geq   p_{i} -  p_{i+1} \,,\]
or equivalently 
\[\lambda(t_{i})\Big(v(a_{i+1}, t_{i}) - p_{i+1}\Big) + p_{i+1} \geq  \lambda (t_i) \big(v(a_{i}, t_{i}) - p_i\big) + p_{i}  \,,\]
and therefore, the principal's objective must have weakly increased when changing the original $(a, p)$ to $(\tilde{a}, \tilde{p})$ where the new mechanism is identical to the original mechanism except that we replace the allocation and payment $(a_i, p_i)$ of type $t_i$ with $(a_{i+1}, p_{i+1})$. We argue that the mechanism $(\tilde{a}, \tilde{p})$ satisfies all IR and all downward IC constraints. First, note that by the above arguments, we have 
\[v(a_{i+1}, t_{i}) - v(a_{i}, t_{i}) + p_i - p_{i+1} \geq v(a_{i+1}, t_{i+1}) - v(a_{i}, t_{i+1}) + p_i - p_{i+1} \geq 0\,, \]
and thus 
\[v(a_{i+1}, t_{i}) - p_{i+1} \geq v(a_{i}, t_{i}) - p_i \geq 0\,. \]
Therefore, $(\tilde{a}, \tilde{p})$ satisfies all IR constraints. Moreover, for any $j < i$, we have 
\[v(a_{i+1}, t_{i}) - p_{i+1} \geq v(a_{i}, t_{i}) - p_i \geq v(a_{j}, t_{i}) - p_j \,, \]
where the second inequality uses IC$[i \rightarrow j]$ under $(a, p)$. Therefore, IC$[i \rightarrow j]$ constraints hold for all $j < i$ under $(\tilde{a}, \tilde{p})$. For any $k > i$, we have 
\[v(a_{k}, t_{k}) - p_{k} \geq v(a_{i+1}, t_{k}) - p_{i+1} = v(\tilde{a}_{i}, t_{k}) - \tilde{p}_{i} \,, \]
where the first inequality uses IC$[k \rightarrow (i+1)]$ under $(a, p)$, and the equality follows by construction. Therefore, the mechanism $(\tilde{a}, \tilde{p})$ is downward incentive compatible, and hence a feasible solution to the (\textbf{Downward-IC}) program. Moreover, as we have argued, $(\tilde{a}, \tilde{p})$ weakly improves the principal's objective---therefore, $(\tilde{a}, \tilde{p})$ must also be an optimal solution to (\textbf{Downward-IC}).

Now, let 
\[r := \min \big\{j: a_j > a_{i+1} \big\} \,.\]
By construction $r \leq i$, and 
\[a_1 \leq a_2 \leq \cdots \leq a_{r-1} \leq a_{i+1} < a_{r} \leq \cdots \leq a_i\,.\]
If $r < i$, then we must have that $o(\tilde{a}) = i-1$. The identical argument above applies to $(\tilde{a}, \tilde{p})$, and yields an optimal solution $(\tilde{a}^{(1)}, \tilde{p}^{(1)})$ to program (\textbf{Downward-IC}) with $\tilde{a}^{(1)}_{i-1} = \tilde{a}^{(1)}_{i} = a_{i+1}$. Repeating this process $(i - r)$ times yields an optimal solution $(\tilde{a}^{(i-r)}, \tilde{p}^{(i-r)})$ to program (\textbf{Downward-IC}) with 
\[\tilde{a}^{(i-r)}_{r} = \tilde{a}^{(i-r)}_{r+1} = \cdots = \tilde{a}^{(i-r)}_{i} = a_{i+1}\,,\]
with $\tilde{a}^{(i-r)}_{j} = a_j$ for all $j \not \in \{r, \dots, i\}$. 

Then, it follows immediately that for all $1 \leq j \leq i$, 
\[ \tilde{a}^{(i-r)}_j \leq  \tilde{a}^{(i-r)}_{j+1}\,,\]
where we use the notation $\tilde{a}^{(0)} = \tilde{a}$. Therefore, 
\[o(\tilde{a}^{(i-r)}) \geq i + 1 > o(a)\,,\]
contradicting that, by construction, $o(a)$ is the largest such index among all optimal solutions to program (\textbf{Downward-IC}). 

\paragraph{Proof of \Cref{lem:binding}.}\hspace{-2mm}By \Cref{lem:mon}, without loss of optimality,  consider solutions to the \textbf{(Downward-IC)} program with monotone $a$. Suppose for contradiction that there exists no such optimal $(a, \tilde{p})$ such that local downward IC binds at every type. Then, for any optimal solution $(a, \tilde{p})$ to (\textbf{Downward-IC}), where $a$ is monotone, we have 
\[\Delta[\tilde{p}] := \sum_{i=1}^{n-1} \Big[\big(v(a_{i+1}, t_{i+1}) - \tilde{p}_{i+1}\big) - \big(v(a_{i}, t_{i+1}) - \tilde{p}_{i}\big)\Big] > 0\,.\]
Fix the optimal allocation rule $a$. Fix a payment rule $p$ such that $(a,p)$ is optimal and the above sum is minimized among all such $\tilde{p}$. Clearly, the above sum must be strictly positive for payment rule $p$ as well.  

Fix some $i$ such that IC$[(i+1) \rightarrow i]$ is slack under $(a, p)$. That is, 
\[\big(v(a_{i+1}, t_{i+1}) - p_{i+1}\big) - \big(v(a_{i}, t_{i+1}) - p_{i}\big) > 0\,.\]
We claim that IC$[k \rightarrow j]$ must be slack for all $k > i$ and all $j \leq i$. Indeed, by  IC$[k \rightarrow (i+1)]$, we have 
\[\big(v(a_{k}, t_{k}) - p_{k}\big) - \big(v(a_{i+1}, t_{k}) - p_{i+1}\big) \geq 0\,.\]
Moreover, by the increasing differences property of $v(a, t)$ and the slack of IC$[(i+1) \rightarrow i]$, we have 
\[\big(v(a_{i+1}, t_{k}) - p_{i+1}\big) - \big(v(a_{i}, t_{k}) - p_{i}\big) > 0\,.\]
Similarly, by the increasing differences property of $v(a, t)$ and IC$[i \rightarrow j]$, we have 
\[\big(v(a_{i}, t_{k}) - p_{i}\big) - \big(v(a_{j}, t_{k}) - p_{j}\big) \geq 0\,.\]
Now, summing the previous three inequalities gives 
\[\underbrace{\big(v(a_{k}, t_{k}) - p_{k}\big) - \big(v(a_{j}, t_{k}) - p_{j}\big)}_{=: \delta_{kj}} > 0\,, \]
and hence IC$[k \rightarrow j]$ holds with slack. Now, let 
\[\delta := \min\big\{\delta_{kj}: k > i,\, j \leq i\big\} > 0\,.\]
Let 
\[q := \sum_{k > i} \mu(t_k) \in (0, 1)\,. \]
Consider a modified mechanism $(a, \tilde{p})$ where for all $k$, 
\[
    \tilde{p}_k = 
\begin{cases}
 p_k + \delta/2 - q \cdot \delta/2 & \text{ if $k > i$}\,;\\ 
 p_k - q \cdot \delta/2 & \text{ otherwise}\,.
\end{cases}
\]
We first note that $(a, \tilde{p})$ is downward incentive compatible. Indeed, by construction, it suffices to check IC$[k \rightarrow j]$ where $k > i$ and $j \leq i$. But for any such IC constraint, as shown above, there exists at least $\delta$ slack, and hence lowering the deviation payoff gap of type $t_k$ mimicking $t_{j}$ by $\delta/2$ preserves IC$[k \rightarrow j]$.  Moreover, $(a, \tilde{p})$  also satisfies all IR constraints because IC$[j \rightarrow 1]$ constraints hold for all $j$, IR$[1]$ holds, and $v(a, t)$ is nondecreasing in $t$. Now, we claim that $(a, \tilde{p})$ yields a weakly higher objective for the principal. Indeed, by construction 
\[\sum_k \mu(t_k) \tilde{p}_k  =  - q \cdot \delta/2 + \sum_{k \leq i} \mu(t_k) p_k + \sum_{k > i}\mu(t_k) (p_k + \delta/2) = \sum_k \mu(t_k) p_k\,. \]
Moreover, since $q < 1$ by construction, note that for all $k$, we must have 
\[\sum_{j=k}^n \mu(t_j) p_j \leq \sum_{j=k}^n \mu(t_j) \tilde{p}_j \,.\]
The change in the principal's objective is given by 
\[\sum_k \Big(\mu(t_k)\tilde{p}_k  - \mu(t_k)p_k\Big) (1 - \lambda(t_k))\,.\]
Note that $1 - \lambda(t_k)$ is nondecreasing in $k$, and therefore, using summation by parts (Abel's identity), we have
\begin{align*}
 \sum_k \mu(t_k) p_k (1 - \lambda(t_k)) &=   (1 - \lambda(t_1)) \sum_{k=1}^n \mu(t_k)p_k + \sum_{k=2}^n \Big(\sum_{j=k}^n  \mu(t_j) p_j\Big) ( \lambda(t_{k-1}) - \lambda(t_k)) \\
 &\leq (1 - \lambda(t_1)) \sum_{k=1}^n \mu(t_k) \tilde{p}_k + \sum_{k=2}^n \Big(\sum_{j=k}^n \mu(t_j) \tilde{p}_j  \Big) ( \lambda(t_{k-1}) - \lambda(t_k))\,\\
 &=  \sum_k \mu(t_k) \tilde{p}_k (1 - \lambda(t_k))\,,
\end{align*}
where the inequality is due to \textit{(i)} $\sum_k \mu(t_k) \tilde{p}_k = \sum_k \mu(t_k) p_k$ and \textit{(ii)} $\sum_{j=k}^n \mu(t_j) p_j \leq \sum_{j=k}^n \mu(t_j) \tilde{p}_j$ for all $k$. Therefore, $(a, \tilde{p})$ must also be an optimal solution to (\textbf{Downward-IC}). Moreover, 
\[\Delta[\tilde{p}] = \Delta[p] - \delta/2 < \Delta[p]\,,\]
which is a direct contradiction since $p$ is chosen to minimize $\Delta[p]$ to begin with.

\paragraph{Completion for the finite-type case.}\hspace{-2mm}To complete the proof of \Cref{thm:downward} in the finite-type case, it suffices to show that the finite (\textbf{Downward-IC}) program admits a solution. Fix any $a \in \mathcal{A}$. We consider the problem of maximizing over $p \in \R^n$, or equivalently over $u \in \R^n$ where $u_i$ is the net utility of type $i$. For fixed $a$, the problem of maximizing over $u$ is a linear program. Once we show that there always exists a solution for every $a$, it follows that the value function over $a$ would be continuous and admit an optimal solution by a standard compactness argument. For any $i > j$, let  \[\Delta_{ji}(a) := v(a_j, t_i) - v(a_j, t_j) \geq 0\,.\]
Then IC$[i \rightarrow j]$ can be equivalently written as 
\[u_i \geq u_j + \Delta_{ji}(a)\,.\]
IR$[i]$ is simply $u_i \geq 0$. Let $\mathcal{U}(a)$ denote the feasible polyhedron. The objective can be written as 
\[\sum_i \mu_i (\lambda_i u_i + v(a_i, t_i) - u_i)  = \sum_i \mu_i v(a_i, t_i) + \sum_i \mu_i (\lambda_i- 1) u_i \,.\]
Clearly the LP over $u$ is feasible. By the fundamental theorem of LP, to show that it admits an optimal solution, it suffices to show that it cannot be unbounded. To show that, it suffices to show that the objective is bounded along any recession direction. Fix any $u \in \mathcal{U}(a)$. Consider any $d \in \R^n$ such that $u + \alpha d \in \mathcal{U}(a)$ for all $\alpha \geq 0$. We claim that $d_j \leq d_i$ for all $j < i$. Indeed, IC$[i \rightarrow j]$ gives that 
\[(u_i - u_j) + \alpha (d_i - d_j) \geq \Delta_{ji}(a)\,.\]
For this to hold for all $\alpha \geq 0$, we must have $d_i \geq d_j$. Moreover, IR constraints imply that $d_i \geq 0$ for all $i$. Now, note that for any recession direction $d$ with $d_i$ nondecreasing in $i$, by Abel's identity, we have 
\[\sum_i \mu_i (\lambda_i- 1) d_i = \sum_{k} \big(\sum_{i=k}^n \mu_i (\lambda_i- 1)\big) (d_{k} - d_{k-1})  \,,\]
with $d_0 := 0$. Note that since $\E[\lambda] \leq 1$ and $\lambda_i$ is nonincreasing in $i$, we must have 
\[\sum_{i=k}^n \mu_i (\lambda_i- 1) \leq 0\]
for all $k$. Therefore, since $d_i$ is nondecreasing in $i$, we have 
\[\sum_i \mu_i (\lambda_i- 1) d_i = \sum_{k} \big(\sum_{i=k}^n \mu_i (\lambda_i- 1)\big) (d_{k} - d_{k-1}) \leq 0 \,.\]
It follows that the objective value must be nonincreasing along any recession direction $d$. Thus, the objective value is bounded and hence the LP admits an optimal solution.  

\paragraph{Approximation argument for general case.}\hspace{-2mm}We prove the generalized downward sufficiency theorem for the general case where $\mathcal{T}, \mathcal{A}$ are any compact subsets of $\R$. Since $\lambda(t)$ is continuous in $t$, and $v(a, t)$ is continuous on $\mathcal{A} \times \mathcal{T}$, the approximation arguments are identical to those given in Appendix A.1.3 of \citet{yang2022costly}, provided that one shows the full IC program for the general-type case admits a solution with uniformly bounded transfers. The existence proof is identical to that given in Lemma 7 of \citet{yang2022costly} via Helly's selection theorem, except that we need to show that even with welfare weights $\lambda$, the transfers can be without loss of optimality restricted to satisfy a uniform upper and lower bound. 

We argue that it is without loss of optimality to restrict attention to payment rules $p: \mathcal{T} \rightarrow [\underline{v} - (\overline{v} - \underline{v}), \overline{v}]$, where 
\[\underline{v}:= \min_{(a, t) \in \mathcal{A} \times \mathcal{T}} v(a, t)\qquad \overline{v}:= \max_{(a, t) \in \mathcal{A} \times \mathcal{T}} v(a, t)\,.\]
Note that by IR constraints, we immediately have
\[p(t) \leq v(a(t), t) \leq \overline{v}\,.\]
Now to see the lower bound, we first claim that for the lowest type $\underline{t}$ we must have 
$p(\underline{t}) \geq \underline{v}$. Indeed, suppose otherwise. Then, $p(\underline{t}) < \underline{v}$, which means that the IR constraint for the lowest type must be slack. Let $\varepsilon := \underline{v} - p(\underline{t})$. Change the original payment rule $p$ to $\tilde{p}$ where $\tilde{p} = p + \varepsilon$ pointwise. Then, this new mechanism is clearly IC, and satisfies the IR constraint for the lowest type, which together also implies that it satisfies all IR constraints because $v(a, t)$ is nondecreasing in $t$. Note that the change in the objective is given by $\varepsilon - \E[\lambda] \varepsilon \geq 0$ since $\E[\lambda] \leq 1$, proving the claim. 

Now using the IC constraint of type $\underline{t}$ mimicking type $t$, we have 
\[v(a(\underline{t}), \underline{t}) - p(\underline{t}) \geq v(a(t), \underline{t}) - p(t)\]
and thus 
\[p(t) \geq v(a(t), \underline{t}) - v(a(\underline{t}), \underline{t}) + p(\underline{t}) \geq (\underline{v} - \overline{v}) + \underline{v}\,,\]
proving the claim. Therefore, it is without loss of optimality to restrict attention to payment rules $p: \mathcal{T} \rightarrow [\underline{v} - (\overline{v} - \underline{v}), \overline{v}]$. Thus, the full IC program with strict increasing differences preferences always admits a solution by the proof of Lemma 7 in \citet{yang2022costly}. By the same approximation arguments given in Appendix A.1.3 of \citet{yang2022costly}, it follows immediately that \Cref{thm:downward} holds.  

\subsubsection{Proof of \Cref{lem:purification}}\label{subsec:lottery}

First, note that for any $a \preceq_{\text{st}} a' \in \mathcal{A} \subseteq \Delta(\Y)$ where $a \neq a'$, and any $t < t'$, we have 
\[ \E_{x \sim a}[v(x, t') - v(x, t)] < \E_{x \sim a'}[v(x, t') - v(x, t)]\,,\]
given that $v(x, t)$ has strict increasing differences on $\mathcal{Y} \times \mathcal{T}$. Indeed, to see the above must hold with strict inequality, note that by \citet{Strassen1965}, $a'$ can be obtained from $a$ via a coupling $(Y, Y')$ where $Y \preceq Y'$ and $\P(Y \prec Y') > 0$.  Thus, $v(a, t)$ has strict increasing differences on $\mathcal{A} \times \mathcal{T}$. Then, as $(a, p)$ is incentive compatible, $a(\,\cdot\,)$ is stochastically monotone. 

Note that we can identify the compact set $\mathcal{Y}$ as a compact subset of $\R$: given the strict increasing differences on $(\Y, \preceq)$, the map $\psi(x) := v(x, \overline{t}) - v(x, \underline{t})$ is continuous and strictly increasing along $\preceq$, and hence defines a homeomorphism from the compact $\Y$ onto $\psi(\Y) \subset \R$ that preserves the order and thus the induced stochastic dominance order.

Now, by \citet{Strassen1965} again (see Lemma 1 of \citealt{yang2022costly}), there exists some random variable $\epsilon \in \mathcal{E}$ that is independent of $t$, where $ \mathcal{E}$ is some measurable space, and some jointly measurable function $x: \mathcal{T} \times  \mathcal{E} \rightarrow \Y$ such that \textit{(i)} $x(\,\cdot\,, \epsilon)$ is deterministically monotone according to $\preceq$ for all $\epsilon$, and \textit{(ii)} for every $t \in \mathcal{T}$, we have 
\[x(t, \epsilon) \sim  a(t) \in \Delta(\Y)\,. \]

Then, writing $\widetilde{V}(a, t) := \E_{x \sim a}\big[V(x, t) + W(t) v(x, t)\big]$, by the Envelope theorem, under the mechanism $(a, p)$, the principal's objective is given by
\begin{align*}
&\E\Bigg[\widetilde{V}(a(t), t) - W(t) \int_{\underline{t}}^t v_t(a(z), z) \d z \Bigg] - \E[W(t)] U(\underline{t}) \\
&=\E_{t, \epsilon}\Bigg[\widetilde{V}(x(t, \varepsilon), t) - W(t) \int_{\underline{t}}^t v_t(x(z, \varepsilon), z) \d z \Bigg] - \E[W(t)] U(\underline{t}) \\ 
& = \E_{\epsilon}\Big[\underbrace{\E_{t}\big[\widetilde{V}(x(t, \varepsilon), t) - W(t) \int_{\underline{t}}^t v_t(x(z, \varepsilon), z) \d z \big]}_{=:\,K(\varepsilon)}\Big] - \E[W(t)] U(\underline{t})\,,
\end{align*}
where $U$ is the agent's indirect utility. Note that there must exist some $\varepsilon^\dagger \in \mathcal{E}$ such that
\[K(\varepsilon^\dagger) \geq \E_{\varepsilon}[K(\varepsilon)]\,.\]
Thus, the objective is bounded from above by $K(\varepsilon^\dagger) - \E[W(t)] U(\underline{t})$. Let $x^\dagger(\,\cdot\,) = x(\,\cdot\,, \varepsilon^\dagger)$. Since $x^\dagger(\,\cdot\,)$ is monotone and $v(x, t)$ has increasing differences, using the transfer rule implied by the Envelope theorem with its level adjusted so that the type $\underline{t}$ gets a net payoff of $U(\underline{t})$, we can implement $x^\dagger(\,\cdot\,)$ to achieve the upper bound $K(\varepsilon^\dagger) - \E[W(t)] U(\underline{t})$, yielding an IC and IR deterministic mechanism $(x^\dagger, p^\dagger)$ that weakly improves on $(a, p)$.

\subsubsection{Completion of the Proof via Approximation}\label{subsec:complete}

By the argument in the main text and preceding lemmas, we have proved \Cref{thm:main} for the case where $\mathcal{X}$ is a finite set. We now generalize the result to the case where $\mathcal{X}$ is any compact metric space. 

Let the optimal value of the principal's objective be denoted by
\[V(\Y):= \sup_{(x, p)\,:\,\text{IC, IR, Ran}(x) \subseteq \Y} \text{Objective}(x, p)\,,\]
for any $\Y \subset \X$. Since $\E[\lambda] \leq 1$, it is easy to see that $V(\Y)$ is well defined for any $\Y$.

Building on \citet{madarasz2017sellers}, we prove the following finite approximation lemma, which would be the key to generalizing the result: 
\begin{lemma}[Finite Approximation]\label{lem:finite}
For any compact metric $\X$, there exists a sequence $\X^{(n)} \subseteq \X$ where each $\X^{(n)}$ is finite such that $\displaystyle \lim_{n\rightarrow \infty}V(\X^{(n)}) = V(\X)$.
\end{lemma}

We delay the proof of this lemma to the end. Assuming the above lemma, we now extend \Cref{thm:main} to any compact metric $\X$.

\paragraph{Completion of the proof of \Cref{thm:main}.}\hspace{-2mm}We first show that 
\[V(\X^\star) = V(\X)\,,\]
and then show that there exists an optimal solution. Fix the approximating sequence from \Cref{lem:finite}. Fix some $n$. Let $m$ enumerate the options in $\X^{(n)}$. For any $x_m \in \X^{(n)} \setminus \X^\star$, by construction of $\X^\star$, there exists some $\tilde{x}_m \in \X^\star$ such that 
\[x_m \preceq_{\text{surplus}} \tilde{x}_m \,\text{ and }\, x_m \preceq_{\text{elasticity}} \tilde{x}_m\,.\]
For any $x_m \in \X^{(n)} \cap \X^\star$, let $\tilde{x}_m = x_m$. Now, consider the collection
\[\X^\star_n := \{\tilde{x}_m\}_m \subseteq \X^\star\,.\]
Note that, by construction, $\X^\star_n$ is a surplus-elasticity frontier of 
\[\tilde{\X}^{(n)}= \X^{(n)} \cup \{\tilde{x}_m\}_m\]
which is finite. Moreover, by construction, we also have 
$V(\tilde{\X}^{(n)}) \geq V(\X^{(n)}) $ 
and hence by \Cref{lem:finite}, 
\[V(\X) \geq \displaystyle\lim_{n\rightarrow \infty}V(\tilde{\X}^{(n)}) \geq \displaystyle\lim_{n\rightarrow \infty}V(\X^{(n)}) = V(\X)\,,\] 
showing that $V(\tilde{\X}^{(n)})$ converges to $V(\X)$ as well. Now, apply \Cref{thm:main} to  $\tilde{\X}^{(n)}$. By \Cref{thm:main}, we immediately have that 
\[V(\X^\star_n) = V(\tilde{\X}^{(n)})\,,\]
and therefore 
\[ V(\X) \geq V(\X^\star) \geq \lim_{n\rightarrow \infty}V(\X^\star_n) = V(\X)\,, \]
proving the claim. 

We now show that an optimal mechanism using allocations in $\X^\star$ exists. By the proof given in \Cref{subsec:downward}, for any $n$, since the agent's preferences have strict increasing differences on $\X^\star_n$, there exists an optimal mechanism for the full IC program when the allocation space is restricted to $\X^\star_n$. Let $\big(x^{(n)}, p^{(n)}\big)$ be a sequence of the optimal mechanisms. Let $\preceq$ denote the linear order on $\X^\star$ indexed by $s$ (i.e., the surplus order on the frontier). By construction, $x^{(n)}: \mathcal{T} \rightarrow \X^\star$ is a monotone function with respect to $\preceq$. Since $\X^\star$ is compact, using Helly's selection theorem for
monotone functions on linearly ordered sets (\citealt{fuchino1999theorem}, Theorem 7), we can obtain a subsequence $x^{(n_j)}$ that converges pointwise as $j$ increases.\footnote{In particular, note that the strict surplus order on the frontier $\X^\star$ can be equivalently represented by $\psi(x) := v(x, t_0)$ for some fixed $t_0$, which defines a continuous bijection $\psi: \X^\star \rightarrow \psi(\X^\star)$. Since $\X^\star$ is compact, $\psi$ is a homeomorphism. Thus, $\X^\star$ can be equivalently viewed as a subset $\Y \subset \R$, which combined with Helly's selection theorem and the continuous inverse $\psi^{-1}$, implies a pointwise convergence subsequence.}  Let 
\[u^{(n)}(t) = v(x^{(n)}(t), t) - p^{(n)}(t)\]
be the equilibrium payoff of type $t$. Clearly, $u^{(n)}$ is also a monotone function for all $n$. Moreover, as we have shown before, the payment rule can be without loss of optimality restricted to a uniform upper bound and lower bound (independent of $n$). Therefore, $u^{(n)}$ are uniformly bounded monotone functions. Applying the same Helly's selection theorem again to the sequence $u^{(n_j)}$, we immediately obtain a subsequence $(x^{(n_{j_k})}, u^{(n_{j_k})})$ that converges pointwise. Equivalently, the sequence $(x^{(n_{j_k})}, u^{(n_{j_k})}, p^{(n_{j_k})})$ converges pointwise to some $(x, u, p)$. Since $v(x, t)$ is continuous on $\X \times \mathcal{T}$, it follows immediately that $(x, p)$ satisfies all IC and IR constraints. Finally, note that the objective value of this mechanism is given by 
\[\E\big[\lambda(t) u(t) + p(t)\big] = \lim_{k \rightarrow \infty }\E\big[\lambda(t) u^{(n_{j_k})}(t) + p^{(n_{j_k})}(t)\big] = \lim_{k \rightarrow \infty} V(\X^\star_{n_{j_k}}) = V(\X)\,,\]
where the first equality follows by the dominated convergence theorem. Thus, $(x, p)$ is an optimal mechanism that uses allocations only in the surplus-elasticity frontier $\X^\star$, completing the proof of \Cref{thm:main}.

\paragraph{Proof of \Cref{lem:finite}.}\hspace{-2mm}We explicitly construct a sequence of finite allocation spaces $\X^{(n)}$ that satisfies the desired property. For any $n \in \N$, partition $\mathcal{T}$ into $n$  subintervals:
\[
\underline{t}=t_0<t_1<\cdots<t_{n}=\overline{t},\qquad t_{k+1}-t_k=\frac{\overline{t} - \underline{t}}{n} =: \Delta_n\,.
\]
Let $\mathcal{T}^{(n)} := \{t_1,\dots,t_n\}$ and let $\mu^{(n)} \in \Delta(\mathcal{T}^{(n)})$ be defined by: for all $k$
\[
\mu^{(n)}_k := \P\big(t \in [t_{k-1}, t_{k}) \big)=F(t_{k})-F(t_{k-1})\,.
\]
Recall that $F$ is the continuous CDF of $t$. Let $\mu \in \Delta(\mathcal{T})$ be the distribution given by $F$.  Let $\lambda_k^{(n)}:=\lambda(t_k)$ for all $k$.

For any $n$, consider the same screening problem we have, but with the type space and type distribution given by $(\mathcal{T}^{(n)}, \mu^{(n)})$. In particular, the objective is given by 
\[\sum_k \mu^{(n)}_k  \Big(\lambda^{(n)}_k\big(v(x^{(n)}(t_k),t_k) -p^{(n)}(t_k)\big) + p^{(n)}(t_k)\Big)\,. \]
Fixing any allocation rule $x^{(n)}$, the optimization problem for $p^{(n)}$ is a finite-dimensional LP, which by the same argument in \Cref{subsec:downward}, admits an optimal solution. The LP value function is continuous in $x^{(n)} \in \X^n$, which combined with the continuity property of $v(x, t)$ implies that an optimal solution exists for this program by a standard compactness argument. Let $(x^{(n)}, p^{(n)})$ be an optimal solution to this program. Now define 
\[ \X^{(n)} := \text{Ran}(x^{(n)}) \subseteq \X\,,\]
which is a finite subset of $\X$ for any $n$. We claim that 
\[\lim_{n \rightarrow \infty} V(\X^{(n)}) = V(\X)\,.\]

To prove it, let $W(\mathcal{T}^{(n)}, \mu^{(n)})$ be the optimal value for the same screening problem under the type space $\mathcal{T}^{(n)} \subseteq \mathcal{T}$ and the type distribution $\mu^{(n)} \in \Delta(\mathcal{T}^{(n)})$. For two distributions $\mu_1, \mu_2 \in \Delta(\mathcal{T})$, $\mu_1$ and $\mu_2$ are said to be \textit{\textbf{$\delta$-close}} if $\mathcal{T}$ can be partitioned into disjoint measurable sets $S_1,\dots, S_k$ such that $||t- t'||_{\sup} < \delta$ for any $t, t'$ in the same cell $S_j$ and $\mu_1(S_j) = \mu_2(S_j)$ for each $S_j$ (\citealt{madarasz2017sellers}). Note that by construction, $\mu^{(n)}$ and $\mu$ are $\Delta_n$-close for any $n$.

Now, we make use of the following result:
\begin{lemma}[\citealt{madarasz2017sellers}; \citealt{Carroll2017}]\label{lem:local}
For any compact metric type space $\mathcal{T}$ and any $\epsilon >0$, there exists $\delta > 0$ such that for any mechanism $(x, p)$, there exists another mechanism $(\tilde{x}, \tilde{p})$ such that \textit{(i)} for any two $\delta$-close distributions $\mu, \mu'$, $\E_{\mu'}[\tilde{p}(t)] > \E_{\mu}[p(t)] -\epsilon$, and \textit{(ii)} $\{\tilde{x}(t)\}_{t} \subseteq \emph{Closure}(\{x(t)\}_{t})$, and  \textit{(iii)} $\tilde{u}(t) \geq u(t)$ for all $t \in \mathcal{T}$ where $\tilde{u}, u$ are the indirect utility functions given $(\tilde{x}, \tilde{p})$ and $(x, p)$ respectively.
\end{lemma}

Now, fix any $\varepsilon > 0$. Let $\delta > 0$ be given by \Cref{lem:local}.  Fix any mechanism $(x, p)$ defined on the type space $\mathcal{T}$. Let $(\tilde{x}, \tilde{p})$ be given by \Cref{lem:local}. Fix any $n$ such that $\Delta_n \leq \delta$. By the Envelope theorem, $u(t)$ is continuous and hence $\lambda(t)u(t)$ is continuous on the compact set $[\underline{t}, \overline{t}]$, which by the Heine--Cantor theorem implies that $\lambda(t)u(t)$ is uniformly continuous. Let $\omega(\,\cdot\,)$ denote its modulus of continuity. By \Cref{lem:local}, we have   
\[\E_{\mu^{(n)}}[\tilde{p}(t)] > \E_{\mu}[p(t)] - \varepsilon\]
and for all $t \in \mathcal{T}$, 
\[\tilde{u}(t) \geq u(t)\,.\]
The restriction of $(\tilde{x}, \tilde{p})$ to the type space of $\mathcal{T}^{(n)}$ defines a valid mechanism for the type space $\mathcal{T}^{(n)}$ with the objective value given by exactly 
\begin{align*}
    \E_{\mu^{(n)}}\Big[\tilde{p}(t) + \lambda(t) \tilde{u}(t) \Big] &>  \E_{\mu}\Big[p(t)\Big] - \varepsilon + \E_{\mu^{(n)}}\Big[\lambda(t) \tilde{u}(t) \Big]\\
    &\geq  \E_{\mu}\Big[p(t)\Big] - \varepsilon + \E_{\mu^{(n)}}\Big[\lambda(t) u(t) \Big]\\
    &\geq  \E_{\mu}\Big[p(t)\Big] - \varepsilon + \E_{\mu}\Big[\lambda(t) u(t) \Big] - \omega(\Delta_n)\\
    & =  \E_{\mu}\Big[p(t) + \lambda(t) u(t) \Big] - \varepsilon - \omega(\Delta_n)\,.
\end{align*}
Therefore, 
\[W(\mathcal{T}^{(n)}, \mu^{(n)}) \geq  \E_{\mu}\Big[p(t) + \lambda(t) u(t) \Big] - \varepsilon - \omega(\Delta_n)\,.\]
For the fixed mechanism $(x,p)$, we have $\omega(\Delta_n) \to 0$ as $n \rightarrow \infty$. Therefore,
\[ \liminf_{n \to \infty} W(\mathcal{T}^{(n)}, \mu^{(n)}) \geq \E_\mu[p(t) + \lambda(t) u(t)] - \varepsilon\,.
\]
Since $(x,p)$ was arbitrary, taking the supremum over all mechanisms then gives 
\[
\liminf_{n \to \infty} W(\mathcal{T}^{(n)}, \mu^{(n)}) \geq V(\mathcal{X})-\varepsilon\,.
\]
Taking $\varepsilon \rightarrow 0$ then gives 
\[\liminf_{n\rightarrow\infty} W(\mathcal{T}^{(n)}, \mu^{(n)}) \geq  V(\X) \,.\]

Now we claim that 
\[\liminf_{n\rightarrow \infty} V(\X^{(n)}) \geq \liminf_{n\rightarrow \infty}W(\mathcal{T}^{(n)}, \mu^{(n)})\,.\]
Fix any $\varepsilon > 0$. Let $\delta > 0$ be given by \Cref{lem:local}.  Fix any $n$ such that $\Delta_n \leq \delta$. Fix any optimal mechanism $(x^{(n)}, p^{(n)})$ defined on $\mathcal{T}$ with type distribution given by $\mu^{(n)} \in \Delta(\mathcal{T})$, such that it also satisfies $\text{Ran}(x^{(n)}) \subseteq \X^{(n)}$. Since $\lambda(t)$ is continuous on the compact set $\mathcal{T}$, $\lambda$ is also uniformly continuous with some modulus of continuity $\omega_\lambda(\,\cdot\,)$. Indeed, by the Envelope theorem we have that for all $|t - s| \leq \Delta$, 
\[|\lambda(t) u^{(n)}(t) - \lambda(s) u^{(n)}(s)| \leq ||\lambda||_{\infty} \sup_{x, t}|v_t(x, t)|\Delta +  ||u^{(n)}||_\infty\sup_{|t - s| \leq \Delta }|\lambda(t) - \lambda(s)| \]
Since $\sup_{x, t}|v_t(x, t)| \leq M$ for some constant $M$, and $||u^{(n)}||_\infty \leq M'$ for some $M'$ (since, as we argued, we can take $u(\underline{t}) = 0$ without loss and the derivative is uniformly bounded), we can further bound the above by 
\[|\lambda(t) u^{(n)}(t) - \lambda(s) u^{(n)}(s)| \leq \underbrace{||\lambda||_{\infty} M \Delta +  M' \omega_\lambda(\Delta)}_{=:\omega(\Delta)}\,.\]
Then, $\omega(\Delta)$ is a modulus of continuity of $\lambda u^{(n)}$ that holds uniformly over $n$. Let $(\tilde{x}, \tilde{p})$ be given by \Cref{lem:local} against $(x^{(n)}, p^{(n)})$. 
 By \Cref{lem:local}, we have   
\[\E_{\mu}[\tilde{p}(t)] > \E_{\mu^{(n)}}[p^{(n)}(t)] - \varepsilon\]
and for all $t \in \mathcal{T}$, 
\[\tilde{u}(t) \geq u^{(n)}(t)\,,\]
and $\text{Ran}(\tilde{x}) \subseteq \text{Closure}(\text{Ran}(x^{(n)})) \subseteq \X^{(n)}$. Therefore, by construction, we have  
\begin{align*}
    V(\X^{(n)}) &\geq \E_{\mu}\big[\tilde{p}(t) + \lambda(t) \tilde{u}(t)\big]\\
    &\geq  \E_{\mu^{(n)}}[p^{(n)}(t)] - \varepsilon + \E_{\mu}[\lambda(t) u^{(n)}(t)]\\
      &\geq  \E_{\mu^{(n)}}[p^{(n)}(t)] - \varepsilon + \E_{\mu^{(n)}}[\lambda(t) u^{(n)}(t)] - \omega(\Delta_n)\\
        &=  W(\mathcal{T}^{(n)}, \mu^{(n)})- \varepsilon - \omega(\Delta_n)\,.
\end{align*}
Therefore, we immediately have  
\[ \liminf_{n\rightarrow \infty} V(\X^{(n)}) \geq \liminf_{n\rightarrow \infty}W(\mathcal{T}^{(n)}, \mu^{(n)})- \varepsilon\,,\]
and hence taking $\varepsilon \rightarrow 0$ and using our previous conclusion gives  
\[\liminf_{n\rightarrow \infty} V(\X^{(n)}) \geq \liminf_{n\rightarrow \infty}W(\mathcal{T}^{(n)}, \mu^{(n)}) \geq V(\X)\,.\]
Evidently, $V(\X^{(n)}) \leq V(\X)$ and thus $\limsup_{n} V(\X^{(n)}) \leq V(\X)$. Thus, 
\[\lim_{n \rightarrow \infty} V(\X^{(n)}) = V(\X)\]
proving \Cref{lem:finite}.

\subsection{Proof of \Cref{thm:main2}}

Similar to the proof of \Cref{thm:main}, we prove \Cref{thm:main2} first assuming $\X$ is finite and then generalize it to an arbitrary compact metric $\X$ via approximation.    

\subsubsection{Proof under a Finite Allocation Space}

Suppose $\X$ is finite. Let $\X^\star$ be a strong frontier. Let $\X^\star = \{x_0, x_1, \dots, x_m\}$ where $x_i \prec_{\text{surplus}} x_j$ for all $i < j$ with $x_0 = \emptyset$ being the outside option. Note that for any $i \in \{1, \dots, m-1\}$,  
\[\frac{v(x_{i+1}, t) - v(x_i, t)}{v(x_{i}, t) - v(x_{i-1}, t)}\]
is strictly increasing in $t$ by the ordered incremental elasticity property, which is equivalent to: for any $t_1 < t_2$ we have 
\[\frac{v(x_{i+1}, t_2) - v(x_i, t_2)}{v(x_{i+1}, t_1) - v(x_{i}, t_1)}\]
is strictly increasing in $i$. We claim that for any $t_1 < t_2 \in (\underline{t}, \overline{t})$, there exists a convex function $g$ such that 
\[v(x, t_2) = g(v(x, t_1)) \text{ for all $x \in \X^\star$}\,.\]
To see this, fix any $t_1 < t_2 \in (\underline{t}, \overline{t})$. Let $z_i = v(x_i, t_1)$ and $y_i = v(x_i, t_2)$ for all $i$. Since the elements in $\X^\star$ are strictly ordered by the surplus order, we know that both $z_i$ and $y_i$ are strictly increasing in $i$. Define $g(z_i) = y_i$ on $\{z_i\}_{i=0}^m$ and extend it via linear interpolation on $[z_0, z_m]$. Then we have $0 = y_0 = g(z_0) = g(0)$, and on each segment $[z_i, z_{i+1}]$, 
\[g'(z) = \frac{g(z_{i+1}) - g(z_{i})}{z_{i+1} - z_{i}} = \frac{v(x_{i+1}, t_2) - v(x_{i}, t_2)}{v(x_{i+1},t_1) - v(x_i, t_1)} > 0 \,.\]
Since the above slope is increasing in $i$, it follows immediately that $g$ is convex, proving the claim. 

Now, fix any stochastic mechanism. Without loss, we can write it as $t \mapsto (x(t, \varepsilon), p(t, \varepsilon))$ where $x(t, \varepsilon) \in \mathcal{X}, p(t, \varepsilon) \in \R$ and $\varepsilon$ is a randomization device following an independent uniform distribution. Now, for each $\varepsilon$, applying the construction in \Cref{lem:reconstruct} to the deterministic mechanism defined by $(x(\,\cdot\,, \varepsilon), p(\,\cdot\,, \varepsilon))$, we obtain a reconstructed stochastic mechanism $(a(\,\cdot\,, \varepsilon), p(\,\cdot\,, \varepsilon))$ where $a(\,\cdot\,, \varepsilon) \in \Delta(\X^\star)$ that keeps the same objective value, maintains the truth-telling payoff for each type, and weakly decreases any downward deviation payoff. Now, since the original mechanism is downward IC (when averaging over $\varepsilon$), it follows immediately that for all $\hat{t} < t$, we have 
\begin{align*}
    \E_{\varepsilon}\Big[v(a(t, \varepsilon), t) - p(t, \varepsilon)\Big] &= \E_{\varepsilon}\Big[v(x(t, \varepsilon), t) - p(t, \varepsilon)\Big]  \\
    &\geq \E_{\varepsilon}\Big[v(x(\hat{t}, \varepsilon), t) - p(\hat{t}, \varepsilon)\Big]   \\
    &\geq \E_{\varepsilon}\Big[v(a(\hat{t}, \varepsilon), t) - p(\hat{t}, \varepsilon)\Big]\,,
\end{align*}
i.e., the downward IC constraints continue to hold under the compound stochastic mechanism $t \mapsto (a(t, \varepsilon), p(t, \varepsilon))$. Note that for each $t$, the compound lottery $a(t, \varepsilon)$ is equivalent to a simple lottery $\tilde{a}(t) \in \Delta(\X^\star)$. Let $\tilde{p}(t) = \E_\varepsilon[p(t, \varepsilon)]$ for all $t$. Then it follows immediately that $t \mapsto (\tilde{a}(t), \tilde{p}(t)) \in \Delta(\X^\star) \times \R$ satisfies all downward IC and all IR constraints and gives the same expected payoff to all types $t$ and to the principal, assuming truthful reporting. 

Now, we perform one more round of reconstruction. Fix any $t \in (\underline{t}, \overline{t})$. The lottery $\tilde{a}(t) \in \Delta(\X^\star)$ induces a lottery in the utility space $\{z_i\}_{i=0}^m$ where $z_i = v(x_i, t)$. Let $\Bar{z} := \E_{x \sim \tilde{a}(t)}[v(x, t)]$. Then $\Bar{z} \in [z_i, z_{i+1}]$ for some $i$. In particular, there exists a unique lottery $\alpha(t) \in \Delta(\{x_i, x_{i+1}\})$ such that $\E_{x \sim \alpha(t)}[v(x, t)] = \Bar{z}$. Moreover, it is easy to see that $v^{\alpha} \preceq_{\text{mps}} v^{\tilde{a}}$ where $v^{\alpha}, v^{\tilde{a}}$ are the utility lotteries for $t$ from $\alpha$ and $\tilde{a}$ respectively, and $\preceq_{\text{mps}}$ is the mean-preserving spread order. For any $t' > t \in (\underline{t}, \overline{t})$, by our previous observation, we can write 
\[v(x, t') = g(v(x, t)) \text{ for all $x \in \X^\star$}\,,\]
where $g$ is a convex function. It follows immediately that 
\[\E_{x\sim\alpha(t)}[v(x, t')] = \E_{x\sim\alpha(t)}[g(v(x, t))] = \E[g(v^\alpha)] \leq \E[g(v^{\tilde{a}})] = \E_{x\sim\tilde{a}(t)}[v(x, t')]\,, \]
where the inequality holds because $v^{\alpha} \preceq_{\text{mps}} v^{\tilde{a}}$ and $g$ is convex. 

It remains to define the reconstruction at the boundary types. For $t = \overline{t}$, let $\alpha(\overline{t})$ be any lottery supported on some $\{x_i, x_{i+1}\}$ such that $\E_{x \sim \alpha(\overline{t})}[v(x, \overline{t})] = \E_{x \sim \tilde{a}(\overline{t})}[v(x, \overline{t})]$. For $t = \underline{t}$, take any sequence $t_n \downarrow \underline{t}$ along which $\big(\alpha(t_n), \tilde{p}(t_n)\big)$ converges, and assign type $\underline{t}$ the limit. After redefining $\big(\alpha(\underline{t}), \tilde{p}(\underline{t})\big)$ according to this limit, note that the mechanism $t \mapsto \big(\alpha(t), \tilde{p}(t)\big)$ satisfies all downward IC constraints and all IR constraints, and gives the same payoff to almost all types and hence the same objective value to the designer, assuming truthful reporting.\footnote{This holds since $F$ has no atom at $\underline{t}$; however, even when $F$ has an atom at $\underline{t}$, an exact construction at $\underline{t}$ without the limiting argument can handle it.}

Now, observe that for any $t$, we have 
\[\alpha(t) \in \mathcal{A}:= \Big\{a\in \Delta(\X^\star): a \in \Delta\big(\{x_{j-1}, x_j\}\big) \text{ for some $x_j \in \X^\star$}\Big\}\,,\]
which is the totally ordered set of lotteries used in the proof of \Cref{thm:main}. Applying the identical arguments in the proof of  \Cref{thm:main}---using \Cref{thm:downward} and \Cref{lem:purification}---we immediately obtain \Cref{thm:main2} for any finite $\X$.

\subsubsection{Completion of the Proof via Approximation}

The argument is similar to the approximation argument given for \Cref{thm:main}. Let the optimal value of the principal's objective be denoted by
\[V(\Delta(\Y)):= \sup_{(a, p)\,:\,\text{IC, IR, Ran}(a) \subseteq \Delta(\Y)} \text{Objective}(a, p)\,,\]
for any $\Y \subset \X$. Since $\E[\lambda] \leq 1$, it is easy to see that $V(\Delta(\Y))$ is well defined for any $\Y$.

Parallel to the finite approximation lemma (\Cref{lem:finite}), we will prove a stochastic finite approximation lemma: 
\begin{lemma}[Stochastic Finite Approximation]\label{lem:stoc-finite}
For any compact metric $\X$, there exists a sequence $\X^{(n)} \subseteq \X$ where each $\X^{(n)}$ is finite such that $\displaystyle \lim_{n\rightarrow \infty}V(\Delta(\X^{(n)})) = V(\Delta(\X))$.
\end{lemma}

Again, we delay the proof of this lemma to the end. Assuming the above lemma, we now extend \Cref{thm:main2} to any compact metric $\X$.

\paragraph{Completion of the proof of \Cref{thm:main2}.}\hspace{-2mm}We show that 
\[V(\X^\star) = V(\Delta(\X))\,.\]
Note that by \Cref{thm:main}, we already have the existence of an optimal mechanism among all deterministic mechanisms. Fix the approximating sequence from \Cref{lem:stoc-finite}. The argument proceeds as before. Fix some $n$. Let $m$ enumerate the options in $\X^{(n)}$. For any $x_m \in \X^{(n)} \setminus \X^\star$, by construction of $\X^\star$, there exists some $\tilde{x}_m \in \X^\star$ such that 
\[x_m \preceq_{\text{surplus}} \tilde{x}_m \,\text{ and }\, x_m \preceq_{\text{elasticity}} \tilde{x}_m\,.\]
For any $x_m \in \X^{(n)} \cap \X^\star$, let $\tilde{x}_m = x_m$. Now, consider the collection
\[\X^\star_n := \{\tilde{x}_m\}_m \subseteq \X^\star\,.\]
Note that, by construction, $\X^\star_n$ is a surplus-elasticity frontier of 
\[\tilde{\X}^{(n)}= \X^{(n)} \cup \{\tilde{x}_m\}_m\]
which is finite. Moreover, since $\X^\star_n \subseteq \X^\star$,  $\X^\star_n$ continues to have the ordered incremental elasticity property and hence is a strong frontier of $\tilde{\X}^{(n)}$. Now, applying \Cref{thm:main2} to $\tilde{\X}^{(n)}$, we obtain 
\[V(\X^\star_n) = V(\Delta(\tilde{\X}^{(n)}))\,.\]
Moreover, by construction, we also have 
$V(\Delta(\tilde{\X}^{(n)})) \geq V(\Delta(\X^{(n)})) $ 
and hence by \Cref{lem:stoc-finite}
\[V(\Delta(\X)) \geq \displaystyle\lim_{n\rightarrow \infty}V(\Delta(\tilde{\X}^{(n)})) \geq \displaystyle\lim_{n\rightarrow \infty}V(\Delta(\X^{(n)})) = V(\Delta(\X))\,,\] 
showing that $V(\Delta(\tilde{\X}^{(n)}))$ converges to $V(\Delta(\X))$ as well. Therefore, we have 
\[ V(\Delta(\X)) \geq V(\X^\star) \geq \lim_{n\rightarrow \infty } V(\X^\star_n) = \lim_{n\rightarrow \infty }V(\Delta(\tilde{\X}^{(n)}))  =  V(\Delta(\X))\,,\]
proving \Cref{thm:main2}. 

\paragraph{Proof of \Cref{lem:stoc-finite}.}\hspace{-2mm}The high-level argument is similar to that of \Cref{lem:finite}. We explicitly construct a sequence of finite allocation spaces $\X^{(n)}$ that satisfies the desired property. For any $n \in \N$, partition $\mathcal{T}$ into $n$  equal-length subintervals as before. Let $\Delta_n$ denote the length of each subinterval. Let $\mathcal{T}^{(n)} := \{t_1,\dots,t_n\}$ and let $\mu^{(n)} \in \Delta(\mathcal{T}^{(n)})$ be defined by: for all $k$
\[
\mu^{(n)}_k := \P\big(t \in [t_{k-1}, t_{k}) \big)=F(t_{k})-F(t_{k-1})\,.
\]
Let $\mu \in \Delta(\mathcal{T})$ be the distribution given by $F$.  Let $\lambda_k^{(n)}:=\lambda(t_k)$ for all $k$.

For any $n$, consider the same screening problem with stochastic allocation space $\Delta(\X)$, but with the type space and type distribution given by $(\mathcal{T}^{(n)}, \mu^{(n)})$. In particular, the objective is given by 
\[\sum_k \mu^{(n)}_k  \Big(\lambda^{(n)}_k\big(v(a^{(n)}(t_k),t_k) -p^{(n)}(t_k)\big) + p^{(n)}(t_k)\Big)\,. \]
Note that since we are allowing for stochastic mechanisms, the above program is a linear program in $(a^{(n)},  p^{(n)}) \in \Delta(\X)^n \times \R^n$. As argued before, without loss of optimality we can restrict attention to $p^{(n)} \in [-K, K]^{n}$ for a large enough universal constant $K$. We claim that there exists an optimal solution $(a^{(n)}, p^{(n)})$ to the above program where $\{a^{(n)}(t)\}_{t \in \mathcal{T}^{(n)}} \subseteq \Delta(\X^{(n)})$ where $\X^{(n)}$ is a finite set. 

To see this, we enrich the space of optimizing variables: Consider choosing a joint lottery across allocations and payments and across types: 
\[\beta \in \Delta((\X \times [-K, K])^{n})\,,\]
where the $i$-th dimension marginal $\beta_i$ prescribes the stochastic allocation and stochastic payment for a reported $t_i$ type. Evidently, the above objective is a linear functional of $\beta$ with constraints given by $n \times n$ many IC constraints on $\beta$, and $n$ many IR constraints on $\beta$. Each IC constraint is a moment constraint on $\beta$; each IR constraint is also a moment constraint on $\beta$. Since $(\X \times [-K, K])^{n}$ is compact, we have that $\Delta((\X \times [-K, K])^{n})$ is compact in the weak-$^*$ topology. Then, by Bauer's maximum principle, there exists an optimal solution $\beta^{(n)}$ that is an extreme point of the set of measures constrained by $n^2 + n$ moment constraints. By \citet{winkler1988extreme}, it follows immediately that $\beta^{(n)}$ is finitely supported---in particular, it implies that $\beta^{(n)}_i$ is finitely supported for each type $t_i$.  For each type $t_i$, take the allocation-marginal $a^{(n)}_i$ of $\beta^{(n)}_i$ and let $\X^{(n)}_i \subseteq \X$ be its finite support. Also take the expectation $p^{(n)}_i$ of the payment-marginal of $\beta^{(n)}_i$. It then follows that $(a^{(n)}, p^{(n)})$ is an optimal solution to our finite-type program. Moreover, let 
\[\X^{(n)} = \bigcup_{i=1}^n \X^{(n)}_i\]
which is finite. It then follows immediately that $\{a^{(n)}(t)\}_{t \in \mathcal{T}^{(n)}} \subseteq \Delta(\X^{(n)})$, proving the claim. 

Now, we claim that 
\[\lim_{n \rightarrow \infty} V(\Delta(\X^{(n)})) = V(\Delta(\X))\,.\]
This step is very similar to that given in \Cref{lem:finite} and in particular relies on \Cref{lem:local}. In particular, let $W_{\Delta}(\mathcal{T}^{(n)}, \mu^{(n)})$ be the optimal value for the stochastic screening problem under the type space $\mathcal{T}^{(n)} \subseteq \mathcal{T}$ and the type distribution $\mu^{(n)} \in \Delta(\mathcal{T}^{(n)})$. Since any stochastic mechanism $(a, p)$ defined on $\mathcal{T}$ also induces a uniformly continuous indirect utility function $u$, replacing $\X$ with $\Delta(\X)$ in the argument given in \Cref{lem:finite}, we immediately obtain 
\[\liminf_{n\rightarrow\infty} W_\Delta(\mathcal{T}^{(n)}, \mu^{(n)}) \geq  V(\Delta(\X)) \,.\]
Moreover, since $\sup_{x, t}|v_t(x, t)| \leq M$ for some constant $M$, we have 
\[\sup_{a \in \Delta(\X),\, t \in \mathcal{T}} \big|\E_{x \sim a}[v_t(x, t)]\big| \leq M \]
as well. Then, fix any optimal mechanism $(a^{(n)}, p^{(n)})$ defined on $\mathcal{T}$ with type distribution given by $\mu^{(n)} \in \Delta(\mathcal{T})$, such that it also satisfies $\text{Ran}(a^{(n)}) \subseteq \Delta(\X^{(n)})$. Note that the induced indirect utility function $u^{(n)}$ is uniformly continuous with a modulus of continuity $\omega(\,\cdot\,)$ that holds uniformly across $n$. Replacing $\X$ with $\Delta(\X)$ in the argument given in \Cref{lem:finite}, we also immediately obtain 
\[\liminf_{n \rightarrow \infty} V(\Delta(\X^{(n)})) \geq \liminf_{n\rightarrow\infty} W_\Delta(\mathcal{T}^{(n)}, \mu^{(n)})\,.\]
Therefore, it follows immediately that 
\[V(\Delta(\X)) \geq \limsup_{n \rightarrow \infty} V(\Delta(\X^{(n)})) \geq \liminf_{n \rightarrow \infty} V(\Delta(\X^{(n)})) \geq \liminf_{n\rightarrow\infty} W_\Delta(\mathcal{T}^{(n)}, \mu^{(n)}) \geq  V(\Delta(\X)) \,,\]
showing that they are all equal and hence proving \Cref{lem:stoc-finite}.

\subsection{Proof of \Cref{cor:profile}}

The demand profile method is valid by the argument given in the main text. We show that $\X^\star$ is minimal optimal if $\argmax \{\Delta P(s, q) \cdot q\} \subset (0, 1)$. Suppose for contradiction that there exists a closed subset $S \subset [0, 1]$ such that $[0, 1]\setminus S$ has a positive Lebesgue measure and $\{x_s\}_{s \in S}$ is an optimal menu. Let $(x, p)$ be the optimal mechanism supporting this optimal menu. Equivalently, we can represent $(x, p)$ via a nonincreasing function $s \mapsto q(s)$ by 
\[q(s) := \P\big(x(t) \geq s\big)\,,\]
and write the revenue obtained by $(x, p)$ as
\[\int \big(\Delta P(s, q(s)) \cdot q(s)\big) \d s\,.\]
Since $S$ has measure strictly less than $1$, there exists an open interval $(s_1, s_2) \subseteq [0,1]\setminus S$. Then, by definition of $q$, there exists a constant $\bar{q}$ such that $q(s) = \bar{q}$ for all $s \in (s_1, s_2)$. Since we have argued that the upper bound via pointwise upgrade pricing can be obtained: 
\[\int \max_{q \in [0, 1]}(\Delta P(s, q) \cdot q) \d s\,,\]
it follows that for almost all $s \in [0, 1]$, we have 
\[q(s) \in \argmax_{q \in [0, 1]}\Big(\Delta P(s, q) \cdot q\Big)\,,\]
because otherwise $(x, p)$ cannot be optimal. Since the frontier is interior, it follows that $\bar{q} \in (0,1)$. Fix any
$a < b < c \in (s_1,s_2)$. Since $\bar{q}$ maximizes $\Delta P(s,q) \cdot q$ for almost all $s \in (s_1,s_2)$, we have
\[\max_q P(x_b, q \mid x_a) \cdot q 
= \max_q \int_a^b \Delta P(s,q)\cdot q \d s
\leq \int_a^b \max_q \big(\Delta P(s, q) \cdot q \big) \d s
= P(x_b, \bar{q} \mid x_a) \cdot \bar{q}\,,\]
and hence $\bar{q}$ maximizes $P(x_b,q \mid x_a)\cdot q$. Likewise, $\bar{q}$ maximizes $P(x_c,q \mid x_b)\cdot q$. Since the two incremental demand curves are positive and differentiable in $q$, the first-order conditions then imply that $\eta(x_b, \bar{q} \mid x_a) = \eta(x_c, \bar{q} \mid x_b) = 1$, contradicting the
ordered incremental elasticity property.   

\subsection{Proof of \Cref{thm:main3}}

\paragraph{Optimality.}\hspace{-2mm}We first show that a generalized frontier $\X^\star$ is an optimal menu. Since $\X^\star$ is a generalized frontier, there exists an order $\preceq$ on $\X^\star$ and a set of lotteries $\mathcal{A} \subseteq \Delta(\X^\star)$ that are totally ordered by $\preceq_{\text{st}}$ such that the agent's preferences have strict increasing differences on $\X^\star$ according to the order $\preceq$ and every element $x$ can be covered by $\mathcal{A}$. Now, fix any mechanism $(x, p)$. Fix any type $t$. By the definition of a covering, we know that there exists some $a \in \mathcal{A}$ such that 
\[v(a, t) = v(x(t), t)\,,\]
since the point $\big(1 - F(t), v(x(t), t)\big)$ is in the graph of the demand curve $P(x(t), \,\cdot\,)$ and hence must be in the graph of the demand curve $P(a, \,\cdot\,)$ for some $a \in \mathcal{A}$. Moreover, by the definition of a covering, we can take such an $a$ with $P(a, \,\cdot\,)$ weakly single-crossing $P(x(t), \,\cdot\,)$ from below, and hence for all $t' > t$, we have 
\[v(a, t') \leq v(x(t), t')\,.\]
We will use $a(t)$ as the reconstructed allocation rule following this procedure for all $t$, but before that we first ensure that there exists a measurable selection $t \mapsto a(t)$ following this procedure. Formally, let 
\[\Gamma(t):= \Big\{a \in \mathcal{A}: v(a, t) = v(x(t), t) \text{ and } v(a, t') \leq v(x(t), t') \text{ for all $t' \geq t$} \Big\}\,.\]
By the earlier argument, we observe that for any $t$, $\Gamma(t) \subseteq \Delta(\X^\star)$ is a non-empty compact set.  Moreover, the graph of $\Gamma$,  
\[\text{Graph}(\Gamma) = \Big\{(t, a) \in \mathcal{T} \times \mathcal{A}: v(a, t) = v(x(t), t) \text{ and } \max_{t' \in [t, \overline{t}]} \big\{v(a, t') - v(x(t), t')\big\} \leq 0  \Big\}\,,\]
is Borel measurable given that $x(\,\cdot\,)$ is Borel measurable. Then, by the Arsenin--Kunugui uniformization theorem (\citealt{kechris2012classical}, Theorem 18.18), there exists a measurable selection $t \mapsto a(t)$ such that $a(t) \in \Gamma(t)$ for all $t$. Fix this measurable selection and let the reconstruction be defined by $(a, p)$. Then, $(a, p)$ satisfies all IR and all downward IC constraints, and yields the same objective value to the principal, assuming truthful reporting. 

Note that $\{a(t)\}_{t \in \mathcal{T}} \subseteq \mathcal{A}$ which is totally ordered by $\preceq_{\text{st}}$. By the same argument as in the proof of \Cref{thm:main}, we know that for any $a \preceq_{\text{st}} a' \in \mathcal{A}$ where $a \neq a'$, we have 
\[ \E_{x \sim a}[v(x, t_2) - v(x, t_1)] < \E_{x \sim a'}[v(x, t_2) - v(x, t_1)]\,,\]
given that $v(x, t)$ has strict increasing differences on $\mathcal{X}^\star \times \mathcal{T}$.

Therefore, as in the proof of \Cref{thm:main}, we can apply \Cref{thm:downward} to $(a, p)$ and obtain a fully IC and IR mechanism $(\tilde{a}, \tilde{p})$ where $\text{Ran}(\tilde{a}) \subseteq \mathcal{A}$ and $(\tilde{a}, \tilde{p})$ yields a weakly higher objective than $(a, p)$.\footnote{Formally, given the strict increasing differences of $v(a, t)$ on  $\mathcal{A} \times \mathcal{T}$, the lotteries in $\mathcal{A}$ can be 
equivalently ordered by $\psi(a) := v(a, \overline{t}) - v(a, \underline{t})$, which is strictly increasing along $\preceq_{\text{st}}$ and continuous, and hence defines a homeomorphism $\psi: \mathcal{A} \rightarrow \psi(\mathcal{A})$ onto a compact subset of $\R$ since $\mathcal{A}$ is compact. Relabeling allocations by $\psi$, we can then apply \Cref{thm:downward} with the allocation space $\psi(\mathcal{A}) \subset \R$.} Then, as in the proof of \Cref{thm:main}, we can apply \Cref{lem:purification} to $(\tilde{a}, \tilde{p})$ and obtain a fully IC and IR deterministic mechanism $(x^\dagger, p^\dagger)$, where $\text{Ran}(x^\dagger) \subseteq \mathcal{X}^\star$, that further improves the objective. This completes the proof. 

\paragraph{Stochastic Optimality.}\hspace{-2mm}Now, suppose that the generalized frontier $\X^\star$ has pointwise ordered demand curves that satisfy the ordered incremental elasticity property. We show that the frontier is optimal among all stochastic mechanisms. Fix any stochastic mechanism. As in the proof of \Cref{thm:main2}, we can write it equivalently as $t \mapsto (x(t, \varepsilon), p(t))$ where $\varepsilon$ is a randomization device drawn independently of $t$. For each realization of $\varepsilon$, we can apply the same reconstruction argument just like before. Then, as in the proof of \Cref{thm:main2}, we can obtain a modified stochastic mechanism $t \mapsto (a(t, \varepsilon), p(t))$---in particular, by the same measurable selection argument as above, $a(t, \varepsilon)$ can be made jointly measurable in $(t, \varepsilon)$---that satisfies all IR constraints and all downward IC constraints (when averaged over $\varepsilon$) and weakly improves the principal's objective, under truthful reporting.  Writing out the compound lotteries as simple lotteries over $\X^\star$, we can equivalently write the downward-IC stochastic mechanism as $t \mapsto (\tilde{a}(t), p(t))$ where $\tilde{a}(t) \in \Delta(\X^\star)$ for all $t$. 

Now, consider $\X^\star$ itself as an allocation space. Clearly, given the incremental elasticity property, $\X^\star$ is a strong frontier for $\X^\star$ itself. By inspection, the proof of \Cref{thm:main2} not only establishes that any strong frontier is optimal among all stochastic mechanisms, but in fact proves a stronger claim: This holds among all downward-IC stochastic mechanisms---formally, for any strong frontier $\Y^\star$ of $\Y$, there exists a full-IC, deterministic mechanism $(x, p)$ such that $\text{Ran}(x) \subseteq \Y^\star$ that is optimal among all downward-IC stochastic mechanisms (i.e., that use elements in $\Delta(\mathcal{Y})$).\footnote{The finite case follows by the identical proof of \Cref{thm:main2}. For general compact $\Y^\star$, round every realized allocation monotonically down to a finite ordered grid in $\Y^\star$ and reduce each type's payment to keep truthful reporting payoff the same. Then, given increasing differences, all downward IC constraints are preserved. Compactness and continuity make the objective loss vanish as the grid becomes dense, so the claim follows from the finite case.} Since $\X^\star$ is a strong frontier of $\X^\star$, we immediately obtain that $\X^\star$ is optimal among all downward-IC stochastic mechanisms---in particular, there exists a full-IC, deterministic mechanism $(x^\dagger, p^\dagger)$ where $\text{Ran}(x^\dagger) \subseteq \X^\star$ such that the objective value under $(x^\dagger, p^\dagger)$ is weakly higher than that under $(\tilde{a}, p)$ and hence weakly higher than the original stochastic mechanism $t \mapsto (x(t, \varepsilon), p(t))$. Since the original stochastic mechanism was arbitrary, this then completes the proof of \Cref{thm:main3}.

\subsection{Proof of \Cref{thm:converse}}

Fix a finite menu $\mathcal{Y} = \{\varnothing, y_1, \dots, y_m\}$ where $y_i$'s are ordered according to the surplus order. Suppose the menu also satisfies the ordered incremental elasticity property which is agnostic to the type distribution and equivalent to that for all $i = 1, \dots, m-1$, we have 
\[0 \leq \frac{\d}{\d t} \log \Big(v(y_{i}, t) - v(y_{i-1}, t)\Big)  < \frac{\d}{\d t} \log \Big(v(y_{i+1}, t) - v(y_{i}, t)\Big)\,, \tag{A.4} \label{eq:incremental}\]
where $y_0 = \varnothing$. Then, writing $d_i(t) := v(y_{i}, t) - v(y_{i-1}, t)$, and 
\[\rho_i(t_L, t_H) := \frac{d_i(t_L)}{d_i(t_H)}\]
for any $t_L < t_H \in \mathcal{T}$, we have 
\[0 < \rho_m(t_L, t_H) < \cdots < \rho_1(t_L, t_H) < 1\,.\]

\paragraph{Sufficiency.}\hspace{-2mm}This is an immediate consequence of the proof of \Cref{thm:main3}. Indeed, suppose $\Y$ is a generalized frontier and let $\mathcal{A} \subseteq \Delta(\mathcal{Y})$ be the covering family of lotteries. Fix any IC and IR mechanism $(x, p)$. Let 
\[\Gamma(t):= \Big\{a \in \mathcal{A}: v(a, t) = v(x(t), t) \text{ and } v(a, t') \leq v(x(t), t') \text{ for all $t' \geq t$} \Big\}\,.\]
By the definition of a covering, for any $t$, $\Gamma(t) \subseteq \Delta(\Y)$ is a non-empty compact set.  Moreover, the graph of $\Gamma$,  
\[\text{Graph}(\Gamma) = \Big\{(t, a) \in \mathcal{T} \times \mathcal{A}: v(a, t) = v(x(t), t) \text{ and } \max_{t' \in [t, \overline{t}]} \big\{v(a, t') - v(x(t), t')\big\} \leq 0  \Big\}\,,\]
is Borel measurable given that $x(\,\cdot\,)$ is Borel measurable. Then, by the Arsenin--Kunugui uniformization theorem (\citealt{kechris2012classical}, Theorem 18.18), there exists a measurable selection $t \mapsto a(t)$ such that $a(t) \in \Gamma(t)$ for all $t$. This results in a downward-IC and IR mechanism $(a, p)$ that attains the same objective value as $(x, p)$. The rest is identical to the proof of \Cref{thm:main3} since the downward sufficiency theorem (\Cref{thm:downward}) holds for all Borel type distributions and the purification lemma (\Cref{lem:purification}) also holds for all Borel type distributions.

\paragraph{Necessity.}\hspace{-2mm}We now prove the necessity direction. Suppose $\mathcal{Y}$ is a profit-optimal menu under every Borel type distribution. As in the proof of \Cref{thm:main}, let 
\[\mathcal{A} :=  \Big\{a\in \Delta(\mathcal{Y}): a \in \Delta\big(\{y_{j-1}, y_j\}\big) \text{ for some $y_j \in \mathcal{Y}$}\Big\}\]
which is a collection of lotteries that can be totally ordered by stochastic dominance with respect to the order on $\mathcal{Y}$. Moreover, by \eqref{eq:incremental}, we also know that every incremental demand curve is strictly decreasing. Thus, to show that $\mathcal{Y}$ is a generalized frontier, it suffices to show that every element $x \in \mathcal{X}$ is covered by $\mathcal{A}$. 

Suppose for contradiction that there exists some $x \in \mathcal{X}$ that is not covered by $\mathcal{A}$. For any $r \in [0, m]$, write $r = i - 1 + \alpha$ for some $i \in \{1, \dots, m\}$ and $\alpha \in [0,1]$, and define
\[
a_r := (1-\alpha) \delta_{y_{i-1}} + \alpha \delta_{y_i} \,.
\]
Write $W(r, t):= v(a_r, t)$. Note that for each $t$, $W(\,\cdot\,, t)$ is continuous and strictly increasing from $0$ to $v(y_m, t)$. 

Consider the degenerate prior concentrated on type $t$. Clearly, when restricted to $\mathcal{Y}$, the optimal revenue is $v(y_m, t)$. For menu $\mathcal{Y}$ to be optimal for all such type distributions, we must have 
\[v(x, t) \leq v(y_m, t)\]
for all types $t$. Since $v(x, t) \geq 0$, this then implies that for any $t \in \mathcal{T}$, there exists a unique $r(t) \in [0, m]$ such that 
\[W(r(t), t) = v(x, t)\,.\]

Note that if $r(t)$ were nondecreasing in $t$, then $x$ would be covered by $\mathcal{A}$. Indeed, at any type $t_0$, the lottery $a_{r(t_0)}$ matches the value of $x$ for type $t_0$ by construction. For any $t < t_0$, we have 
\[W(r(t_0), t) \geq W(r(t), t) = v(x, t)\]
and for any $t > t_0$, we have 
\[W(r(t_0), t) \leq W(r(t), t) = v(x, t)\,.\]
Thus, $W(r(t_0), \,\cdot\,)$ weakly single-crosses $v(x, \,\cdot\,)$, with a crossing type being $t_0$. Since this holds for all $t_0 \in \mathcal{T}$, $x$ would be covered by $\mathcal{A}$ if $r(t)$ were nondecreasing in $t$.  

Therefore, since $x$ is not covered by $\mathcal{A}$, there must exist some $t_L < t_H$ such that 
\[r_L := r(t_L) > r(t_H) =: r_H\,.\]
Choose $i$ such that $r_L \in (i-1, i]$. Consider the binary prior $\mu$ that assigns probability 
\[\pi:= \rho_i(t_L, t_H) = \frac{d_i(t_L)}{d_i(t_H)} \in (0, 1)\]
to type $t_H$ and probability $1 - \pi$ to type $t_L$. 

We first calculate the maximal revenue under this prior $\mu$ and menu $\mathcal{Y}$. Since strict increasing differences hold on $\mathcal{Y}$, this follows by standard one-dimensional characterization. Indeed, the high type $t_H$ must be allocated $y_m$ and the downward IC constraint from $t_H$ to $t_L$ must be binding. Thus, the maximal revenue is given by 
\begin{align*}
    \Pi_\mathcal{Y}(\mu) &= \max_{0 \leq k \leq m} \Big( \pi \cdot  \big(v(y_m, t_H)  - v(y_k, t_H) +  v(y_k, t_L)\big)  +  (1 - \pi) \cdot v(y_k, t_L) \Big)\\
    &= \pi \cdot v(y_m, t_H) + \max_{0 \leq k \leq m}\Big(v(y_k, t_L) - \pi \cdot v(y_k, t_H) \Big)\,. \tag{A.5} \label{eq:1dopt}
\end{align*}
Now, writing $\phi_k = v(y_k, t_L) - \pi \cdot v(y_k, t_H)$, we have 
\begin{align*}
    \phi_k - \phi_{k-1} &= \big( v(y_k, t_L) - \pi \cdot v(y_k, t_H) \big) - \big( v(y_{k-1}, t_L) - \pi \cdot v(y_{k-1}, t_H)\big) \\
    &= d_k(t_L) - \pi \cdot d_k(t_H) = d_k(t_H) \cdot \big[\rho_k(t_L, t_H) - \pi\big]\,.
\end{align*}
Since $d_k(t_H) > 0$, the sign of $\phi_k - \phi_{k-1}$ is the same as the sign of $\rho_k(t_L, t_H) - \pi = \rho_k(t_L, t_H) - \rho_i(t_L, t_H)$, which is strictly positive if $k < i$, zero if $k = i$, and strictly negative if $k > i$ since 
\[\rho_m(t_L, t_H) < \cdots < \rho_1(t_L, t_H)\,.\]
Consequently, the optimizers to \eqref{eq:1dopt} are given by $k \in \{i, i-1\}$, which implies that assigning any lottery supported on $\{y_i, y_{i-1}\}$ to type $t_L$ would give the same maximal revenue. Now, since $r_L \in (i -1, i]$ by construction, this implies that 
\[\Pi_\Y(\mu) = \pi \cdot v(y_m, t_H) + W(r_L, t_L) - \pi \cdot W(r_L, t_H)\,.\]

Now, consider instead assigning $x$ to type $t_L$ and $y_m$ to type $t_H$. This allocation is implementable because 
\[v(x, t_H) - v(x, t_L) = W(r_H, t_H) - W(r_L, t_L) < W(r_L, t_H) - W(r_L, t_L) \leq v(y_m, t_H) - v(y_m, t_L)\,,\]
where the strict inequality follows from $r_H < r_L$ by construction, and the weak inequality follows because $v(y_m, t) - W(r_L, t) = \E_{y \sim a_{r_L}}[v(y_m, t) - v(y, t)]$ is nondecreasing in $t$. Given the above increasing differences, the revenue-maximizing transfers for this allocation are:
\[p_L := v(x, t_L) \text{ and } p_H := v(x, t_L) + v(y_m, t_H) - v(x, t_H)\,.\]
In particular, the revenue under this mechanism is 
\[\Pi_{\{x, y_m\}}(\mu) = \pi \cdot v(y_m, t_H) + v(x, t_L) - \pi \cdot v(x, t_H) = \pi \cdot v(y_m, t_H) + W(r_L, t_L) - \pi \cdot W(r_H, t_H)\,.\]
Therefore, we have 
\[\Pi_{\{x, y_m\}}(\mu) - \Pi_\Y(\mu)  =  \pi \cdot \Big[ W(r_L, t_H)  - W(r_H, t_H) \Big] > 0\,,\]
where the strict inequality uses again that $r_L > r_H$ by construction, contradicting that $\mathcal{Y}$ is an optimal menu for profit maximization under type distribution $\mu$, as desired. 

\subsection{Proof of \Cref{cor:pure}}
Note that for any mechanism $(a, p): \mathbf{v} \mapsto (a(\mathbf{v}), p(\mathbf{v}))$, the objective is given by 
\begin{align*}
 \E\Big[\lambda\big(a(\mathbf{v}) \cdot \mathbf{v} - p(\mathbf{v}) \big) + p(\mathbf{v}) \Big] &= \E\Big[\E[\lambda\mid \mathbf{v}]\big(a(\mathbf{v}) \cdot \mathbf{v} - p(\mathbf{v}) \big) + p(\mathbf{v}) \Big]   \\
 & = \E\Big[\E[\lambda \mid v^{\overline{b}}]\big(a(\mathbf{v}) \cdot \mathbf{v} - p(\mathbf{v}) \big) + p(\mathbf{v}) \Big] \,.
\end{align*}
Let $t$ be a random variable following the same distribution as $v^{\overline{b}}$. Also write $\lambda(t) := \E[\lambda \mid v^{\overline{b}} = t]$ which is a continuous nonincreasing function by assumption. By Strassen's theorem (see Lemma 1 in \citealt{yang2022costly}), under the stochastic ratio-monotonicity, we can write  
\[\mathbf{v} \eqid \Big(t, \big(v(b, t; \varepsilon)\big)_{b \neq \overline{b}}\Big)\,,\]
for some independent random variable $\varepsilon$ and jointly measurable functions $v(b, t; \varepsilon)$ such that $v(b, t; \varepsilon)/v(\overline{b}, t; \varepsilon)$ is nondecreasing in $t$. Now, for any $\varepsilon$, in the relaxed problem where the principal observes the realized $\varepsilon$, by \Cref{prop:nesting}, pure bundling is optimal among all stochastic mechanisms (formally, \Cref{prop:nesting} assumes differentiability but ratio-monotonicity suffices by the proof of \Cref{thm:main2}). Moreover, for any $\varepsilon$, the optimal price is determined by maximizing the objective 
\[\E[\lambda(t) \big(x(t) t - p(t)\big) + p(t)]\,,\]
which does not depend on $\varepsilon$ (since $\varepsilon$ is independent of $t$). Thus, there exists  a single pure-bundling mechanism that does not depend on $\varepsilon$ and solves the principal's relaxed problem when observing $\varepsilon$. This mechanism is implementable even without observing $\varepsilon$, and hence it must be optimal in the original problem, completing the proof. 

It remains to prove the comparative statics result. Given the first part, note that the optimal mechanism is always pure bundling with a single posted price. Thus, the designer solves the following problem: 
\[\max_{p}\Big\{\Big(\E[\lambda(t)(t-p) \mid t \geq p] + p\Big) \cdot (1 - F(p))\Big\}\,,\]
where $F$ is the distribution of $t$ (which follows the same distribution as the grand bundle value). Writing this problem in the quantile space $q := F(t)$, we have 
\[\max_{k \in [0, 1]}\Big\{\Big(\E\big[\lambda(q)\big(F^{-1}(q)-F^{-1}(k)\big) \mid q \geq k\big] + F^{-1}(k)\Big) \cdot (1 - k)\Big\}\,.\]
Letting $\mathcal{W}(\lambda, k)$ denote the above objective, for any $\widetilde{\lambda}, \lambda$ where $\lambda$ is weakly majorized by $\widetilde{\lambda}$, we claim that $\mathcal{W}$ has decreasing differences: 
\[\mathcal{W}(\widetilde{\lambda}, k) - \mathcal{W}(\lambda, k)\]
is nonincreasing in $k$. To see this, using integration by parts, we have 
\begin{align*}
    \mathcal{W}(\widetilde{\lambda}, k) - \mathcal{W}(\lambda, k) &= \E\Big[\big(\widetilde{\lambda}(q) - \lambda(q)\big) \cdot \big(F^{-1}(q)-F^{-1}(k)\big) \mid q \geq k\Big] \cdot (1 - k)\\
    &= \int_{k}^1 \big(\widetilde{\lambda}(q) - \lambda(q)\big) \cdot \big(F^{-1}(q)-F^{-1}(k)\big) \d q \\
    &= \Big[\int_{k}^q \big(\widetilde{\lambda}(s) - \lambda(s)\big) \d s\Big] \cdot \big(F^{-1}(q) - F^{-1}(k)\big) \Bigm|^1_{k} -\int_{k}^1 \Big[\int_{k}^q \big(\widetilde{\lambda}(s) - \lambda(s)\big) \d s\Big] \d F^{-1}(q)\\
    &= \int_{k}^1 \Big[\int_{k}^1 \big(\widetilde{\lambda}(s) - \lambda(s)\big) \d s \Big] \d F^{-1}(q) -\int_{k}^1 \Big[\int_{k}^q \big(\widetilde{\lambda}(s) - \lambda(s)\big) \d s\Big] \d F^{-1}(q) \\
    &= \int_{k}^1 \Big[\int_{q}^1 \big(\widetilde{\lambda}(s) - \lambda(s)\big) \d s \Big] \d F^{-1}(q)\,.
\end{align*}
Since by assumption $\int_{q}^1 \big(\widetilde{\lambda}(s) - \lambda(s)\big) \d s \geq 0$ for all $q \in [0, 1]$, it follows immediately that the above is weakly decreasing in $k$, as desired. Since $\mathcal{W}$ has decreasing differences, the optimizer sets $K^\star(\lambda)$ satisfy monotone comparative statics by the Milgrom--Shannon theorem---in particular, the optimizer sets $K^\star(\lambda)$ are ordered by the strong set order as $\lambda$ moves to $\widetilde{\lambda}$. Since the optimal prices are given by $p^\star = F^{-1}(k^\star)$ for any optimal quantile threshold $k^\star$ and $F^{-1}$ is nondecreasing, it follows immediately that the optimal prices are also ordered by the strong set order as $\lambda$ moves to $\widetilde{\lambda}$.

\section{Approximation to Positively Correlated Environments}\label{app:multidim}

In this appendix, we show that our results are locally robust to multidimensional types---indeed, we show that the one-dimensional-type model we study approximates general screening environments with preferences that are sufficiently positively correlated. Throughout the appendix, we assume a finite $\X$. For example, in the optimal bundling context, $\X = 2^{G}$ where $G$ is the set of goods. 

Without loss of generality, we can embed the type space in $\R^{|\X|}$, where $\mathbf{v} = (v^1,\dots, v^{|\X|})$ describes the value of the agent for each allocation in $\X$. Let $\gamma \in \Delta(\R^{|\X|})$ be the distribution of $\mathbf{v}$ supported on some compact product set $\mathcal{K}$. A mechanism $(x, p)$ is defined as 
\[(x, p):  \mathcal{K} \rightarrow \X \times \R\]
that satisfies the IC and IR constraints, and a stochastic mechanism $(a, p)$ is defined as 
\[(a, p):  \mathcal{K} \rightarrow \Delta(\X) \times \R\]
that satisfies the IC and IR constraints. The designer evaluates any stochastic mechanism using some objective within the same class of objectives described in the main model: 
\[\E\big[\lambda(\mathbf{v}) \big(\mathbf{v} \cdot a(\mathbf{v}) - p(\mathbf{v})\big)\big] + \E\big[p(\mathbf{v})\big]\,,\]
where $\E[\lambda(\mathbf{v})] \leq 1$ and $\lambda(\mathbf{v})$ is continuous and nonincreasing in $\mathbf{v}$.

Importantly, note that all the screening-frontier concepts defined in the main text, i.e., the surplus-elasticity frontier, the strong frontier, and the generalized frontier, are well-defined for any joint distribution $\gamma$. Indeed, these concepts only depend on the \textit{\textbf{sold-alone demand curves}}---and hence only the \textit{\textbf{marginal distributions}} of $\gamma$. We now show that these frontiers are approximately optimal for any distribution $\gamma$ that is sufficiently positively correlated in a suitably defined sense. Suppose that $\gamma$ has bounded continuous marginal distributions. For any such joint distribution $\gamma$, there exists a unique \textit{\textbf{comonotonic distribution}} $\gamma^{\text{mon}} \in \Delta(\mathcal{K})$ that shares the same marginal distributions with $\gamma$ but is maximally positively correlated, given by the \textit{\textbf{Fréchet--Hoeffding coupling}}. To describe the sense in which our model is locally robust to multidimensional types, we apply the misspecification framework of \citet{madarasz2017sellers}. For two distributions $\gamma_1, \gamma_2 \in \Delta(\mathcal{K})$, $\gamma_1$ and $\gamma_2$ are said to be \textit{\textbf{$\delta$-close}} if $\mathcal{K}$ can be partitioned into disjoint measurable sets $S_1,\dots, S_k$ such that $||v- v'||_{\infty} < \delta$ for any $v, v'$ in the same cell $S_j$ and $\gamma_1(S_j) = \gamma_2(S_j)$ for each $S_j$ (\citealt{madarasz2017sellers}). We say that a distribution $\gamma$ is \textbf{\textit{within-$\delta$-positively-correlated}} if it is $\delta$-close to the comonotonic distribution $\gamma^{\text{mon}}$.

\begin{theorem}\label{thm:positive-correlation}
Consider any collection of sold-alone demand curves $\{P(x, \,\cdot\,)\}_{x \in \X}$ where each demand curve $P(x, \,\cdot\,)$ is strictly decreasing and continuously differentiable (for $x\neq\varnothing$). Suppose that $\X^\star$ is a surplus-elasticity frontier given $\{P(x, \,\cdot\,)\}_{x \in \X}$.  Then, for any $\epsilon > 0$, there exists $\delta > 0$ such that,  for any within-$\delta$-positively-correlated distribution $\gamma$ whose sold-alone demand curves coincide with $\{P(x, \,\cdot\,)\}_{x \in \X}$, the menu $\X^\star$ yields an objective value within $\varepsilon$ of the optimal objective value. Moreover, the same holds if $\X^\star$ is a generalized frontier, and the same holds among all stochastic mechanisms if $\X^\star$ is a strong frontier. 
\end{theorem}
\begin{proof}
Let $\mathcal{T} = [0, 1]$ and $t$ be drawn from a uniform distribution on $[0, 1]$. For all $x$, let  
\[v(x, t) := P\big(x, 1-t\big)\,.\]
It follows immediately that the sold-alone demand curves under $\Big(\big\{v(x, t)\big\}_{x \in \X}, U[0, 1] \Big)$ are given by $\big\{P(x, \,\cdot\,)\big\}_{x \in \X}$. Moreover, $v(x, t)$ is strictly increasing in $t$ for all $x \neq \emptyset$ and is continuously differentiable in $t$. Since $\X$ is finite, this implies that $v_t(x, t)$ is a continuous function on $\X \times \mathcal{T}$. Moreover, let 
\[\widetilde{\lambda}(t) := \lambda\Big(\big(P(x, 1-t)\big)_{x \in \X}\Big)\,.\]
Then $\widetilde{\lambda}(t)$ is a continuous and nonincreasing function in $t$ since $\lambda(\mathbf{v})$ is continuous and nonincreasing in $\mathbf{v}$. Together, $\Big(\big\{v(x, t)\big\}_{x \in \X}, \,\widetilde{\lambda}(t),\, U[0, 1] \Big)$ then satisfies the assumptions of \Cref{sec:main}.\footnote{Note that we can without loss of generality consider $\E[\widetilde{\lambda}(t)] \leq 1$, since otherwise there exists no within-$\delta$-positively-correlated $\gamma$ such that $\E[\lambda(\mathbf{v})] \leq 1$ for $\delta$ small enough and the theorem holds vacuously.} Therefore, by \Cref{thm:main}, if $\X^\star$ is a surplus-elasticity frontier given $\big\{P(x, \,\cdot\,)\big\}_{x \in \X}$, then $\X^\star$ is optimal under $\Big(\big\{v(x, t)\big\}_{x \in \X},\,\widetilde{\lambda}(t),\, U[0, 1] \Big)$. 

Since the demand curves are bounded, without loss of generality, we restrict the type space to the compact set $\mathcal{K} \subset \R^{|\X|}$ containing the supports of all distributions under consideration. Moreover, note that, for any type distribution $\gamma$, it is without loss of optimality to restrict attention to IC and IR mechanisms with the agent's indirect utility function satisfying 
\[\inf_{\mathbf{v}}u(\mathbf{v}) = 0\,.\]
Indeed, if the above is not satisfied, then uniformly shifting the transfers for all types would result in  a weak improvement on the objective since $\E[\lambda(\mathbf{v})] \leq 1$. Now, note that for any IC and IR mechanism, the indirect utility function can be written as 
\[u(\mathbf{v}) = \sup_{x \in \X}\big(v^{x} - p(x) \big)_+\,,\]
where $p(x)$ is the tariff schedule implementing the original mechanism (by the taxation principle), and hence $u(\mathbf{v})$ is $1$-Lipschitz with respect to the sup norm. Since the type space is a compact subset of $\R^{|\X|}$, it follows that $u(\mathbf{v})$ is uniformly bounded by some constant $M$. This also implies that the family of functions
$\mathbf{v} \mapsto \lambda(\mathbf{v})u(\mathbf{v})$
has a common modulus of continuity, denoted by $\omega$,
that is independent of the mechanism. Indeed, for any $\mathbf{v}, \mathbf{v}'$ in the compact support with $||\mathbf{v} - \mathbf{v}'||_{\infty} \leq \delta$, we have 
\[\big|\lambda(\mathbf{v})u(\mathbf{v}) - \lambda(\mathbf{v}')u(\mathbf{v}')\big| \leq |u(\mathbf{v})| \cdot |\lambda(\mathbf{v}) - \lambda(\mathbf{v}')| + |\lambda(\mathbf{v}')| \cdot |u(\mathbf{v}) - u(\mathbf{v}')|\leq M \cdot \omega_{\lambda}(\delta) + ||\lambda||_{\infty} \cdot ||\mathbf{v} - \mathbf{v}'||_{\infty}\,,\]
where $\omega_{\lambda}$ is the modulus of continuity of $\lambda(\,\cdot\,)$, which exists since $\lambda(\,\cdot\,)$ is continuous on a compact set and hence uniformly continuous by the Heine--Cantor theorem. By inspection, the above gives a modulus of continuity for $\lambda(\mathbf{v}) \cdot u(\mathbf{v})$ that is independent of the mechanism, as desired. Finally, the last auxiliary observation we make is that for any two $\delta$-close distributions $\gamma, \gamma'$, and any uniformly continuous function $h$ with a modulus of continuity $\omega_h$, we have
\[\Big|\E_{\gamma}[h(\mathbf{v})] - \E_{\gamma'}[h(\mathbf{v})]\Big| \leq \omega_{h}(\delta)\,,\]
because for each cell $S_j$ where $\gamma(S_j) = \gamma'(S_j)$ in the definition of $\delta$-closeness, we have $||\mathbf{v} - \mathbf{v}'||_\infty < \delta$ for all $\mathbf{v}, \mathbf{v}' \in S_j$, and hence for each cell $S_j$, 
\[\Big|\E_{\gamma}[h(\mathbf{v}) \mid \mathbf{v} \in S_j] - \E_{\gamma'}[h(\mathbf{v}) \mid \mathbf{v} \in S_j]\Big| \leq \sup_{\mathbf{v}, \mathbf{v'} \in S_j} |h(\mathbf{v}) - h(\mathbf{v}')| \leq \omega_h(\delta)\,.\]

Fix any $\varepsilon > 0$. Let $\delta_1 > 0$ be given by \Cref{lem:local} against $\varepsilon / 4 > 0$. Let $0 < \delta \leq \delta_1$ be such that $\omega(\delta) < \varepsilon / 4$. Fix any distribution $\gamma$ that is within-$\delta$-positively-correlated and has sold-alone demand curves coinciding with $\big\{P(x, \,\cdot\,)\big\}_{x \in \X}$. Let 
\[V(\X, \gamma)\]
denote the optimal objective value given distribution $\gamma$.

First, we claim that 
\[V(\X, \gamma^{\text{mon}}) \geq V(\X, \gamma) -\varepsilon/2\,.\]
Indeed, fix any mechanism $(x, p)$ (which, as argued, can be assumed to have $\inf_{\mathbf{v}}u(\mathbf{v}) = 0$). Let $(\tilde{x}, \tilde{p})$ be given by \Cref{lem:local}. Then, we have 
\[\E_{\gamma^{\text{mon}}}[\tilde{p}(\mathbf{v})] > \E_{\gamma}[p(\mathbf{v})] - \varepsilon/4\,,\]
$\text{Ran}(\tilde{x}) \subseteq \text{Closure}(\text{Ran}(x))$, and $\tilde{u}(\mathbf{v}) \geq u(\mathbf{v})$ for all $\mathbf{v}$. Moreover, since $\tilde{u}(\mathbf{v}) \geq u(\mathbf{v}) \geq 0$ and $\lambda(\mathbf{v}) \geq 0$, by our previous observations, we can write 
\[\E_{\gamma^{\text{mon}}}\big[\lambda(\mathbf{v}) \tilde{u}(\mathbf{v})\big] \geq \E_{\gamma^{\text{mon}}}\big[\lambda(\mathbf{v}) u(\mathbf{v})\big] \geq \E_{\gamma}\big[\lambda(\mathbf{v}) u(\mathbf{v})\big]  -  \omega(\delta) > \E_{\gamma}\big[\lambda(\mathbf{v}) u(\mathbf{v})\big]  -  \varepsilon / 4\,.\]
Combining these two, we obtain 
\[V(\X, \gamma^{\text{mon}}) \geq \E_{\gamma^{\text{mon}}}\big[\lambda(\mathbf{v}) \tilde{u}(\mathbf{v}) + \tilde{p}(\mathbf{v})\big] \geq \E_{\gamma}\big[\lambda(\mathbf{v}) u(\mathbf{v}) + p(\mathbf{v})\big] - \varepsilon / 2\,.\]
Taking sup over the IC and IR mechanisms $(x, p)$ on the right-hand side then gives the desired inequality.

Now, we claim that 
\[V(\X^\star, \gamma)  \geq  V(\X^\star, \gamma^{\text{mon}}) -\varepsilon/2\,.\]
This follows by the identical argument as above except that we use allocation space $\X^\star$ instead of $\X$ and swap the roles of $\gamma$ and $\gamma^{\text{mon}}$---indeed, since $\X^\star$ is a compact subset of $\X$ and $\delta$-closeness is a symmetric relation, the identical argument establishes the above. 

Moreover, note that, by construction, $\big(v(x, t)\big)_{x \in \X}$ with $t \sim U[0, 1]$ results in a comonotonic value distribution sharing the same marginals with $\gamma$. Therefore, its induced value distribution in $\R^{|\X|}$ must coincide with $\gamma^{\text{mon}}$ since the comonotonic coupling is unique given continuous marginals. By construction, this also implies that the joint distribution of $\big(\widetilde{\lambda}(t), \big(v(x, t)\big)_{x \in \X}\big)$ with $t \sim U[0, 1]$ is the same as the joint distribution of $\big(\lambda(\mathbf{v}), \mathbf{v}\big)$ with $\mathbf{v} \sim \gamma^{\text{mon}}$. Therefore, by \Cref{thm:main}, as we have argued, it then follows that $V(\X^\star, \gamma^{\text{mon}}) = V(\X, \gamma^{\text{mon}})$. Finally, combining these together, we obtain 
\[V(\X^\star, \gamma) \geq V(\X^\star, \gamma^{\text{mon}}) -\varepsilon/2  =  V(\X, \gamma^{\text{mon}}) -\varepsilon/2 \geq  V(\X, \gamma) -\varepsilon\,,\]
proving the result. 

Note that the same conclusion holds for any generalized frontier $\X^\star$ by applying \Cref{thm:main3} instead of \Cref{thm:main} in the last step. Now, suppose $\X^\star$ is a strong frontier. The same conclusion holds among all stochastic mechanisms because the above proof holds identically when replacing $\X$ with $\Delta(\X)$ and applying \Cref{thm:main2} instead of \Cref{thm:main} in the last step, while keeping everything else the same---in particular, note that 
\[u(\mathbf{v}) = \sup_{a \in \Delta(\X)}\big(\mathbf{v} \cdot a - p(a) \big)_+\,,\]
continues to be $1$-Lipschitz in the sup norm for any IC and IR stochastic mechanism, and $\{\delta_{x}\}_{x \in \X^\star}$ is a compact subset of $\Delta(\X)$, which itself is compact under the weak-$^*$ topology (metrizable since $\X$ is compact). 
\end{proof}

\end{document}